\documentclass{article}
\usepackage[left=1.25in,top=1.25in,right=1.25in,bottom=1.25in,head=1.25in]{geometry}
\usepackage{amsfonts,amsmath,amssymb,amsthm}
\usepackage{verbatim,float,url,enumerate}
\usepackage{graphicx,subfigure,epsfig,psfrag}
\usepackage{dsfont,bm,color,appendix}
\usepackage{natbib}
\usepackage{pdfpages}
\usepackage{algorithm}
\graphicspath{{figures_oct11/}}

\def\bone{{\bf 1}}
\def\bI{{\bf I}}
\def\bw {{\bf w}}
\def\by {{\bf y}}
\def\bA{{\bf A}}
\def\bD{{\bf D}}

\def\bM{{\bf M}}
\def\bR{{\bf R}}
\def\bU{{\bf U}}
\def\bV{{\bf V}}
\def\bW{{\bf W}}
\def\bX{{\bf X}}
\def\bZ{{\bf Z}}
\newcommand\bbR{\mathbb{R}}

\providecommand{\minimize}{\mathop{\rm minimize}}

\providecommand{\subjectto}{\mathop{\rm subject\;to}}
\newcommand{\norm}[1]{\left\|{#1}\right\|} 
\newcommand{\lone}[1]{\norm{#1}_1} 
\newcommand{\ltwo}[1]{\norm{#1}_2} 
\newcommand{\linf}[1]{\norm{#1}_\infty} 

\newtheorem{theorem}{Theorem}
\newtheorem{lemma}[theorem]{Lemma}

\title{Principal component-guided sparse regression}
\author{ J. Kenneth Tay, Jerome Friedman and Robert Tibshirani  \\
Department of Statistics, and Department of  Biomedical Data Science \\ Stanford University }

\begin{document}
\maketitle

\begin{abstract}
We propose a new method for supervised learning, especially suited to wide data where the number of features is much greater than the number of observations. The method combines the lasso ($\ell_1$) sparsity penalty with a quadratic penalty that shrinks the coefficient vector toward the leading principal components of the feature matrix.  We call the  proposed method  the ``{\em principal components lasso}''  (``pcLasso'').  The method can be especially powerful if the features are pre-assigned to groups (such as cell-pathways, assays or protein interaction networks). In that case, pcLasso shrinks each group-wise component of the solution toward the leading principal components of that group. In the process, it also carries out selection of the feature groups. We provide some theory for this method and illustrate it on a number of simulated and real data examples.
\end{abstract}

\section{Introduction}

We consider the usual linear regression model: given $n$ realizations of $p$ predictors $\bX=\{x_{ij}\}$ for $i=1,2,\ldots,n$ and $j=1,2,\ldots,p$, the response $\by=(y_1,\ldots,y_n)$ is modeled as
\begin{equation}
y_i=\beta_0\bone + \sum_j x_{ij} \beta_j +\epsilon_i,
\end{equation}
with $\epsilon \sim (0,\sigma^2)$. The ordinary least squares (OLS) estimates of $\beta_j$ are obtained by minimizing the residual sum of squares (RSS). There has been much work on regularized estimators that offer an advantage over the OLS estimates, both in terms of accuracy of prediction on future data and interpretation of the fitted model. One focus has been on the {\em lasso} \citep{Ti96}, which minimizes
\begin{equation}
J(\beta_0,\beta)=\frac{1}{2}\|\by-\beta_0\bone-\bX \beta\|_2^2+\lambda \|\beta\|_1,
\end{equation}
where $\beta=(\beta_1, \ldots,\beta_p)^T$, and the tuning parameter $\lambda \geq 0$ controls the sparsity of the final model. This parameter is often selected by cross-validation (CV). The objective function $J(\beta_0,\beta)$ is convex, which means that solutions can be found efficiently even for very large $n$ and $p$, in contrast to combinatorial methods like best subset selection.

The lasso is an alternative to ridge regression \citep{HK70}, which uses an objective function with an  $\ell_2$ penalty:
\begin{equation}
\frac{1}{2}\|\by-\beta_0\bone-\bX \beta\|_2^2+\lambda \|\beta\|_2^2.
\end{equation}
This approach has the disadvantage of producing dense models. However, it also biases the solution vector toward the leading right singular vectors of $\bX$, which can be effective for improving prediction accuracy.

The elastic net \citep{ZH2005} generalizes the lasso, effectively combining ridge regression and the lasso, by adding an $\ell_2$ penalty to the lasso's objective function:
\begin{equation}
\frac{1}{2}\|\by-\beta_0\bone -\bX \beta\|_2^2+\lambda\|\beta\|_1+ \frac{\lambda_2}{2}\|\beta\|_2^2.
\end{equation}

We have found that the bias of ridge regression (and hence the elastic net) toward the leading singular vectors is somewhat mild. Furthermore, if the predictors come in pre-assigned groups--- a focus of this paper--- these methods do not exploit this information.

In this paper, we propose a new method, the {\em principal components lasso} (``pcLasso"), that strongly biases the solution vector toward the leading singular vectors and exploits group structure. We give the full definition of the pcLasso for multiple groups in Section \ref{sec:pcLasso}. For motivation, we describe the simpler single group case in the next section.

\section{Quadratic penalties based on principal components}
 Assume that we have centered the columns of $\bX$, and let $\bX=\bU\bD\bV^T$ be the singular value decomposition (SVD) of $\bX$. Here, $\bD$ is a diagonal matrix with diagonal entries equaling the singular values $d_1 \geq d_2 \ldots \geq d_m > 0$, with $m={\rm rank} (\bX)$. Assume also that $\by$ has been centered, so that we may omit the intercept from the model.

 In generalizing the $\ell_2$ penalty, one can imagine a class of penalty functions of the form
 $$\beta^T \bV \bZ \bV^T \beta,$$
 where $\bZ$ is a diagonal matrix whose diagonal elements are functions of the squared singular values:
 $$Z_{11}= f_1(d_1^2, d_2^2, \ldots, d_m^2),  Z_{22}=f_2(d_1^2, d_2^2, \ldots, d_m^2), \ldots $$
 The ridge penalty chooses $Z_{jj}=1$ for all $j$, leading to the quadratic penalty $\beta^T\beta$. Here we propose a   different choice: we minimize
\begin{eqnarray}
J(\beta)= \frac{1}{2} \ltwo{\by-\bX\beta}^2 + \lambda  \lone{\beta} + \frac{\theta}{2} \beta^T \bV \bD_{d^2_{1}-d^2_j} \bV ^T\beta,
\label{eqn:pcLasso1}
\end{eqnarray}
where  $\bD_{d^2_{1}-d^2_j}$ is an $m \times m$ diagonal matrix with diagonal entries equaling $d_1^2 - d_1^2, d_1^2 - d_2^2,\ldots, d_1^2 - d_m^2$.
We call this the ``pcLasso penalty" and use it more generally in the next section. The value $\theta \geq 0$ is a tuning parameter determining the weight given to the second penalty. We see that this penalty term gives no weight (a ``free ride'') to the component of $\beta$ that aligns with the first right singular vector of $\bX$, i.e, the first principal component (PC). Beyond that, the size of the penalty depends on the gap between the squared singular values $d_j^2$ and $d_1^2$. If we view the quadratic penalty as representing a Bayesian (log-)prior, it puts infinite prior variance on the leading eigenvector.

It is instructive to compare \eqref{eqn:pcLasso1} with $\lambda = 0$ (and with $\frac{\theta}{2}$ replaced with $\theta$) to ridge regression which uses the penalty $\theta \beta^T \beta$. By transforming to principal coordinates, it follows that ridge regression produces the fit
\begin{equation}
\bX\hat\beta_R=\sum_{j=1}^m \frac{d_j^2}{d_j^2+ \theta} u_j  u_j^T\by.
\label{eqn:ridgepc}
\end{equation}

In contrast, pcLasso \eqref{eqn:pcLasso1} gives
\begin{equation}
\bX\hat\beta_L=\sum_{j=1}^m \frac{d_j^2}{d_j^2+ \theta (d_{1}^2-d_j^2)} u_j  u_j^T\by.
\label{eqn:pcLassopc}
\end{equation}

The latter expression corresponds to a more aggressive form of shrinkage toward the leading singular vectors. To see this, Figure \ref{fig:shrinkage} shows the shrinkage factors $\frac{d_j^2}{d_j^2+\theta}$ and $\frac{d_j^2}{d_j^2+ \theta(d_{1}^2-d_j^2)}$  for a $100 \times 20$ design matrix $\bX$ with a strong rank-1 component. In each panel, we have chosen the value $\theta$ (different for each method) to yield the degrees of freedom indicated at the top of the panel. (For each method, the degrees of freedom is equal to the sum of the shrinkage factors.)

\begin{figure}
\centerline{\includegraphics[width=5in]{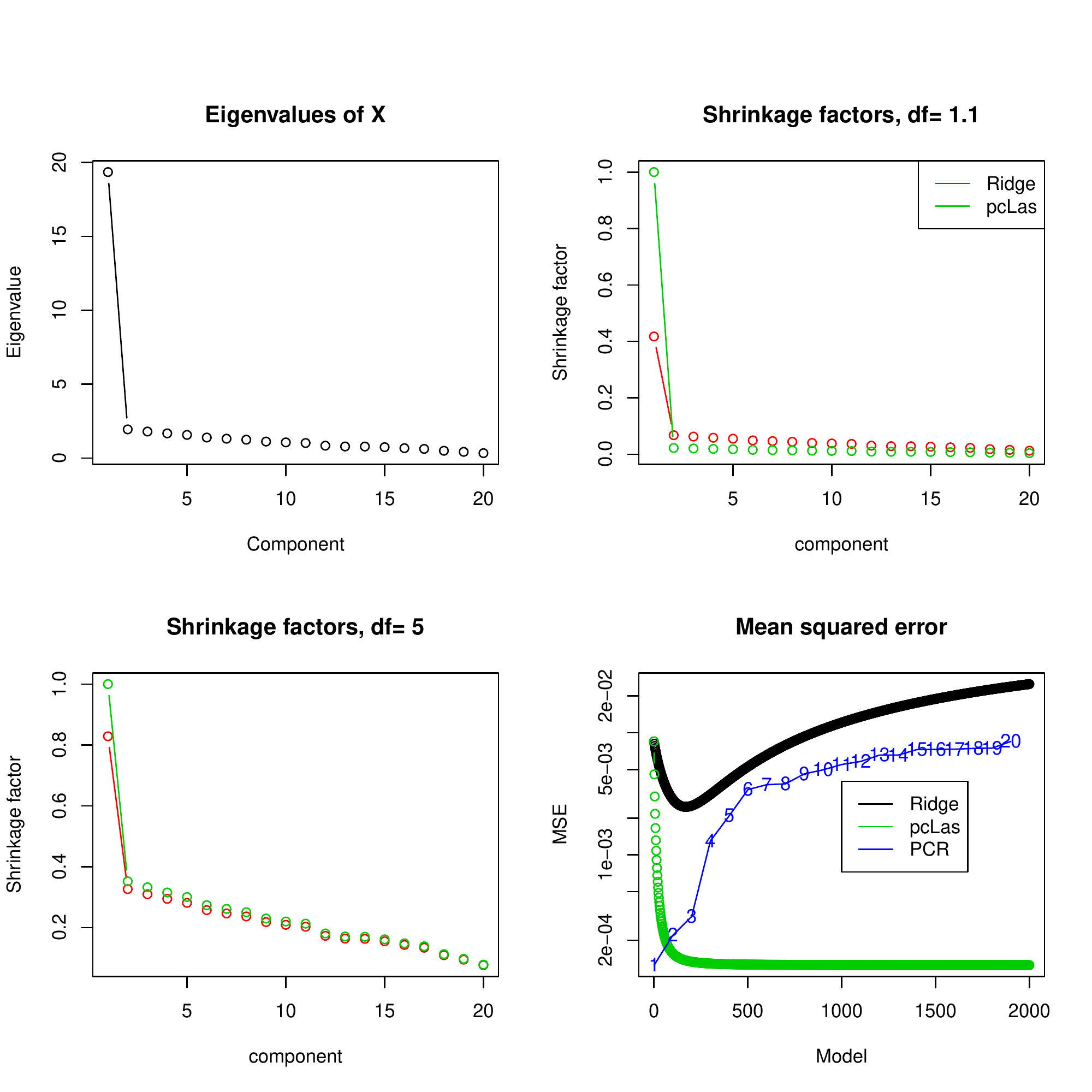}}
\caption[fig:shrinkage]{\em Simulation example: The design matrix $\bX$ has dimensions $100 \times 20$, and has a strong rank-1 component: the eigenvalues of $\bX$ are shown in the top-left panel. Top-right and bottom-left panels show the shrinkage factors from the pcLasso penalty \eqref{eqn:pcLasso1} (green) and ridge regression (red) across the PCs, for different degrees of freedom. For small degrees of freedom, pcLasso \eqref{eqn:pcLasso1} shrinks more aggressively to the top PCs. For larger degrees of freedom, the two methods have similar shrinkage profiles. The bottom right panel shows the mean-squared error (MSE) of ridge regression and pcLasso as the regularization parameter $\theta$ varies along the horizontal axis. Also included is the MSE for PC regression for various ranks. We see that pcLasso achieves about the same MSE as PC regression of rank-1, while ridge regression does not.}
\label{fig:shrinkage}
\end{figure}

We see that when the degrees of freedom is about ``right'' (top-left panel), ridge regression shrinks the top component about twice as much as does pcLasso (top-right panel). This contrast is less stark when the degrees of freedom is  increased to 5 (bottom-left panel). This aggressiveness can help pcLasso improve prediction accuracy when the signal lines up strongly with the top PCs. More importantly, with pre-specified feature groups, the pcLasso penalty in \eqref{eqn:pcLasso1} can be modified in such a way that it shrinks the subvector of $\beta$ in each group toward the leading PCs of that group. This exploits the group structure to give potentially better predictions and simultaneously selects groups of features. We give details in the next section.

The rest of this paper is organized follows. In  Section \ref{sec:pcLasso} we give a full definition of our proposal. We illustrate the effectiveness of pcLasso on some real data examples in Section \ref{sec:realdata}, and give details of our computational approach in Section \ref{sec:computation}. Section \ref{sec:df} derives a formula for the degrees of freedom of the pcLasso fit. In Section \ref{sec:sim} we discuss an extensive simulation study, comparing our proposal to the lasso, elastic net and principal components regression. We study the theoretical properties of pcLasso in Section \ref{sec:theory}, showing that, under certain assumptions, it yields lower $\ell_2$ and prediction error than the lasso. We end with a discussion and ideas for future work. The Appendix contains further details and proofs, as well as a full description of the set-up and results for our simulation study in Section \ref{sec:sim}.

\section{Principal components lasso (``pcLasso")}
\label{sec:pcLasso}
\subsection{Definition}
Suppose that the $p$ predictors are grouped in $K$ non-overlapping groups (we discuss the overlapping groups case later in Section \ref{sec:overlapping-groups}). For example, these groups could be different gene pathways, assays or protein interaction networks. For $k = 1, \ldots, K$, let $\bX_k$ denote the $p_k$ columns of $\bX$ corresponding to group $k$, let $m_k = \text{rank}(\bX_k)$, and let $(\bV_k, d_k)$ denote the right singular vectors and singular values of $\bX_k$. pcLasso minimizes
\begin{eqnarray}
J(\beta)= \frac{1}{2} \ltwo{\by-\bX\beta}^2 + \lambda  \lone{\beta} + \frac{\theta}{2} \sum_k \beta_k^T\Bigl( \bV_k \bD_{d_{k1}^2-d_{kj}^2}\bV_k ^T\Bigr)\beta_k.
\label{eqn:pcLasso}
\end{eqnarray}
Here, $\beta_k$ is the sub-vector of $\beta$ corresponding to group $k$, $d_k=(d_{k1}, \ldots d_{km_k})$ are the singular values of $\bX_k$ is decreasing order, and $\bD_{d_{k1}^2-d_{kj}^2}$ is a diagonal matrix with diagonal entries $d_{k1}^2-d_{kj}^2$ for $j=1,2,\ldots m_k$. Optionally, we might also multiply each penalty term in the sum by a factor $\sqrt{p_k}$, $p_k$ being the group size, as is standard in related methods such as the group lasso \citep{glasso}.

Some observations on solving \eqref{eqn:pcLasso} are in order:
\begin{itemize}
\item The objective function is convex and the non-smooth component is separable. Hence, it can be optimized efficiently by a coordinate descent procedure (see Section \ref{sec:computation} for details).

\item Our approach is to fix a few values of $\theta$ (including zero, corresponding to the lasso), and then optimize the objective \eqref{eqn:pcLasso} over a path of $\lambda$ values. Cross-validation is used to choose $\hat\theta$ and $\hat\lambda$.

\item The parameter  $\theta$  is not unitless, and is difficult to interpret. In our software, instead of asking the user to specify $\theta$, we ask the user to specify a function of $\theta$, denoted by ``rat''. This is the ratio
between the shrinkage factors in \eqref{eqn:pcLassopc} for $k=2$  and $k=1$ (the latter being equal to 1). This ratio is between 0 and 1, with 1 corresponding to $\theta=0$ (the lasso) and lower values corresponding to more shrinkage. In practice, we find that using the values $rat = 0.25, 0.5, 0.75, 0.9, 0.95$ and $1$ covers a wide range of possible models.

\item A potentially costly part of the computation is the SVD of each $\bX_k$. If $\bX_k$ large, for computational ease, we can use an SVD of rank $< \text{rank}(\bX_k)$ as an approximation without much loss of performance.
 \end{itemize}

For some insight into how pcLasso shrinks coefficients, we consider the contours of the penalty in the case of two predictors. Let $\rho$ denote the correlation between these two predictors. In this case, it is easy to show that the two right singular vectors of $\bX$ are $\frac{1}{\sqrt{2}}(1, 1)^T$ and $\frac{1}{\sqrt{2}}(1, -1)^T$, with the former (latter resp.) being the leading singular vector if $\rho > 0$ ($\rho < 0$ resp.). Figure \ref{fig:contours} presents the penalty contours for pcLasso along with that for ridge, lasso and elastic net regression for comparison. (See Appendix \ref{sec:contours} for detailed computations on the pcLasso contours.) We can see that pcLasso imposes a smaller penalty in the direction of the leading singular vector, and that this penalty is smaller when the correlation between the two predictors is stronger. We also see that as $\theta \rightarrow 0$, the pcLasso contours become like those for the lasso, while as $\theta \rightarrow \infty$, they become line segments in the direction of the leading singular vector.
\begin{figure}[ht]
\includegraphics[width=2in]{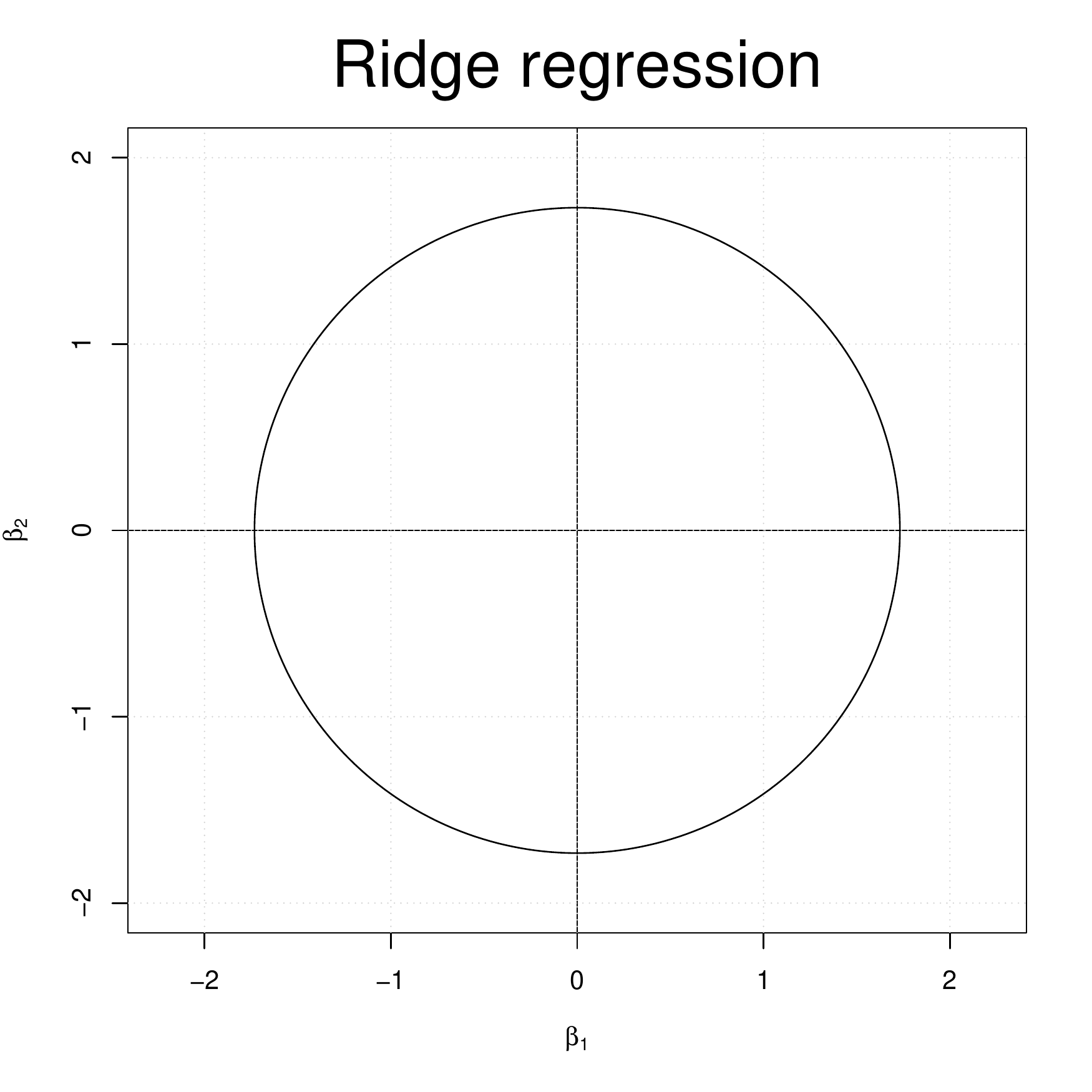} \includegraphics[width=2in]{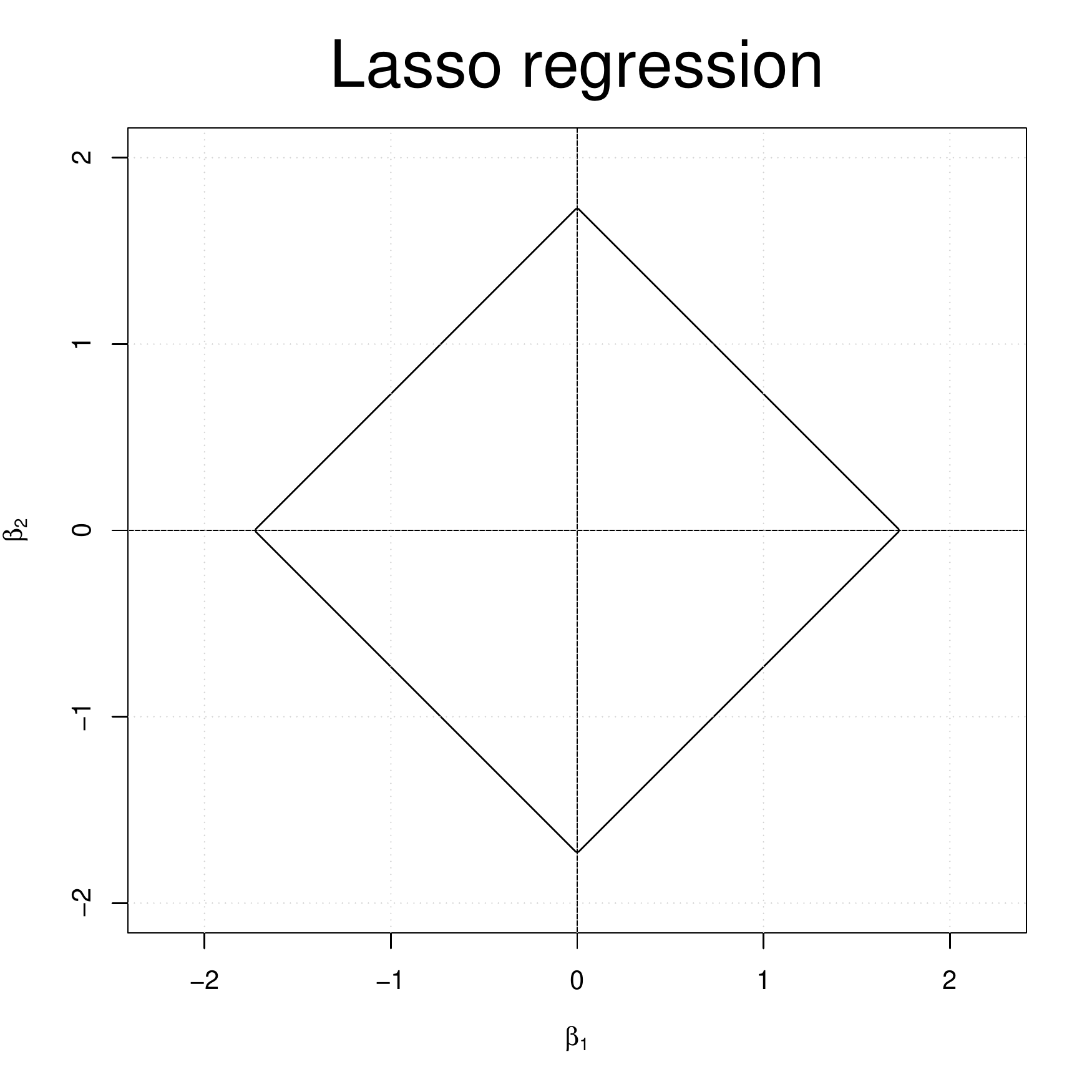} \includegraphics[width=2in]{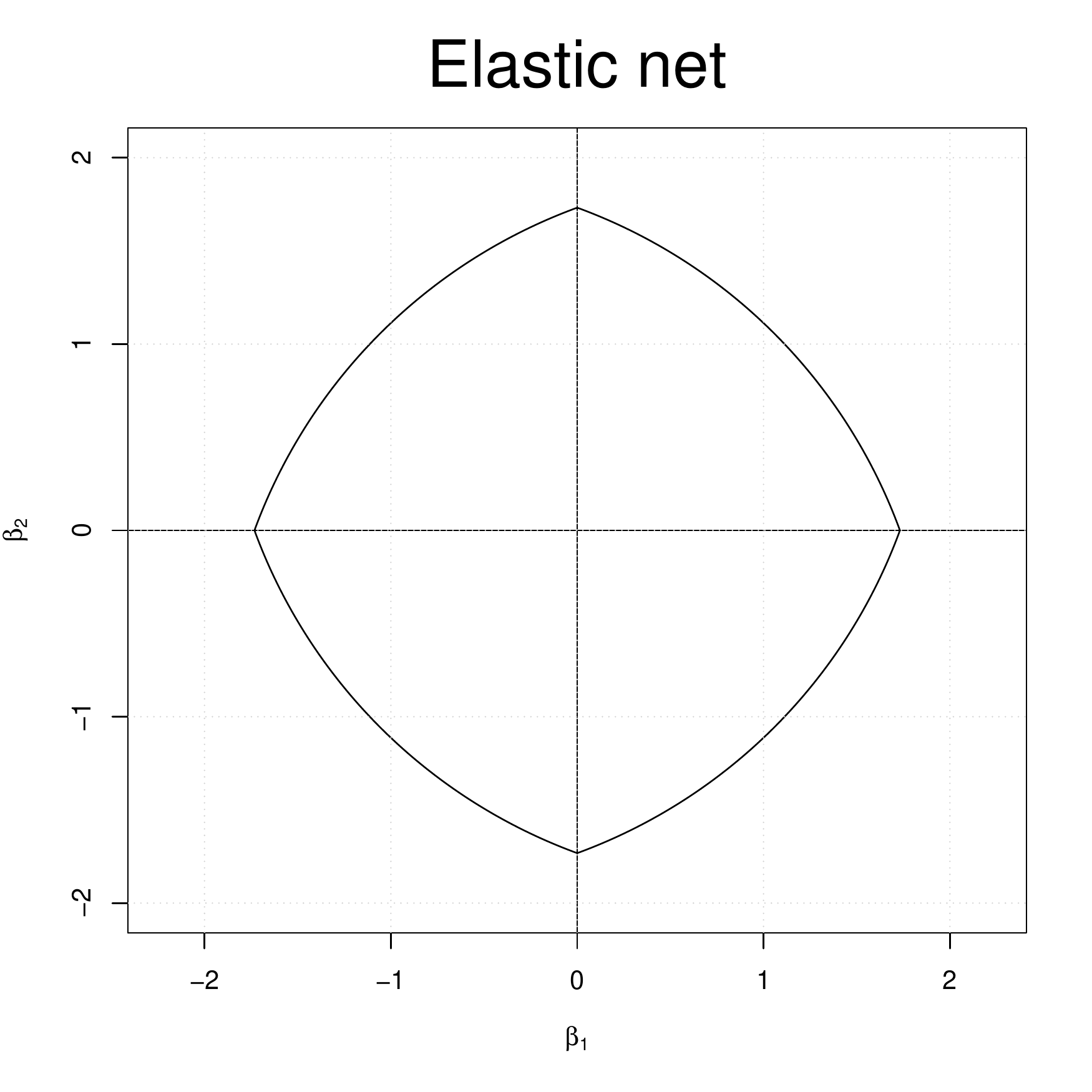}

\includegraphics[width=2in]{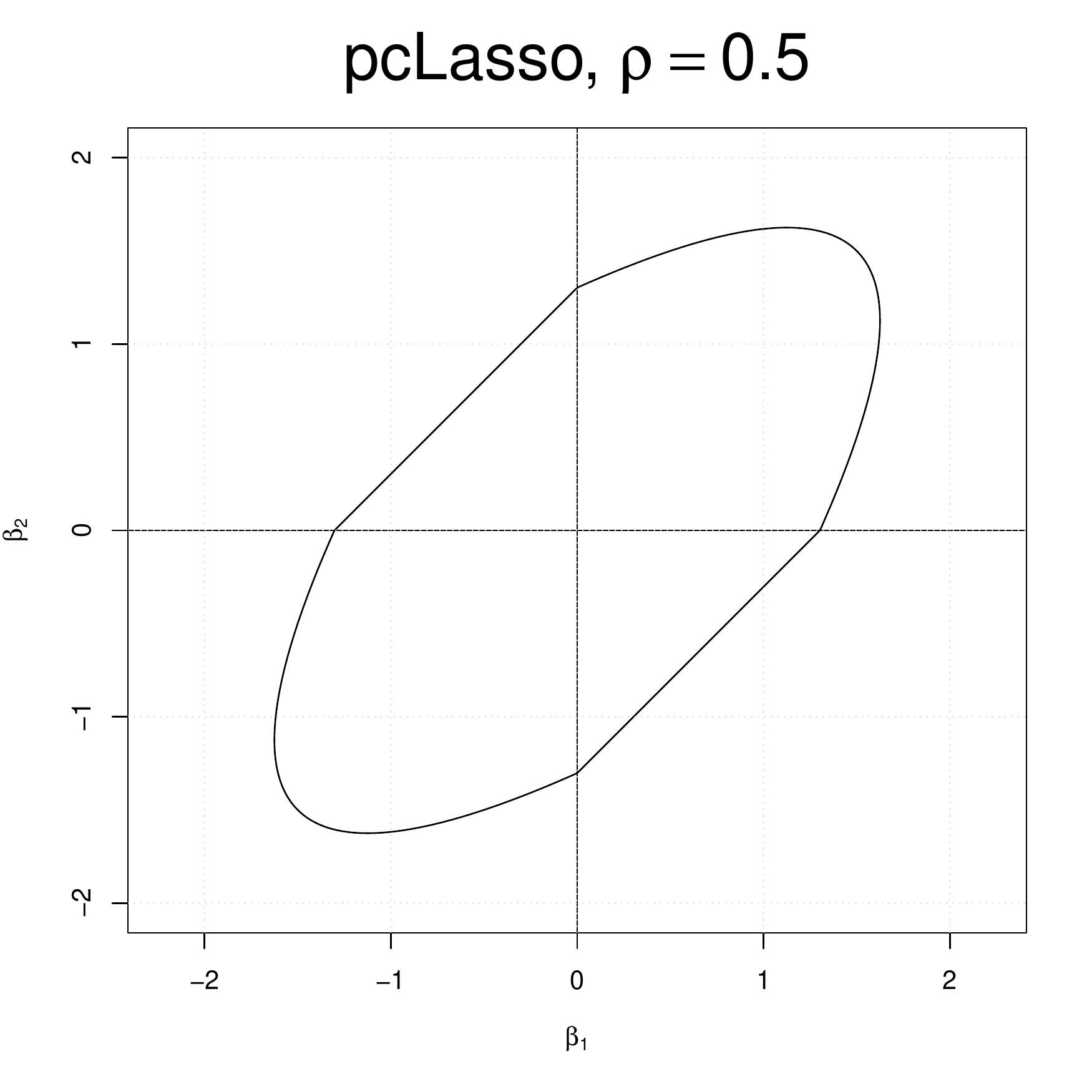} \includegraphics[width=2in]{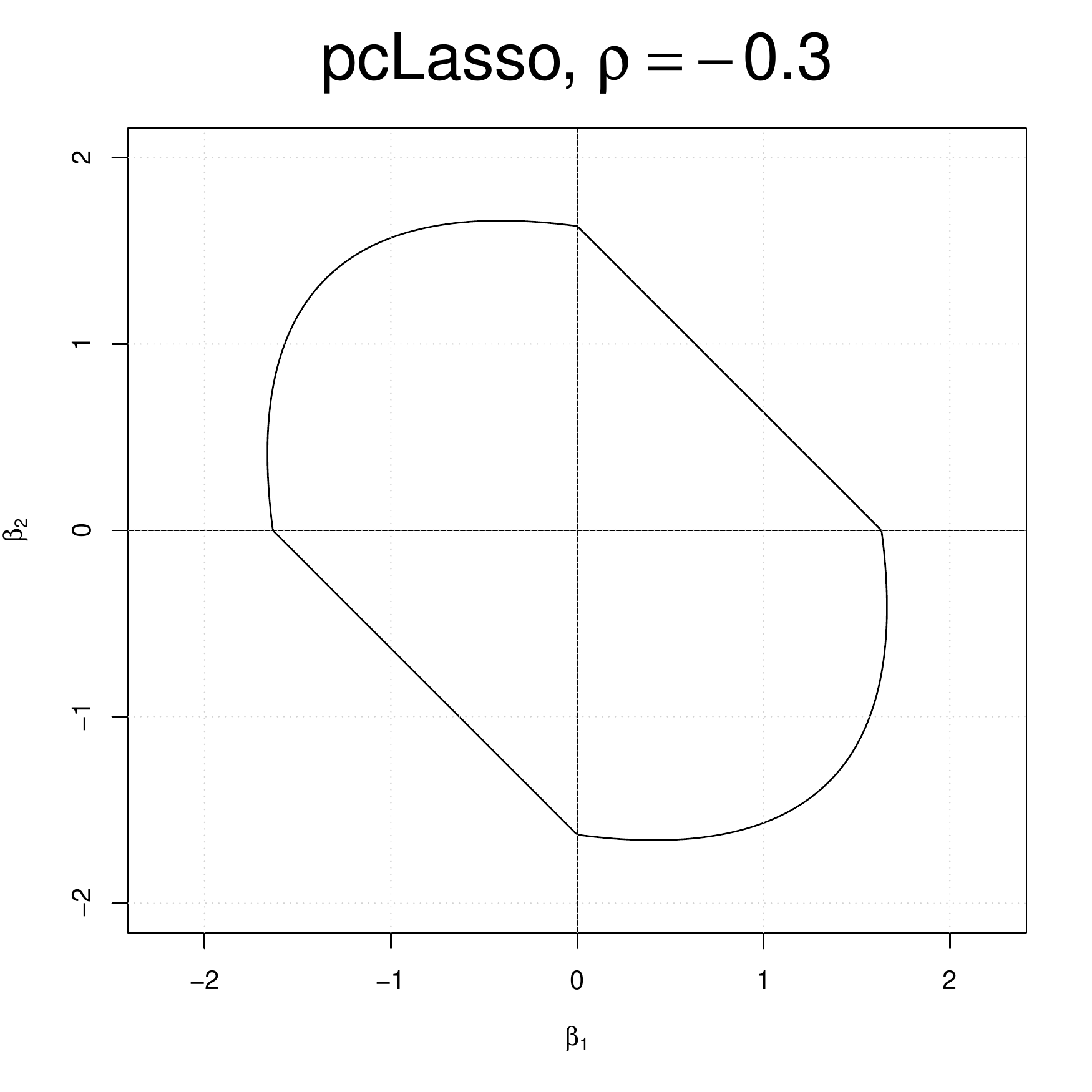} \includegraphics[width=2in]{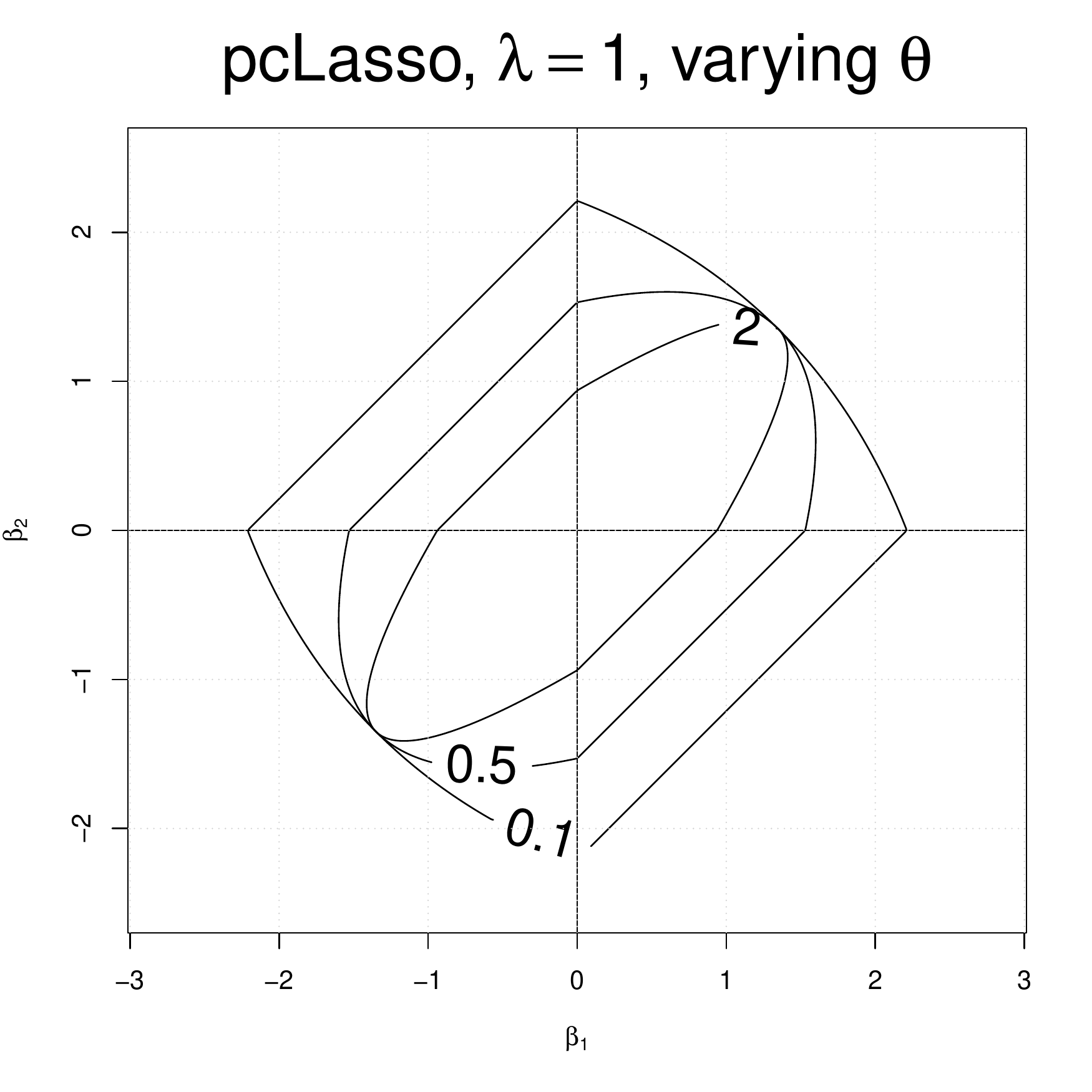}
\caption[fig:contours]{\em Contours for the pcLasso penalty, compared with those for ridge, lasso and elastic net regression. pcLasso penalizes the coefficient vector less in the direction of the leading singular vector. As $\theta \rightarrow 0$, the pcLasso contours become like those for the lasso, while as $\theta \rightarrow \infty$, they become line segments in the direction of the leading singular vector.}
\label{fig:contours}
\end{figure}

Figure \ref{fig:contours3D} shows the penalty contours for a particular set of three predictors, with the ratio $\lambda / \theta$ decreasing as we go from left to right. As this ratio goes from $\infty$ to $0$, the contours move from being the $\ell_1$ ball (as in the lasso) to becoming more elongated in the direction of the first PC.

\begin{figure}[ht]
\includegraphics[width=2in]{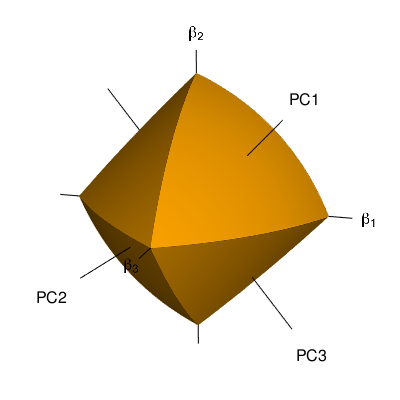} \includegraphics[width=2in]{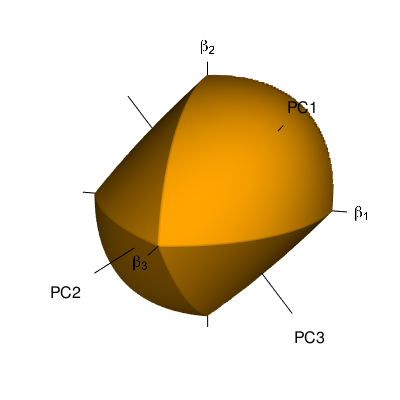} \includegraphics[width=2in]{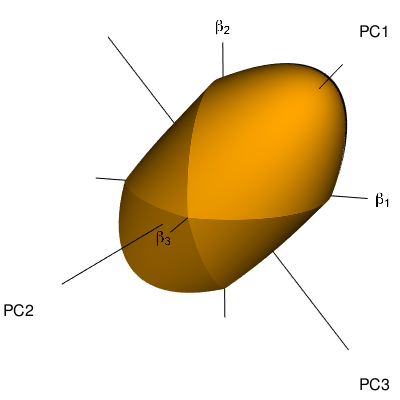}
\caption[fig:contours3D]{\em Contours for the pcLasso penalty for three predictors. For a fixed value of $\theta$, $\lambda$ decreases as we go from left to right. We see that as we give less relative weight to the $\ell_1$ penalty, the penalty contours look less and less like that of the lasso and become more elongated in the direction of the first PC.}
\label{fig:contours3D}
\end{figure}

\subsection{The overlapping groups setting}\label{sec:overlapping-groups}
As in the group lasso \citep{Yuan06modelselection}, there are two approaches one can take to deal with overlapping groups. One approach is to apply the group penalty to the given groups: this however has the undesirable effect of zeroing out a predictor's coefficient if any group to which it belongs is zeroed out. The overlap also causes additional computational challenges as the penalty is no longer separable. A better approach is the ``overlap group lasso'' penalty \citep{JOV2009}, where one replicates columns of $\bX$ to account for multiple memberships, hence creating non-overlapping groups. The model is then fit as if the groups are non-overlapping and the coefficients for each original feature are summed to create the estimated coefficient for that feature. We take this same approach for pcLasso.

\section{Real data examples}
\label{sec:realdata}
\subsection{Blog Feedback dataset}
This dataset is from the UC Irvine collection. Here is the description of the data and task:
\medskip

{\em ``This data originates from blog posts. The raw HTML-documents of the blog posts were crawled and processed.
The prediction task associated with the data is the prediction of the number of comments in the upcoming 24 hours. In order to simulate this situation, we choose a basetime (in the past) and select the blog posts that were published at most 72 hours before the selected base date/time. Then, we calculate all the features of the selected blog posts from the information that was available at the basetime, therefore each instance various ranks, corresponds to a blog post. The target is the number of comments that the blog post received in the next 24 hours relative to the basetime.

In the train data, the basetimes were in the years 2010 and 2011. In the test data the basetimes were in February and March 2012. This simulates the real-world situation in which training data from the past is available to predict events in the future.'' }
\medskip

There are a total of $60,021$ observations; we took a random sample of 1000 observations from training set (years 2010 and 2011),
and 5 days ---2012.03.27 through 2012.03.31--- for the test data (a total of 817 observations). There are a total of 281 attributes (features). We applied pcLasso with no pre-specified groups of features. The left panel of Figure \ref{fig:blog} shows the eigenvalues of the $\bX^T\bX$ matrix, scaled so that the largest eigenvalue is 1. The eigenvalues for a random Gaussian matrix of the same size are also shown for reference. For the Blog Feedback dataset, there is a sharp drop-off in the eigenvalues after the first 3 or 4 eigenvalues, indicating strong correlation among the features.

The right panel shows the test error for the lasso, elastic net, PC regression of various ranks and pcLasso for different values of the ``rat" parameter. The results for elastic net are very close to that for the lasso. We see that pcLasso achieves a somewhat lower test error by shrinking towards the top principal components.

\begin{figure}
\centering
\includegraphics[width=6in]{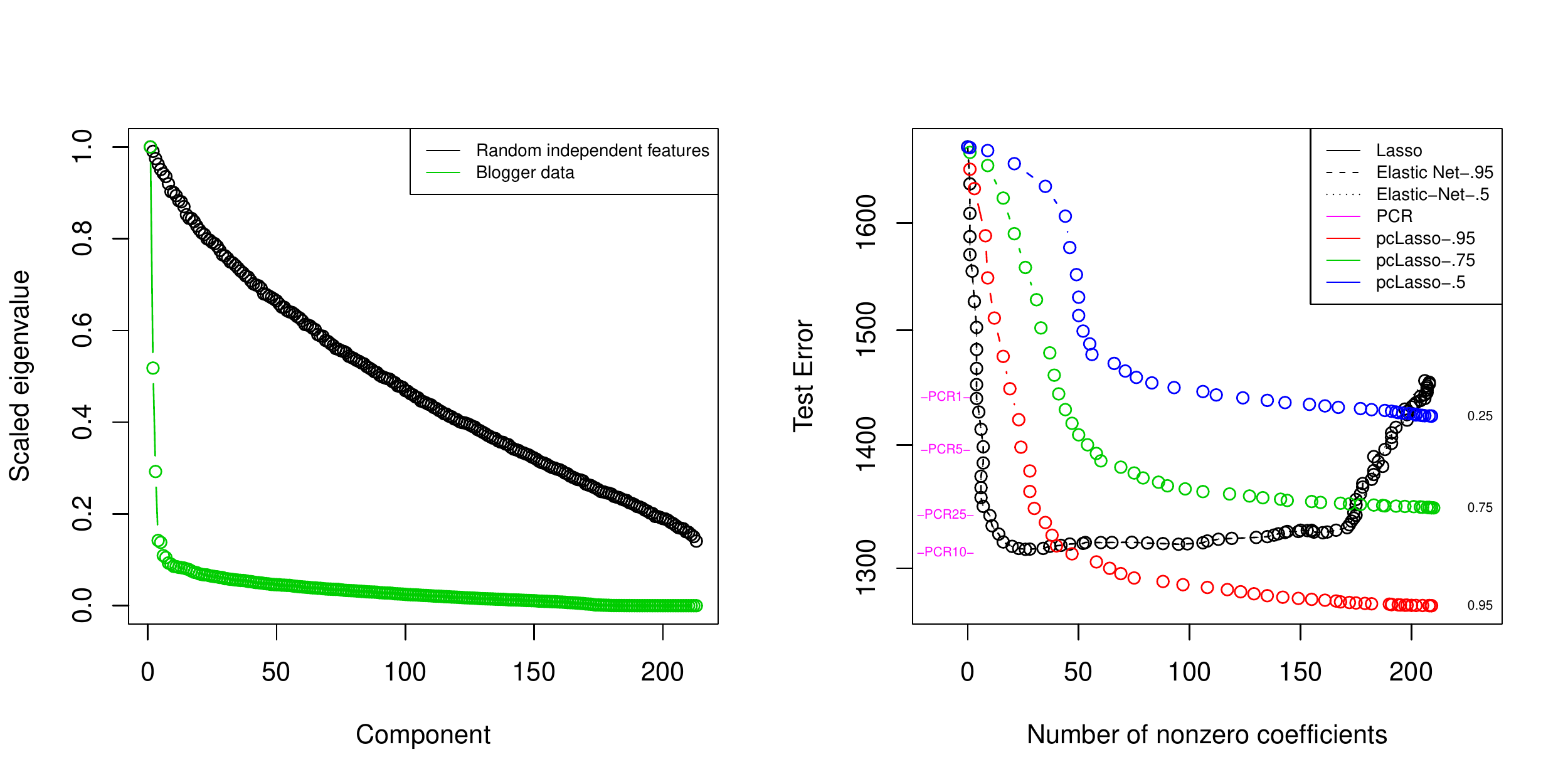}
\caption{\em Results for Blogger Feedback dataset: Left panel shows the eigenvalues of the $\bX^T\bX$ matrix, scaled so that the largest eigenvalue is 1. The eigenvalues for a random Gaussian matrix of the same size are also shown for reference. Right panel shows the test error for different methods. The lasso and elastic net have give similar results on this dataset. The errors for principal components regression (PCR) do not vary across the horizontal axis and so are marked on the left side of the panel.}
\label{fig:blog}
\end{figure}

\subsection{p53 microarray expression data}
Here, we analyze the data from  a gene expression study, as described in \cite{YB2016}, taken from \cite{sub2005}. It  involves the mutational status of the gene p53 in cell lines. The study aims to identify pathways that correlate with the mutational status of p53, which regulates gene expression in response to various signals of cellular stress. The data consists of 50 cell lines, 17 of which are classified as normal and 33 of which carry mutations in the p53 gene. To be consistent with the analysis in \cite{sub2005}, 308 gene sets that have size between 15 and 500 are included in our analysis. These gene sets contain a total of 4301 genes and is available in the R package {\tt grpregOverlap}. When the data is expanded to handle the overlapping groups, it contains a total of $13,237$ columns.

We divided the data into 10 cross-validation (CV) folds and applied the lasso, pcLasso and the group lasso using the R package {\tt grpregOverlap}. (The data was too large for the sparse group lasso package, \texttt{SGL}.) The left panel of Figure \ref{fig:p53fig} plots the number of non-zero pathways (i.e. pathways with at least one zero coefficient) against the number of non-zero coefficients for each method, as the complexity parameter of each method is varied. We see that pcLasso induces some group-level sparsity, although not as strongly as the group lasso. However, the right panel shows that pcLasso achieves higher cross-validation area under the curve (AUC) than the other approaches.

\begin{figure}
\includegraphics[width=3in]{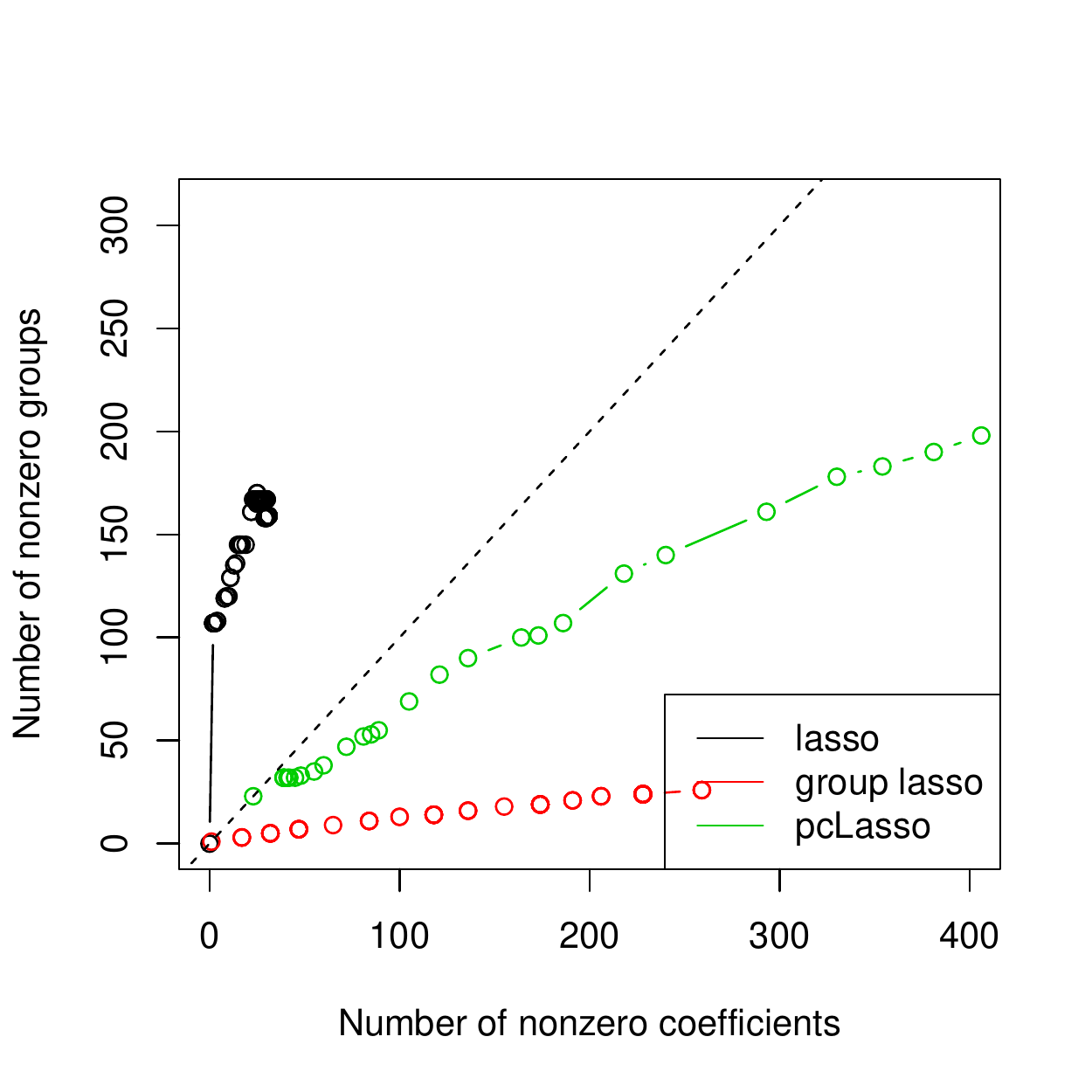}
\includegraphics[width=3in]{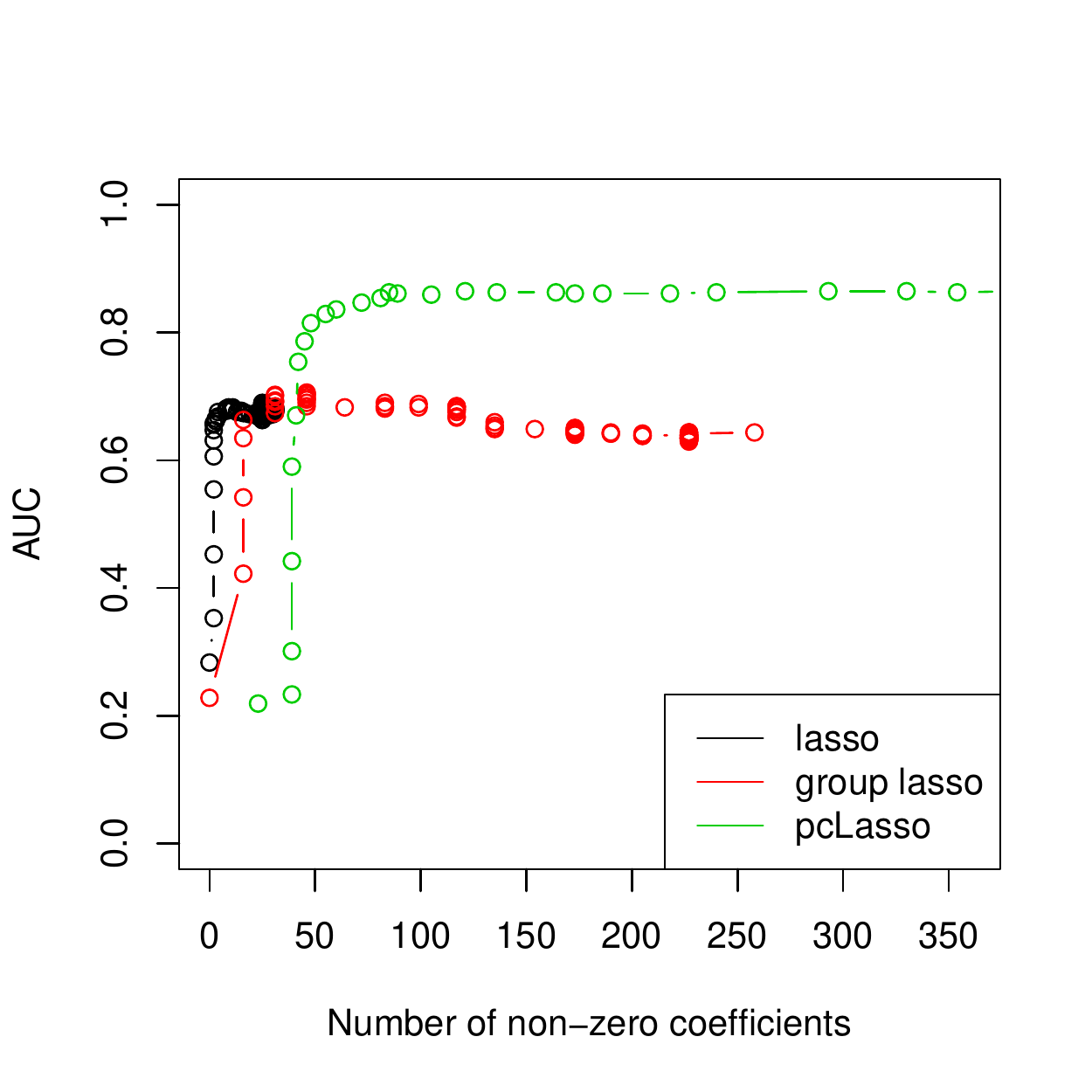}
\caption[fig:fig2]{\em Results for p53 dataset: Left panel shows the number of non-zero pathways vs. the number of non-zero coefficients for the lasso, group lasso and pcLasso. pcLasso exhibits some group-level sparsity, but not as strongly as the group lasso. The right panel is plot of CV area under the curve (AUC) vs. the number of non-zero coefficients in the selected model. pcLasso achieves the best result here.}
\label{fig:p53fig}
\end{figure}

\section{Computation of the pcLasso solution}
\label{sec:computation}
Consider first the simpler case of just one group. Let $\bX=\bU\bD\bV^T$  be the singular value decomposition of $\bX$ with $\bD={\rm diag} (d_1,\ldots, d_m)$, where $m = \text{rank}(\bX)$, and  let $\bA=\bV \bD_{d_1^2-d^2_j}\bV ^T$. The objective function for pcLasso is
\begin{eqnarray}
J(\beta)= \frac{1}{2} \ltwo{\by-\bX\beta}^2 + \lambda  \lone{\beta} + \frac{\theta}{2} \beta^T \bA\beta.
\end{eqnarray}

Since $J$ is convex in $\beta$ and the non-smooth part of the penalty is separable, we can minimize $J$ by applying coordinate descent. The coordinate-wise update has the form
\begin{equation*}
\tilde \beta_j \leftarrow  \frac{\mathcal{S} \Bigl( \sum_i x_{ij}r_i^{(j)}   -\theta s_j,\lambda \Bigr)} {\sum_i x_{ij}^2+\theta A_{jj} }, \end{equation*}
where $r_i^{(j)} = y_i - \sum_{j' \neq j} x_{ij'} \beta_{j'}$ is the partial residual with predictor $j$ removed, $s_j=\sum_\ell A_{j\ell}\beta_\ell-A_{jj}\beta_j$, and $\mathcal{S}$ is the soft-thresholding operator.

These updates can be easily generalized to the general case of $K$ non-overlapping groups \eqref{eqn:pcLasso} and where observations are given different weights (observation $i$ is given weight $w_i$). Using the notation of Section \ref{sec:pcLasso}, we apply coordinate descent as shown in Algorithm {\ref{alg1}.

\bigskip
\begin{algorithm}
\caption{ \em Computation of the pcLasso solution}
\label{alg1}
\begin{enumerate}
\item Input: response $\by_{n\times 1}$, features $\bX_{n\times p}$,  $K$ non-overlapping groups of features  each of
size $p_k$,  fixed value of $\theta\geq 0$. Observation weights $w_i \geq 0$.

Compute the SVD of the columns of $X$ corresponding to each group:
$$\bX_k=\bU_k \bD_k \bV_k^T.$$
Compute $\bA^k=  \bV_k \bD_{d_{k1}^2-d_{kj}^2}\bV_k ^T$ for each group $k$. Given weights $w_i$, let $v_j=\sum_i \tilde w_i x_{ij}^2$, where $\tilde w_i=n w_i/\sum_i w_i$. Set $\hat\beta=0$.

\item Define a grid of values $ (\lambda_{max}, \ldots \lambda_{min})$, where $\lambda_{max}$ is the smallest value of $\lambda$ yielding $\hat\beta=0$. For $\lambda \in  (\lambda_{max}, \ldots \lambda_{min})$:
\begin{description}
\item For $j=1,2,\ldots p,1,2,\ldots$ until convergence:
\begin{description}
\item {(a)} Compute the partial residual $r_i^{(j)} = y_i - \sum_{j' \neq j} x_{ij'}\beta_{j'}$.

\item {(b)} Compute $s_j=\sum_\ell A^k_{j\ell}\beta_\ell-A^k_{jj}\beta_j$, where $k$ is the group containing predictor $j$.

\item {(c)} Update $$\tilde \beta_j \leftarrow  \frac{\mathcal{S} \Bigl( \sum_i \tilde w_i x_{ij}r_i^{(j)}  -\theta s_j,\lambda \Bigr)} {v_j+\theta  A^k_{jj} }. $$
\end{description}
\end{description}
\end{enumerate}

\end{algorithm}
We also generalize this procedure to the binomial/logistic model, using the same Newton-Raphson (iteratively reweighted least squares) framework used in {\tt glmnet} \citep{FHT2010}.

\subsection{pcLasso's grouping effect}
We saw in Figure \ref{fig:p53fig} that pcLasso can have a sparse grouping effect, even though it does not use a two-norm penalty. We investigate this further here. For Figure \ref{fig:sparsegroups}, we generated $n = 50$ observations with 750 predictors, arranged in 50 groups of 15 predictors. Within each group, the predictors were either independent (left panel) or had a strong rank-1 component (right panel). The figure plots the number of non-zero groups (i.e. groups with at least one zero coefficient) against the number of non-zero coefficients for the lasso, the sparse group lasso (SGL) \citep{simon2013} and pcLasso. The latter two methods were run with different tuning parameters. We see that pcLasso shows no grouping advantage over the lasso in the left panel, but produces moderately strong grouping in the right panel. One might argue that this is reasonable: grouping only occurs when the groups have some strong internal structure. The sparse group lasso, on the other hand, always displays a strong grouping effect.

\begin{figure}
\includegraphics[width=6in]{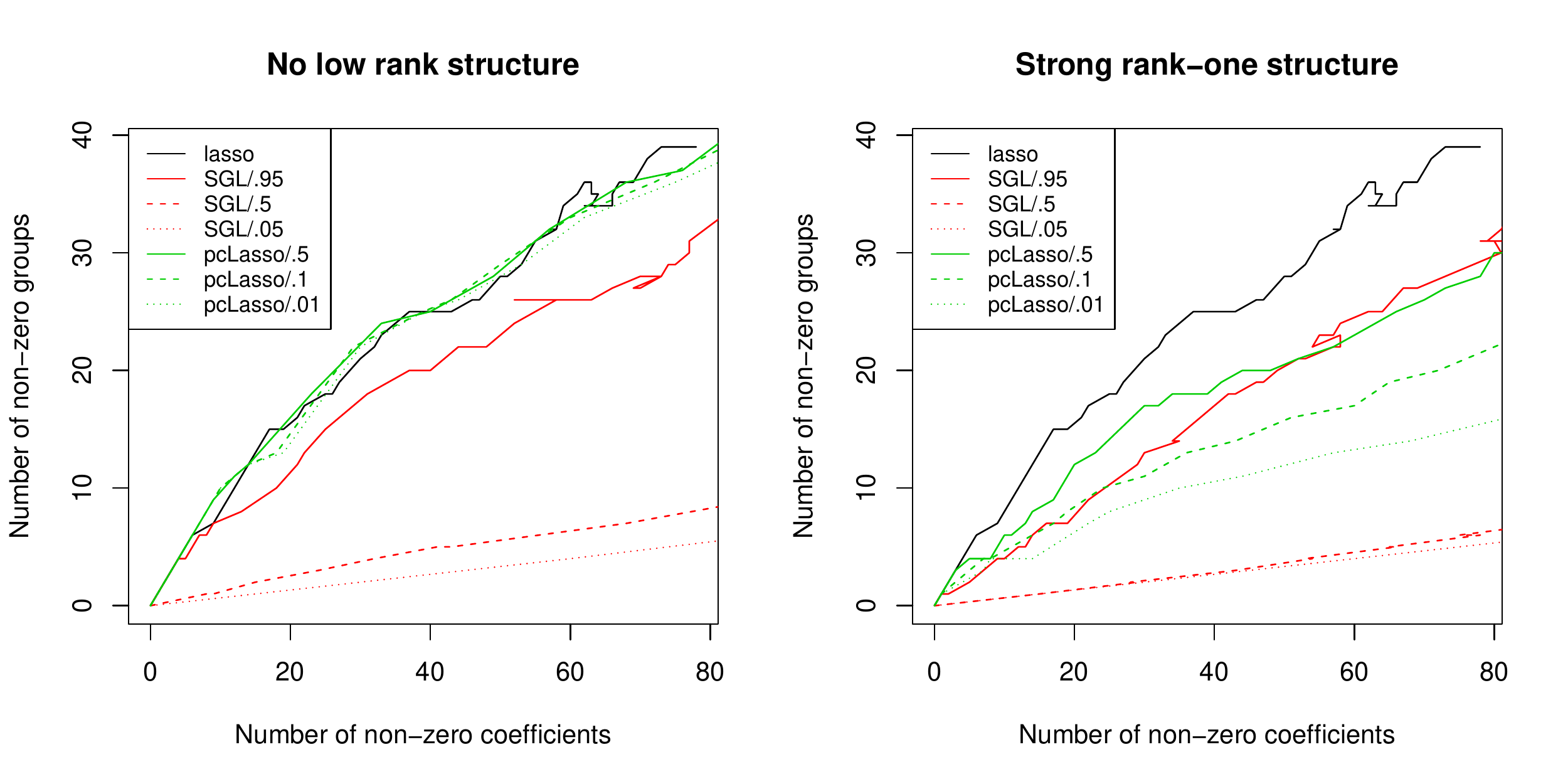}
\caption{\em Grouping results for the lasso, sparse group lasso (SGL) and pcLasso. $n = 50$ and $p = 750$, with the predictors coming in 50 groups of size 15 each. Left panel: Predictors were independent. pcLasso has similar grouping behavior as the lasso. Right panel: Within each group, predictors had a strong rank-1 component. Here, pcLasso shows a stronger grouping effect than the lasso. SGL has a strong grouping effect in both settings.}
\label{fig:sparsegroups}
\end{figure}

\subsection{Connection to the group lasso and a generalization of pcLasso}
The group lasso and sparse group lasso objective functions are given by
\begin{eqnarray}
J_{GL}(\beta)&=&\frac{1}{2} ||\by-\bX\beta||^2 + \lambda_1\sum_k||\beta_k||_2, \\
J_{SGL}(\beta)&=&\frac{1}{2} ||\by-\bX\beta||^2 + \lambda  ||\beta||_1 + \lambda_1\sum_k||\beta_k||_2,
\end{eqnarray}
respectively.  In the original group lasso paper of \cite{Yuan06modelselection} it was assumed that the $\bX_k$, the block of columns in $\bX$ corresponding to group $k$, were orthonormal, i.e. $\bX_k^T\bX_k=\bI$. This makes the computation considerable easier and as pointed out by \cite{SimonTibs}, is equivalent to a penalty of the form $\sum_k \ltwo{\bX_k\beta_k}$. The R package {\tt grpreg} uses this orthogonalization and states that it is essential to its computational performance \citep{Breheny2015}. Since $\ltwo{\bX_k\beta_k} = \ltwo{\bD_k \bV_k^T \beta_k}$, this orthogonalization penalizes the top eigenvectors {\rm more}, in direct contrast to pcLasso which puts {\em less} penalty on the top eigenvectors. Of course, the group lasso without the $\ell_1$ term does not deliver sparsity in the features. We also note that this orthogonalization does not work with the sparse group lasso, as the sparsity among the original features would be lost, and does not work whenever $p_k > n$ for any group $k$.

One could envision a variation of pcLasso, namely a version of the sparse group lasso that uses an objective function of the form
\begin{equation}
\frac{1}{2} ||\by-\bX\beta||^2 + \lambda  ||\beta||_1 + \lambda_1\sum_k||\bR_k^{1/2} \beta_k||_2,
\end{equation}
where $\bR_k$ is any weight matrix. For example, $\bR_k=\bV_k \bD_{d_{k1}^2-d_{kj}^2}\bV_k ^T$ as in \eqref{eqn:pcLasso} would put less penalty on the higher eigenvectors and induce strong group sparsity. We have not experimented with this yet, as the computation seems challenging due to the presence of $\ell_2$ norms.

\subsection{Comparison of timings for model-fitting}

Table \ref{tab:timings} shows a comparison of timings for Algorithm \ref{alg1} against that for other methods, all available as R packages. The competitors include the lasso \texttt{glmnet}, sparse group lasso {\tt SGL} and group lasso packages {\tt gglasso} and {\tt grpreg}. As mentioned earlier, the latter uses orthogonalization within groups to speed up the computations, but this changes the problem. Among these competitors, only {\tt SGL} provides sparsity at the level of individual features. All functions were called with default settings.

We see that {\tt SGL} and {\tt gglasso} start to slow down considerably as the problem size gets moderately large, while the times for {\tt grpreg} are similar to that for pcLasso. However, when the cost for the initial SVD is separated out, the speed of pcLasso is roughly comparable to that of {\tt glmnet}. In practice one can compute the SVD upfront for the given feature matrix, and then use it for subsequent applications of pcLasso. For example, in cross-validation, one could compute just one SVD for the entire feature matrix, rather than one SVD in each fold. Table \ref{tab:timings2} considers
larger problem sizes and we see these same general trends.

\begin{table}[ht]
\centering
\begin{tabular}{rrrrrrrrrr}
  \hline
 & n & p & \texttt{glmnet} & \texttt{SGL} & \texttt{gglasso} & \texttt{grpreg} & SVD for pcLasso & Rest & Total pcLasso \\
  \hline
1 & 100 & 100 & 0.05 & 3.59 & 3.31 & 0.02 & 0.00 & 0.02 & 0.02 \\
  2 & 100 & 200 & 0.02 & 7.14 & 0.39 & 0.03 & 0.01 & 0.01 & 0.01 \\
  3 & 100 & 500 & 0.02 & 29.65 & 2.40 & 0.08 & 0.03 & 0.02 & 0.05 \\
  4 & 500 & 1000 & 0.19 &  & 14.89 & 1.12 & 0.25 & 0.13 & 0.38 \\
  5 & 1000 & 2000 & 0.67 &  &  & 6.38 & 1.50 & 0.54 & 2.04 \\
  6 & 1000 & 5000 & 1.10 &  &  & 17.26 & 11.16 & 2.35 & 13.51 \\
   \hline
\end{tabular}
\caption[tab:timings]{\em Comparison of timings for various algorithms. The predictors are pre-assigned to 10 groups. Time in seconds, average over three runs for an entire path of solutions. The last three columns show the pcLasso model-fitting times broken up by the initial SVD(s) and the rest of the computation. \texttt{SGL} refers sparse group lasso. \texttt{gglasso} and \texttt{grpreg} are group lasso packages; the latter does orthogonalization of predictors.}
\label{tab:timings}
\end{table}

\begin{table}[ht]
\centering
\begin{tabular}{rrrrrrrr}
  \hline
  & n & p & \texttt{glmnet} & \texttt{grpreg} & SVD for pcLasso & Rest & Total pcLasso \\
  \hline
1 & 500 & 1000 & 0.19 & 1.92 & 0.27 & 0.17 & 0.43 \\
  2 & 1000 & 2000 & 0.76 & 5.33 & 1.76 & 0.54 & 2.30 \\
  3 & 2000 & 5000 & 2.49 & 36.57 & 24.15 & 3.81 & 27.96 \\
  4 & 2000 & 10000 & 4.50 & 139.63 & 106.56 & 12.96 & 119.52 \\
   \hline
\end{tabular}
\caption[tab:timings2]{\em As in the previous table, but for larger problem sizes and focusing on just {\tt glmnet}, {\tt grpreg} and pcLasso.}
\label{tab:timings2}
\end{table}

\section{Degrees of freedom}
\label{sec:df}
Given a vector of response values $\by$ with corresponding fits $\hat{\by}$, \cite{Ef86} defines the degrees of freedom as
\begin{equation}
\text{df}(\hat{\by}) = \frac{\sum_i \text{Cov}(\by_i, \hat{\by}_i)}{\sigma^2}.
\end{equation}

The degrees of freedom measures the flexibility of the fit: the larger the degrees of freedom, the more closely the fit matches the response values. For fits of the form $\hat{\by} = \bM \by$, the degrees of freedom is given by $\text{df}(\hat{\by}) = \text{tr}(\bM)$. For the lasso, the number of non-zero elements in the solution is an unbiased estimate of the degrees of freedom \citep{lassodf}. \cite{Zou-thesis} derived an unbiased estimate for the degrees of freedom of the elastic net: if the elastic net solves
\begin{equation}
\minimize_{\beta} \quad \ltwo{\by - \bX\beta}^2 + \lambda_1 \lone{\beta} + \lambda_2 \ltwo{\beta}^2,
\end{equation}
then
\begin{equation}
\widehat{df} = \text{tr} \left[ \bX_{\mathcal{A}} (\bX_{\mathcal{A}}^T \bX_{\mathcal{A}} + \lambda_2 {\bf I})^{-1} \bX_{\mathcal{A}}^T \right]
\end{equation}
is an unbiased estimate for the degrees of freedom, where $\mathcal{A}$ is the active set of the fit.

Using similar techniques as that of \cite{Zou-thesis}, we can derive an unbiased estimate for the degrees of freedom of pcLasso when the number of groups $K$ is equal to one:

\begin{theorem}\label{thm:df}
Let pcLasso be the solution to
\begin{equation}
\minimize_\beta \quad \ltwo{\by - \bX\beta}^2 + \lambda \lone{\beta} + \theta \beta^T \bV \bD_{d_1^2 - d_j^2} \bV^T \beta,
\end{equation}
where $\bD_{d_1^2 - d_j^2}$ is a $p \times p$ matrix with entries $d_1^2 - d_j^2$ ($m = \text{rank}(\bX), d_{m+1} = \dots = d_p = 0$). Let $\bW = (\bV \bD_{d_1^2 - d_j^2} \bV^T)^{1/2}$. Then an unbiased estimate for the degrees of freedom for pcLasso is
\begin{equation}\label{eqn:df-formula}
\widehat{df} = \text{tr} \left[ \bX_{\mathcal{A}} (\bX_{\mathcal{A}}^T \bX_{\mathcal{A}} + \theta \bW_{\mathcal{A}}^T \bW_{\mathcal{A}})^{-1} \bX_{\mathcal{A}}^T \right],
\end{equation}
where $\mathcal{A}$ is the active set of the fit.
\end{theorem}

We verify this estimate through a simulation exercise. Consider the model
\begin{equation}\label{eqn:sim-model}
\by^* = \bX \beta + \mathcal{N}(0,1) \sigma.
\end{equation}

Given this model, for each value of $\lambda$, we can evaluate the degrees of freedom of pcLasso via Monte Carlo methods. For $b = 1, \dots, B$, we generate a new response vector $\by^{*b}$ according to \eqref{eqn:sim-model}. We fit pcLasso to this data, generating predictions $\hat{\by}^{*b}$. This gives us the Monte Carlo estimate
\begin{align*}
df(\lambda) &\approx \sum_{i=1}^n \widehat{\text{Cov}}(\hat{\by}_i^*, \by_i^*) / \sigma^2, \\
\widehat{\text{Cov}}(\hat{\by}_i^*, \by_i^*) &= \frac{1}{B} \sum_{b=1}^B [\hat{\by}_i^{*b} - a_i] [\by_i^{*b} - \by_i^*],
\end{align*}
where the $a_i$'s can be any fixed known constants (usually taken to be 0). On the other hand, each pcLasso fit gives us an active set $\mathcal{A}^{*b}$ which allows us to compute an estimate for the degrees of freedom $\widehat{df}(\lambda)^{*b}$ based on the formula \eqref{eqn:df-formula}. With $B$ replications, we can estimate $\mathbb{E}[\widehat{df}(\lambda)] \approx \frac{1}{B}\sum_{b=1}^B \widehat{df}(\lambda)^{*b}$.

Figure \ref{fig:df-formula} provides evidence for the correctness of \eqref{eqn:df-formula}. The figures on the left are plots of the true value of the degrees of freedom against the mean of the estimate given by \eqref{eqn:df-formula}, with each point corresponding to a value of $\lambda$. We can see that for both values of $\theta$, there is close agreement between the true value $df(\lambda)$ and the expectation of the estimate $\mathbb{E}[\widehat{df}(\lambda)]$. The figures of the right are plots of the bias $\mathbb{E}[\widehat{df}(\lambda)] - df(\lambda)$ along with pointwise 95\% confidence intervals. The zero horizontal line lies inside these confidence intervals.

\begin{figure}
\centering
\includegraphics[width=2.5in]{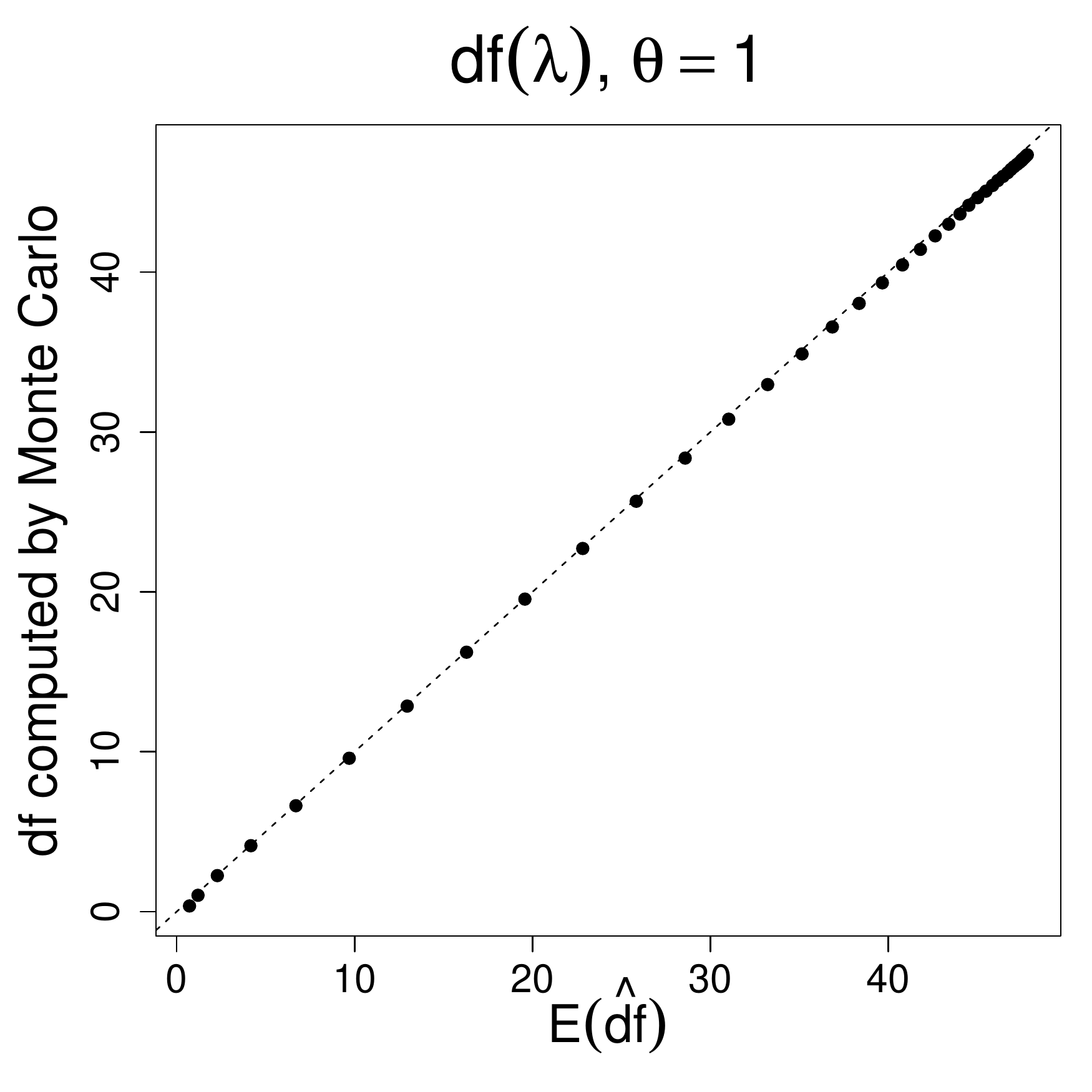} \includegraphics[width=2.5in]{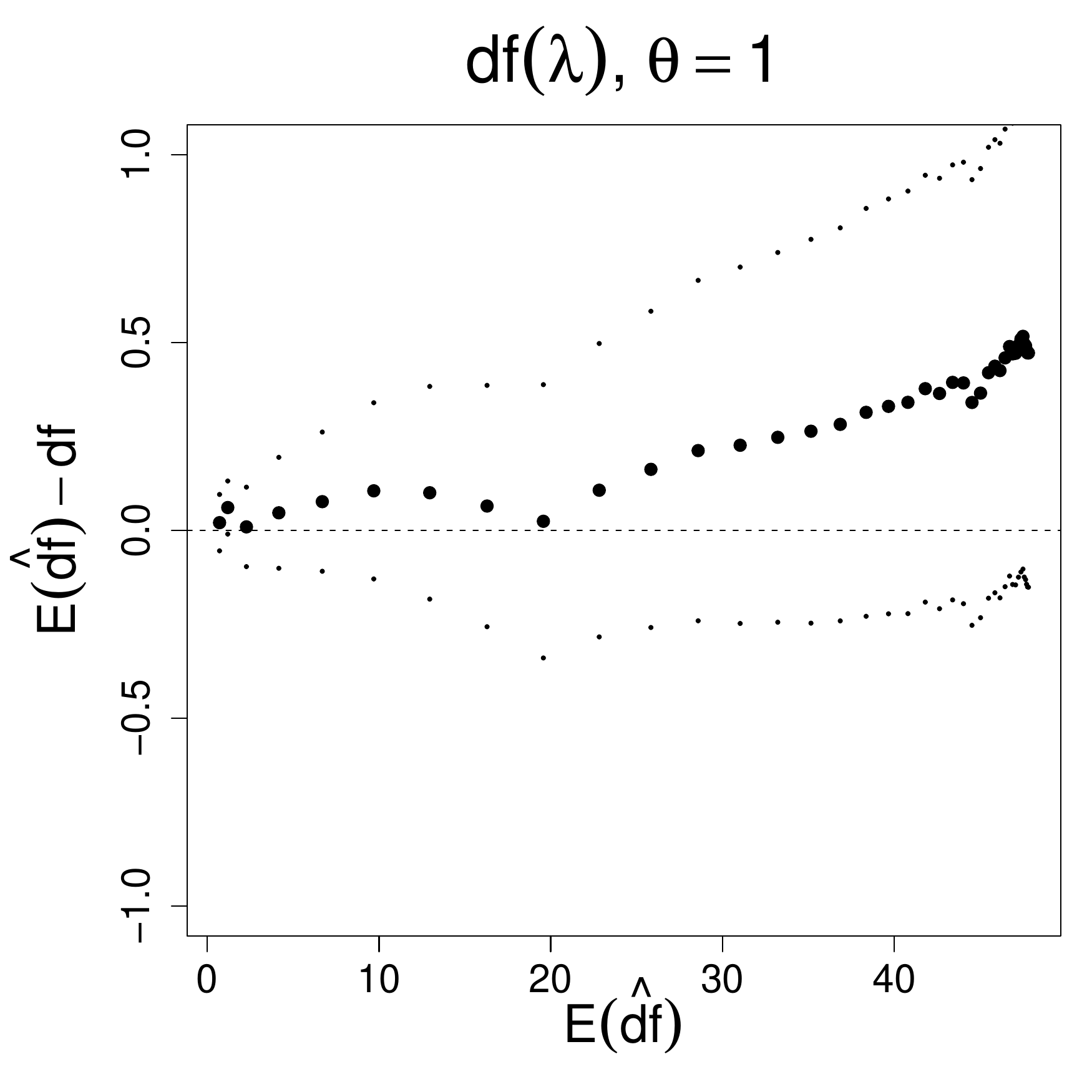}

\includegraphics[width=2.5in]{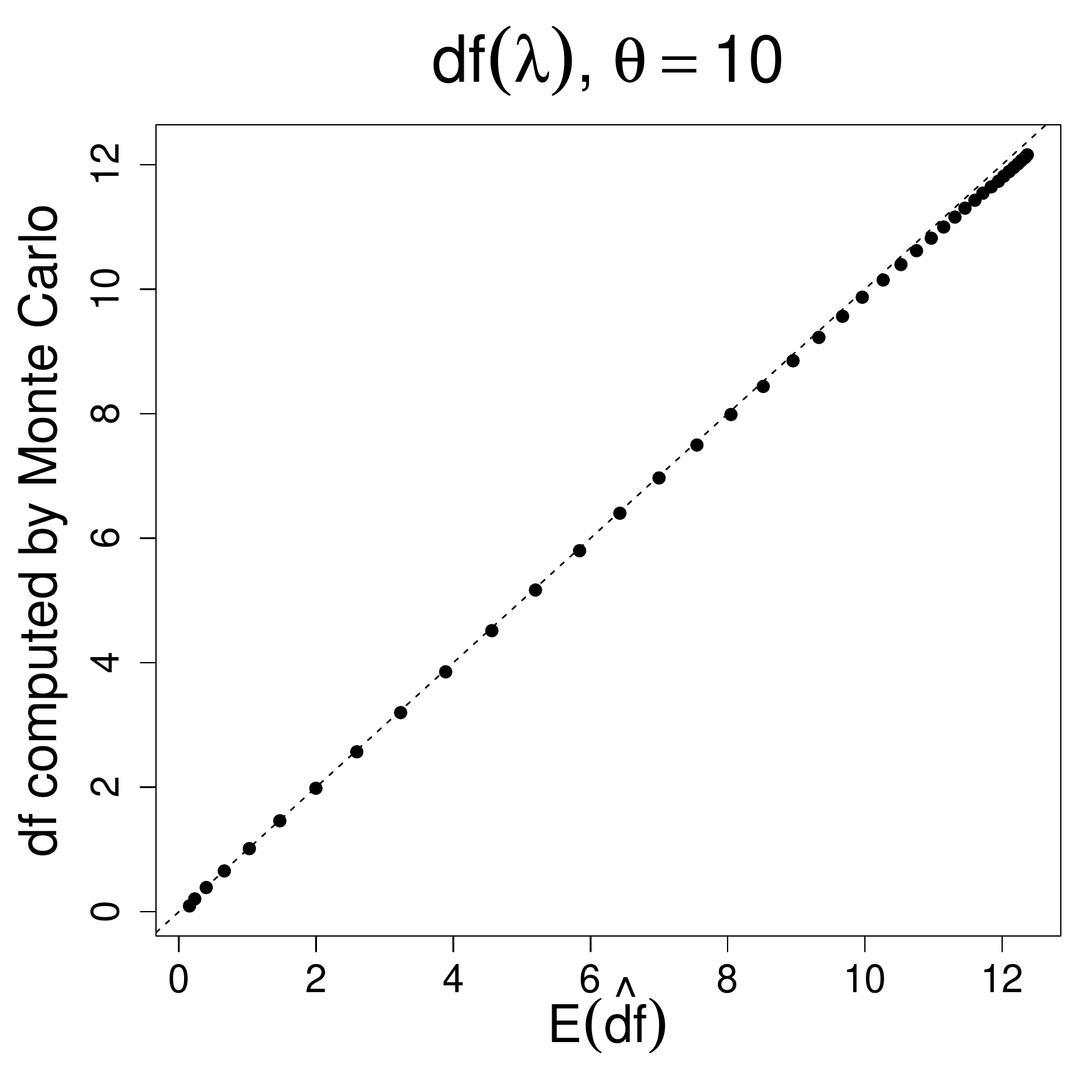} \includegraphics[width=2.5in]{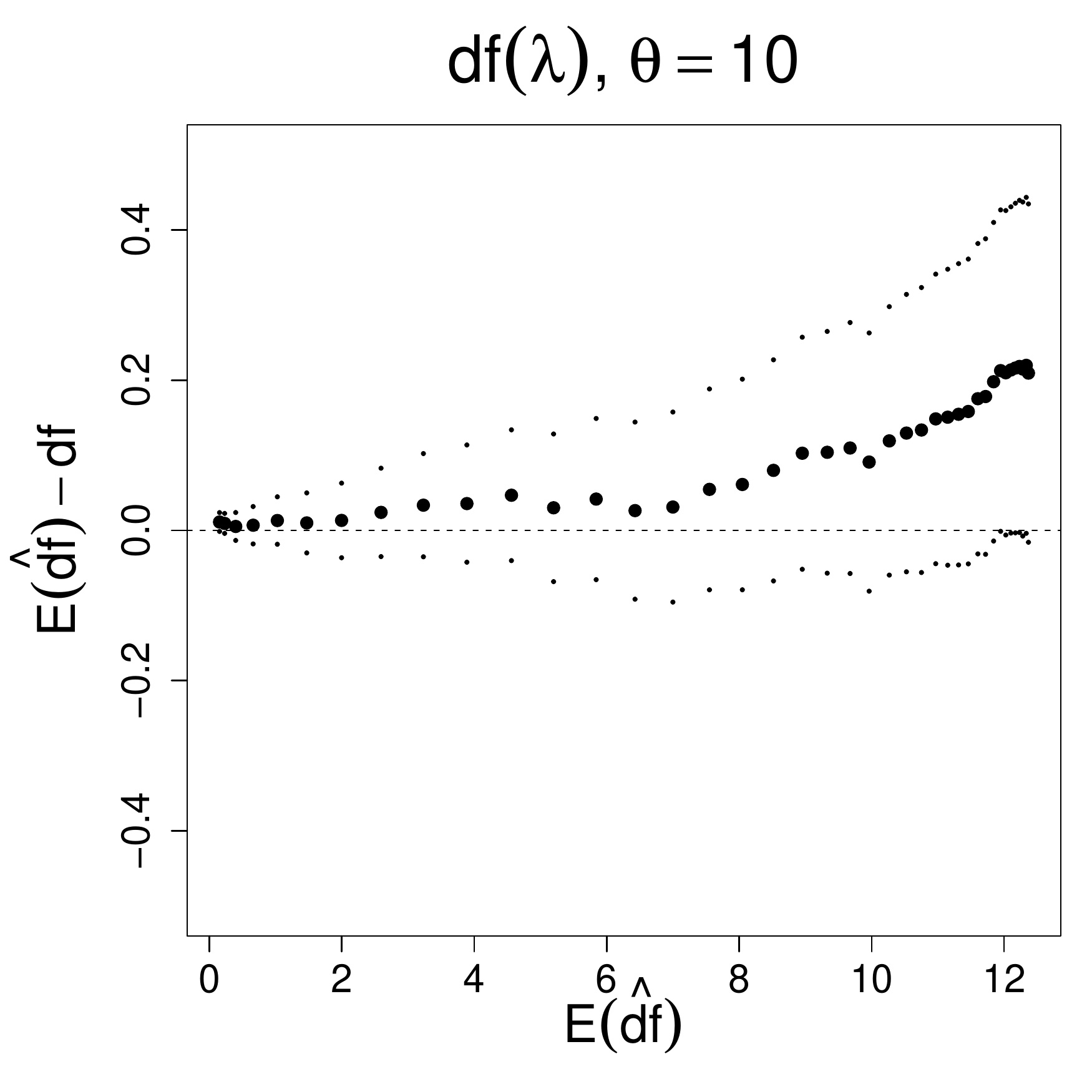}
\caption[fig:df-formula]{\em. Degrees of freedom for pcLasso. The synthetic model is $\by^* = \bX \beta + \mathcal{N}(0,1) \sigma$ with $\beta = 0$, i.e. the null model, and $\sigma = 2$. We set $n = 500$, $p = 100$, and did $500$ Monte Carlo replications. The top two plots are for $\theta = 1$ while the bottom two plots are for $\theta = 10$. In the panels on the left, we compare the true $df(\lambda)$ against $\mathbb{E}[\widehat{df}(\lambda)]$, with the dotted line being the $45^\circ$ line (i.e. perfect match). In the panels on the right, the larger dots represent the bias $\mathbb{E}[\widehat{df}(\lambda)] - df(\lambda)$ across the range of $\mathbb{E}[\widehat{df}(\lambda)]$, and the smaller dots represent the pointwise 95\% confidence intervals. Note that the zero horizontal line lies inside these confidence intervals.}
\label{fig:df-formula}
\end{figure}

 \section{A simulation study}\label{sec:sim}
We tested the performance of pcLasso against other methods in an extensive simulation study. The general framework of the simulation is as follows: The training data has a sample size of $n$ with $p$ features which fall into $K$ groups. Denote the design matrix for group $k$ by $\bX_k$, and let the SVD of $\bX_k$ be $\bX_k = \bU_k \bD_k \bV_k^T$. The true signal is a linear combination of the eigenvectors of the $\bX_k$'s, i.e. $signal = \sum_k \bX_k \bV_k {\bf b}_k$, where the ${\bf b}_k$ are the coefficients of the linear combination. The response $\by$ is the signal corrupted by additive Gaussian noise. We consider three different scenarios:

\begin{itemize}
\item ``Home court" for pcLasso: The signal is related to the top eigenvectors of each group, i.e. the non-zero entries of the ${\bf b}_k$ are all at the beginning.
\item ``Neutral court" for pcLasso: The signal is related to random eigenvectors of each group, i.e. the non-zero entries of the ${\bf b}_k$ are at random positions.
\item ``Hostile court" for pcLasso: The signal is related to the bottom eigenvectors of each group, i.e. the non-zero entries of the ${\bf b}_k$ are all at the end.
\end{itemize}

To induce sparsity in the set-up, we also set ${\bf b}_k = {\bf 0}$ for some $k$, corresponding to group $k$ having no effect on the response. Within each set-up, we looked at a range of signal-to-noise (SNR) ratios in the response, as well as whether there were correlations between predictors in the same group.

For pcLasso, we used the following cross-validation (CV) procedure to select the tuning parameters: For each value of $rat = 0.25, 0.5, 0.75, 0.9, 0.95$ and $1$, we ran pcLasso along a path of $\lambda$ values with \texttt{rat = rat}. (We found that these values of $rat$ covered a good range of models in practice.) For each run, the value of $\lambda$ which gave the smallest CV error was selected, i.e. the \texttt{lambda.min} value as in \texttt{glmnet}. We then compared the CV error across the 6 values of $rat$ and selected the value of $rat$ with the smallest CV error and its accompanying $\lambda$ value. A second version of pcLasso was also run using the same procedure, but with the $\lambda$ values selected by the one standard error rule, i.e. the \texttt{lambda.1se} value as in \texttt{glmnet}.

To compare the methods, we considered the mean-squared error (MSE) achieved on 5,000 test points, as well as the support size of the fitted model. We benchmarked pcLasso against the following methods:

\begin{itemize}
\item The null model, i.e. the mean of the responses in the training data.
\item Elastic net with CV across $\alpha = 0, 0.2, 0.4, 0.6, 0.8$ and $1$, with $\lambda$ values selected both by \texttt{lambda.min} and \texttt{lambda.1se}.
\item Lasso with CV, with $\lambda$ values selected both by \texttt{lambda.min} and \texttt{lambda.1se}.
\item Principal components (PC) regression with CV across ranks.
\end{itemize}

Overall, we found that in ``home court" settings, pcLasso performs the best in terms of test MSE across the range of SNR and feature correlation settings. The gain in performance appears to be larger in low SNR settings compared to high SNR settings. In ``neutral court" and ``hostile court" settings, pcLasso performs on par with the elastic net and the lasso. This is because the cross-validation step often picks $rat = 1$ in these settings, under which pcLasso is equivalent to the lasso. In terms of support size, when there is sparsity pcLasso with $\lambda$ values selected by the one standard error rule tends to pick models which have support size close to the underlying truth.

Below, we present an  illustrative sampling of the results: see Appendix \ref{sec:simstudy-detail} for more comprehensive results across a wider range of simulation settings.

In the first setting, $n = 200$, $p = 1,000$ with the features coming in 10 groups of 100 predictors. The response is a linear combination of the top two eigenvectors of just the first group. The performance on test MSE is shown in Figure \ref{fig:sim-study1}. pcLasso clearly outperforms the other methods when the predictors within each group are uncorrelated. The performance gain is smaller when predictors within each group are correlated with each other. In terms of support size, while pcLasso with $\lambda$ values selected both by \texttt{lambda.min} seems to select models which are too large, pcLasso with $\lambda$ values selected both by \texttt{lambda.1se} has support size closer to the underlying truth.

\begin{figure}
\centering
\includegraphics[width=2.5in]{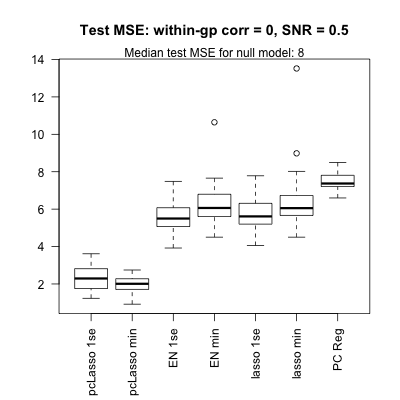} \includegraphics[width=2.5in]{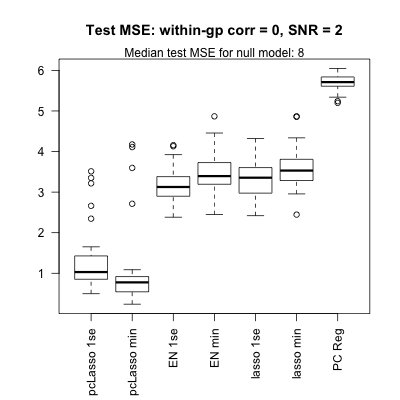}

\includegraphics[width=2.5in]{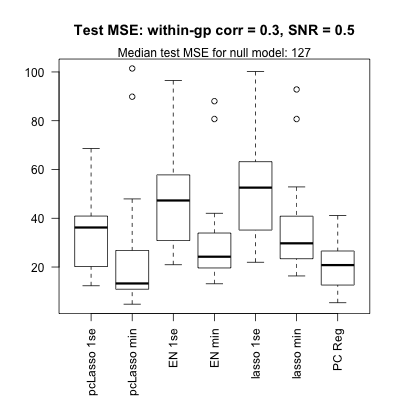} \includegraphics[width=2.5in]{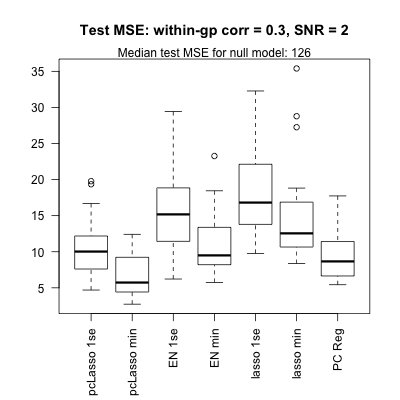}

\caption[fig:sim-study1]{\em ``Home court'' simulation: $n=200$, $p = 1,000$ with features coming in 10 groups of 100 predictors. The response is a linear combination of the top two eigenvectors of just the first group. For the top two figures, the predictors are uncorrelated while for the bottom two figures, each pair of predictors within each group has correlation $0.3$. pcLasso outperforms the other methods across a range of signal-to-noise (SNR) ratios, with the performance gain being larger in the case of uncorrelated predictors.}
\label{fig:sim-study1}
\end{figure}

In the second setting, $n = 200$, $p = 200$ with the features coming in 10 groups of 20 predictors. The response is a linear combination of two random eigenvectors of just the first group. The performance on test MSE is shown in Figure \ref{fig:sim-study2}. pcLasso performs comparably to both the elastic net and the lasso, both when the predictors are uncorrelated and when they are correlated. In terms of support size, pcLasso with $\lambda$ values selected both by \texttt{lambda.1se} has support size close to the underlying truth, especially when SNR is high.

 \begin{figure}
\centering
\includegraphics[width=2.5in]{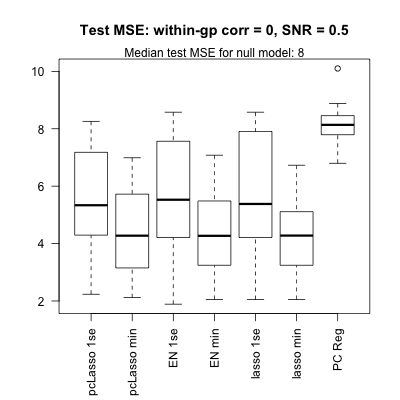} \includegraphics[width=2.5in]{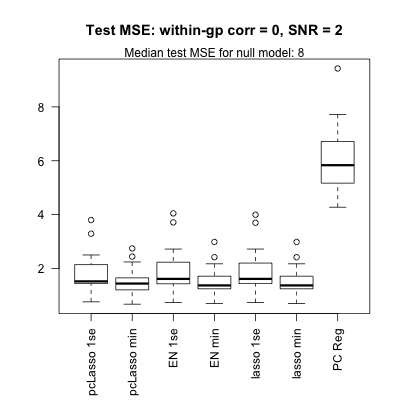}

\includegraphics[width=2.5in]{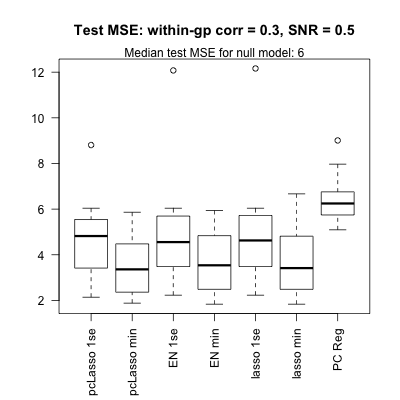} \includegraphics[width=2.5in]{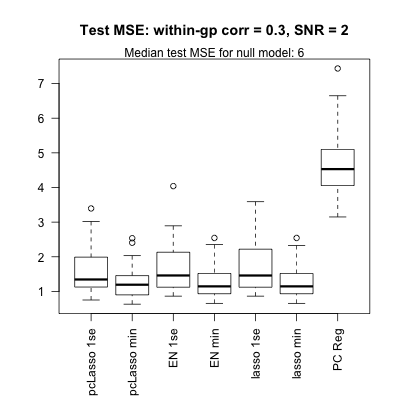}

\caption[fig:sim-study2]{\em ``Neutral court" simulation: $n=200$, $p = 200$ with features coming in 10 groups of 20 predictors. The response is a linear combination of two random eigenvectors of just the first group. For the top two figures, the predictors are uncorrelated while for the bottom two figures, each pair of predictors within each group has correlation $0.3$.  pcLasso performs comparably to the elastic net and the lasso.}
\label{fig:sim-study2}
\end{figure}

In the third setting, $n = 200$, $p = 50$ with the features coming in 5 groups of 10 predictors. The response is a linear combination of the bottom eigenvectors of the first two groups. The performance on test MSE is shown in Figure \ref{fig:sim-study3}. pcLasso performs comparably to both the elastic net and the lasso, both when the predictors are uncorrelated and when they are correlated. In terms of support size, pcLasso with $\lambda$ values selected both by \texttt{lambda.1se} has support size close to the underlying truth when the signal-to-noise ratio is high.

 \begin{figure}
\centering
\includegraphics[width=2.5in]{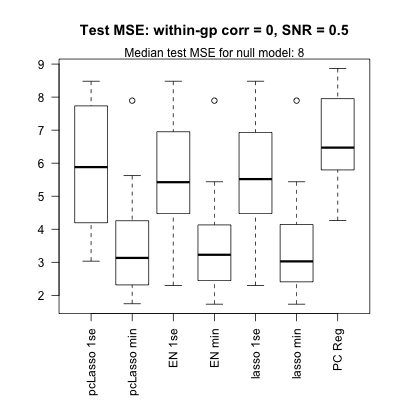} \includegraphics[width=2.5in]{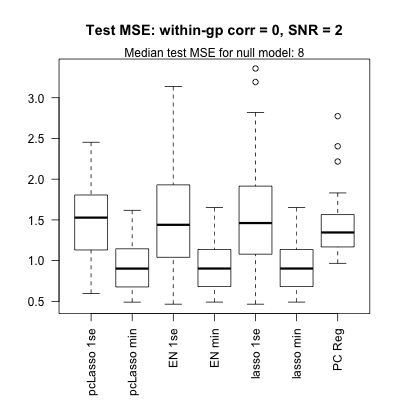}

\includegraphics[width=2.5in]{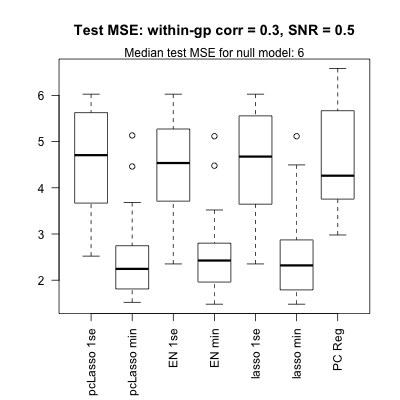} \includegraphics[width=2.5in]{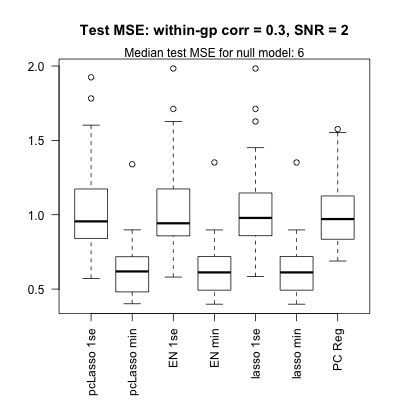}

\caption[fig:sim-study3]{\em ``Hostile court'' simulation: $n = 200$, $p = 50$ with features coming in 5 groups of 10 uncorrelated predictors. The response is a linear combination of the bottom eigenvectors of the first two groups. For the top two figures, the predictors are uncorrelated while for the bottom two figures, each pair of predictors within each group has correlation $0.3$. pcLasso performs comparably to the elastic net and the lasso.}
\label{fig:sim-study3}
\end{figure}







\section{Theoretical properties of pcLasso}\label{sec:theory}
In this section, we present some theoretical properties of pcLasso when the number of groups $K$ is equal to $1$, and compare them with that for the lasso as presented in Chapter 11 of \cite{hastie2015statistical}. In particular, we provide non-asymptotic bounds for $\ell_2$ and prediction error for pcLasso which are better than that for the lasso in certain settings. We note that the results in the section can be extended analogously to pcLasso with $K$ non-overlapping groups if the $\bX_k$ are orthogonal to each other.

To make the proofs simpler, we consider a slightly different penalty for pcLasso: instead of the second penalty term having $\bD_{d_1^2 - d_j^2}$ as an $m \times m$ matrix, where $m = \text{rank}(X)$, we have $\bD_{d_1^2 - d_j^2}$ as a $p \times p$ diagonal matrix with $d_{m+1} = \ldots = d_p = 0$. Let $\bA = \bV \bD_{d_1^2 - d_j^2} \bV^T$.

As the proofs are rather technical, we state just the results here; proofs can be found in Appendix \ref{sec:proofs}.

\subsection{Set-up}
We assume that the true underlying model is
\begin{equation}
\by = \bX \beta^* + \bw,
\end{equation}
where $\bX \in \bbR^{n \times p}$ is the design matrix, $\by, \bw \in \bbR^n$ are the response and noise vectors respectively, and $\beta^* \in \bbR^p$ is the true unknown coefficient vector. We denote the support of $\beta^*$ by $S$. When we solve an optimization problem for $\beta$, we denote the estimate by $\widehat{\beta}$ and the error by $\widehat{\nu} = \widehat{\beta} - \beta^*$.

In this section, assume that $\theta > 0$ is fixed. (If $\theta = 0$, pcLasso is equivalent to the lasso.) Define
\begin{equation}
\widetilde{\by} = \begin{pmatrix} \by \\ 0 \end{pmatrix}, \quad \widetilde{\bX} = \begin{pmatrix} \bX \\ \sqrt{n\theta}\sqrt{\bA} \end{pmatrix}, \quad \widetilde{\bw} = \begin{pmatrix} \bw \\ -\sqrt{n\theta} \sqrt{\bA} \beta^* \end{pmatrix}.
\end{equation}
Thus, $\widetilde{\by} = \widetilde{\bX}\beta^* + \widetilde{\bw}$. The key idea is that with this notation, pcLasso is equivalent to the lasso for the augmented matrices $\widetilde{\bX}$ and $\widetilde{\by}$. Explicitly, the constrained form of pcLasso solves
\begin{equation}\label{eqn:cpcLasso}
\minimize_\beta \quad \ltwo{\widetilde{\by} - \widetilde{\bX} \beta }^2 \qquad \subjectto \lone{\beta} \leq R,
\end{equation}
while the Lagrangian form \eqref{eqn:pcLasso1} solves
\begin{equation}
\minimize_{\beta} \quad \frac{1}{2n}\ltwo{\widetilde{\by} - \widetilde{\bX} \beta }^2 + \lambda \lone{\beta}. \label{eqn:lpcLasso}
\end{equation}

(We have changed the fraction in front of the RSS term to $\frac{1}{2n}$ so that the results are more directly comparable to those in \cite{hastie2015statistical}.) In the high-dimensional setting, it is standard to impose a \textit{restricted eigenvalue condition} on the design matrix $\bX$:
\begin{equation}\label{eqn:rec-lasso}
\frac{1}{n}\nu^T \bX^T \bX \nu \geq \gamma \ltwo{\nu}^2 \qquad \text{for all nonzero } \nu \in \mathcal{C},
\end{equation}
with $\gamma > 0$ and $\mathcal{C}$ an appropriately chosen constraint set. We assume additionally that $d_1^2 > n\gamma$. This is not a restrictive assumption: since $d_1^2$ is the top eigenvalue of $\bX^T \bX$, we automatically have $d_1^2 \geq n\gamma$. Equality can only happen if $\mathcal{C}$ is a subset of the eigenspace associated with $d_1^2$.

Since $\bA$ is a positive semidefinite matrix, \eqref{eqn:rec-lasso} holds trivially for the augmented matrix $\widetilde{\bX}$ as well:
\begin{equation*}
\nu^T \widetilde{\bX}^T \widetilde{\bX} \nu = \nu^T (\bX^T \bX + n\theta \bA)\nu \geq n\gamma \ltwo{\nu}^2.
\end{equation*}
The following key lemma shows that we can improve the constant on the RHS of \eqref{eqn:rec-lasso}. This will give us better rates of convergence for pcLasso.
\begin{lemma}\label{lem:betterbound2}
If $\bX$ satisfies the restricted eigenvalue condition \eqref{eqn:rec-lasso}, then
\begin{equation}
\nu^T \widetilde{\bX}^T \widetilde{\bX} \nu \geq \min \left[ (1 - n\theta)n\gamma + n\theta d_1^2, d_1^2 \right] \ltwo{\nu}^2 \qquad \text{for all nonzero } \nu \in \mathcal{C}.
\end{equation}
\end{lemma}

We note also that the augmented matrix actually satisfies an \textit{unrestricted eigenvalue condition}:
\begin{lemma}\label{lem:uec}
For any design matrix $\bX$, the augmented data matrix $\widetilde{\bX}$ satisfies
\begin{equation}
\nu^T \widetilde{\bX}^T \widetilde{\bX} \nu \geq \min \left( n\theta, 1 \right) d_1^2 \ltwo{\nu}^2 \qquad \text{for all } \nu \in \mathbb{R}^p.
\end{equation}
\end{lemma}

This will give us a better ``slow rate" for convergence of prediction error, as well as bounds on $\ell_2$ error without requiring a restricted eigenvalue condition.

\subsection{Bounds on $\ell_2$ error}
 The following theorem establishes a bound for $\ell_2$ error of the constrained form of pcLasso:
 \begin{theorem}\label{thm:l2bound-c}
Suppose that $\bX$ satisfies the restricted eigenvalue bound \eqref{eqn:rec-lasso} with parameter $\gamma > 0$ over the set $\{ \nu \in \bbR^p: \lone{\nu_{S^c}} \leq \lone{\nu_{S}} \}$. Then any estimate $\widehat{\beta}$ based on the constrained pcLasso \eqref{eqn:cpcLasso} with $R = \lone{\beta^*}$ satisfies the bound
\begin{equation}\label{eqn:l2bound}
\ltwo{\widehat{\beta} - \beta^*} \leq \frac{4 \sqrt{|S|} \linf{\bX^T \bw - n\theta \bA \beta^*} }{ \min \left[ (1 - n\theta)n\gamma + n\theta d_1^2, d_1^2 \right] } \leq \frac{4 \sqrt{|S|} \left(n\theta \linf{\bA \beta^*} + \linf{\bX^T \bw}\right) }{ \min \left[ (1 - n\theta)n\gamma + n\theta d_1^2, d_1^2 \right] }.
\end{equation}
In particular, if $\widehat{\beta}$ is aligned with the first principal component of $\bX$, then
\begin{equation}\label{eqn:l2bound-1pc}
\ltwo{\widehat{\beta} - \beta^*} \leq \frac{4 \sqrt{|S|} \linf{\bX^T \bw} }{\min \left[ (1 - n\theta)n\gamma + n\theta d_1^2, d_1^2 \right]},
\end{equation}
which is a \textbf{better} rate of convergence than that for the lasso.
\end{theorem}

What is interesting is that we can actually obtain a similar (weaker) bound without the restricted eigenvalue condition; the lasso does not have such a bound.

\begin{theorem}\label{thm:l2bound-c2}
For any design matrix $\bX$, any estimate $\widehat{\beta}$ based on the constrained pcLasso \eqref{eqn:cpcLasso} with $R = \lone{\beta^*}$ satisfies the bound
\begin{equation}
\ltwo{\widehat{\beta} - \beta^*} \leq \frac{4 \sqrt{|S|} \linf{\bX^T \bw - n\theta \bA \beta^*} }{ \min \left( n\theta, 1 \right) d_1^2 }.
\end{equation}
\end{theorem}

We can obtain similar results for the Lagrangian form of pcLasso:
\begin{theorem}\label{thm:l2bound-l}
Suppose that $\bX$ satisfies the restricted eigenvalue bound \eqref{eqn:rec-lasso} with parameter $\gamma > 0$ over the set $\{ \nu \in \bbR^p: \lone{\nu_{S^c}} \leq 3\lone{\nu_{S}} \}$. Then any estimate $\widehat{\beta}$ based on the Lagrangian form of pcLasso \eqref{eqn:lpcLasso} with $\lambda \geq \frac{2}{n} \linf{\bX^T \bw - n\theta \bA \beta^*} > 0$ satisfies the bound
\begin{equation}\label{eqn:l2bound-l}
\ltwo{\widehat{\beta} - \beta^*} \leq \frac{3\lambda \sqrt{|S|}}{ \min \left[ (1 - n\theta)n\gamma + n\theta d_1^2, d_1^2 \right] / n}.
\end{equation}

If we remove the restricted eigenvalue condition on $\bX$, then the bound \eqref{eqn:l2bound-l} holds but with the denominator of the RHS being $\min \left( n\theta, 1 \right) d_1^2 / n$ instead.
\end{theorem}

\subsection{Bounds on prediction error}
Like the lasso, pcLasso has a ``slow rate" convergence for prediction error:
\begin{theorem}\label{thm:predbound-lslow}
For any design matrix $\bX$, consider the Lagrangian form of pcLasso \eqref{eqn:lpcLasso} with $\lambda \geq \frac{2}{n} \linf{\bX^T \bw - n\theta \bA \beta^*}$. We have the following bound on prediction error:
\begin{equation}
\frac{1}{n}\ltwo{\bX(\widehat{\beta} - \beta^*)}^2 \leq 12\lambda \lone{\beta^*}.
\end{equation}
\end{theorem}

By using Lemma \ref{lem:uec}, we can get a smaller bound for the RHS, although the expression is much more complicated:
\begin{theorem}\label{thm:predbound-lslow2}
For any design matrix $\bX$, consider the Lagrangian form of pcLasso \eqref{eqn:lpcLasso} with $\lambda \geq \frac{2}{n} \linf{\bX^T \bw - n\theta \bA \beta^*}$. Let $\kappa = \min (n\theta, 1) d_1^2$. We have the following bound on prediction error:
\begin{equation}
\frac{1}{n}\ltwo{\bX(\widehat{\beta} - \beta^*)}^2 \leq \frac{ 3\lambda (-\lambda p + \sqrt{\lambda^2 p^2 + 32\lambda \lone{\beta^*} \kappa p} )}{4\kappa}.
\end{equation}
\end{theorem}

This means that pcLasso has a better ``slow rate" convergence than the lasso.

The preceding two theorems hold for all design matrices $\bX$. If we are willing to assume that $\bX$ satisfies the restricted eigenvalue condition, we recover the same ``fast rate" convergence for prediction error as that for the lasso.
\begin{theorem}\label{thm:predbound-lfast}
Suppose that $\bX$ satisfies the restricted eigenvalue bound \eqref{eqn:rec-lasso} with parameter $\gamma > 0$ over the set $\{ \nu \in \bbR^p: \lone{\nu_{S^c}} \leq 3\lone{\nu_{S}} \}$. For the Lagrangian form of pcLasso with $\lambda \geq \frac{2}{n} \linf{\bX^T \bw - n\theta \bA \beta^*}$, we have
\begin{equation}
\frac{1}{n}\ltwo{\bX(\widehat{\beta} - \beta^*)}^2 \leq \frac{9|S| \lambda^2}{\min \left[ (1 - n\theta)n\gamma + n\theta d_1^2, d_1^2 \right] / n}.
\end{equation}
\end{theorem}





\section{Discussion}
\label{sec:discussion}

We have proposed a new method for supervised learning that exploits the correlation and group structure of predictors to potentially improve prediction performance. It combines the power of PC regression with the sparsity of the lasso. There are several interesting avenues for further research:
\begin{itemize}
\item {\em Efficient screening rules}: To speed up the computation of the lasso solution, {\tt glmnet} uses strong rules \citep{TBFHSTT2012} to discard predictors which are likely to have a coefficient of zero. These strong rules can greatly reduce the number of variables entering the optimization, thus speeding up computation. Using similar techniques, we can derive strong rules for pcLasso to improve computational efficiency. These rules are provided in Appendix \ref{sec:strongrules}, and merit further investigation.

\item {\em Use of different kinds of correlation or grouping information.} Rather than use the covariance of the observed covariates in a group, one could use other kinds of external information to form this covariance, for example gene or protein pathways. Alternatively, one could use unsupervised clustering to generate the feature groups.

\item {\em Optimality of the quadratic penalty term}: Considering the general class of penalty functions of the form $\beta^T \bV \bZ \bV^T \beta$, one could ask if pcLasso's choice of $\bZ = \bD_{d_1^2 - d_j^2}$ is ``optimal" in any sense. One could also look for $\bZ$ which satisfy certain notions of optimality.
\end{itemize}

An R language package {\tt pcLasso} will soon be made available on the CRAN repository.

\medskip

{\bf Acknowledgements:}  We'd like to thank Trevor Hastie, Jacob Bien and Noah Simon for helpful comments. Robert Tibshirani was supported by
NIH grant 5R01 EB001988-16 and NSF grant 19 DMS1208164.

\appendix
\section{Details of the pcLasso penalty contours for two predictors}\label{sec:contours}
Assume that we only have two predictors which are standardized to have mean $0$ and sum of squares $1$. Then $\bX^T \bX = \begin{pmatrix} 1 & \rho \\ \rho & 1 \end{pmatrix}$, where $\rho$ is the sample correlation between the two predictors. Let $\bA = \bV \bD_{d_1^2 - d_j^2} \bV^T$. If $\rho > 0$, the expressions for the SVD of $\bX$ and $\bA$ are
\begin{equation*}
\bX = \bU \cdot \begin{pmatrix} 1 + \rho & 0 \\ 0 & 1 - \rho \end{pmatrix} \cdot \begin{pmatrix} 1 / \sqrt{2} & 1 / \sqrt{2} \\ 1 / \sqrt{2} & -1 / \sqrt{2} \end{pmatrix}, \qquad \bA = 2\rho \begin{pmatrix} 1 & -1 \\ -1 & 1 \end{pmatrix}.
\end{equation*}

If $\rho < 0$, the corresponding expressions are
\begin{equation*}
\bX = \bU \cdot \begin{pmatrix} 1 - \rho & 0 \\ 0 & 1 + \rho \end{pmatrix} \cdot \begin{pmatrix} 1 / \sqrt{2} & -1 / \sqrt{2} \\ 1 / \sqrt{2} & 1 / \sqrt{2} \end{pmatrix}, \qquad \bA = -2\rho \begin{pmatrix} 1 & 1 \\ 1 & 1 \end{pmatrix}.
\end{equation*}

With these expressions, the pcLasso penalty can be written explicitly as
\begin{equation*} \begin{cases} \lambda(|\beta_1| + |\beta_2|) + 2 \theta \rho (\beta_1 - \beta_2)^2 &\text{if } \rho > 0, \\ \lambda(|\beta_1| + |\beta_2|) - 2 \theta \rho (\beta_1 + \beta_2)^2 &\text{if } \rho < 0.  \end{cases}
\end{equation*}

Considering the signs of $\beta_1$ and $\beta_2$, we can get an even more explicit description of the contours. We provide the description for $\rho > 0$ below:
\begin{itemize}
\item $\beta_1 \geq 0, \beta_2 \geq 0$: The parabola $\beta_2 = -\dfrac{2\sqrt{2}\theta \rho}{\lambda}\beta_1^2 + \dfrac{C}{\sqrt{2}\lambda}$, rotated $45^\circ$ clockwise.
\item $\beta_1 \geq 0, \beta_2 \leq 0$: The line $\beta_2 = \beta_1 + \dfrac{\lambda - \sqrt{\lambda^2 + 8\theta \rho C}}{4 \theta \rho}$.
\item $\beta_1 \leq 0, \beta_2 \leq 0$: The parabola $\beta_2 = \dfrac{2\sqrt{2}\theta \rho}{\lambda}\beta_1^2 - \dfrac{C}{\sqrt{2}\lambda}$, rotated $45^\circ$ clockwise.
\item $\beta_1 \leq 0, \beta_2 \geq 0$: The line $\beta_2 = \beta_1 - \dfrac{\lambda - \sqrt{\lambda^2 + 8\theta \rho C}}{4 \theta \rho}$.
\end{itemize}

\section{Strong rules for pcLasso}\label{sec:strongrules}
Recall that in computing the pcLasso solution, we typically hold $\theta$ fixed and compute the solution for a path of $\lambda$ values. By casting the pcLasso optimization problem as a lasso problem with a different response and design matrix, we can derive strong rules for pcLasso from that for the lasso.

We first derive strong rules for pcLasso where predictors do not have preassigned groups, i.e. the solution to \eqref{eqn:pcLasso1}. If we define
\begin{equation}\label{eqn:strongrulestilde}
\widetilde{\by} = \begin{pmatrix} \by \\ 0 \end{pmatrix}, \qquad \widetilde{\bX} = \begin{pmatrix} \bX \\ \sqrt{\theta} \bA^{1/2} \end{pmatrix},
\end{equation}
where $\bA = \bV \bD_{d_1^2 - d_j^2} \bV^T$, then pcLasso is equivalent to the lasso with response vector $\widetilde{\by}$ and design matrix $\widetilde{\bX}$. Fix $\theta$, fix a path of $\lambda$ values $\lambda_1 \geq \lambda_2 \geq \dots$, and let $\hat{\beta}(\lambda_i)$ denote the pcLasso solution at $\lambda = \lambda_i$. Applying the lasso strong rules to this set-up, the sequential strong rule for pcLasso for discarding predictor $j$ at $\lambda = \lambda_i$ is
\begin{align}
&&|\bX_j^T (\by - \bX\hat{\beta}(\lambda_{i-1})) - \theta (\bA^{1/2})_j \bA^{1/2}\hat{\beta}(\lambda_{i-1})| &< 2\lambda_i - \lambda_{i-1}, \nonumber \\
\Leftrightarrow \qquad&& |\bX_j^T (\by - \bX\hat{\beta}(\lambda_{i-1})) - \theta (\bA \hat{\beta}(\lambda_{i-1}))_j | &< 2\lambda_i - \lambda_{i-1}. \label{eqn:strongrules1group}
\end{align}

Next we derive strong rules for pcLasso where predictors come in preassigned non-overlapping groups, i.e. the solution to \eqref{eqn:pcLasso}. If we define $\bA_k = \bV_k \bD_{d_{k1}^2 - d_{kj}^2} \bV_k^T$ for $k = 1, \dots, K$, $\bA$ as the block diagonal matrix
\[ \bA = \begin{pmatrix} \bA_1 & & 0 \\ & \ddots & \\ 0 & & \bA_K \end{pmatrix}, \]
and $\widetilde{\by}$ and $\widetilde{\bX}$ as in \eqref{eqn:strongrulestilde}, then pcLasso is again equivalent to the lasso with $\widetilde{\by}$ and $\widetilde{\bX}$. This results in the same sequential strong rule as in the single group case \eqref{eqn:strongrules1group}, although the matrix $\bA$ is defined differently.

\section{Proofs for Section \ref{sec:theory}}\label{sec:proofs}
Before presenting proofs for Section \ref{sec:theory}, we first prove two technical lemmas which show that the pcLasso solution lies in the constraint set which we will use for the restricted eigenvalue condition \eqref{eqn:rec-lasso}.

\begin{lemma}\label{lem:cpcLasso_rec}
If we solve the constrained form of pcLasso \eqref{eqn:cpcLasso} with $R = \lone{\beta^*}$, then $\lone{\widehat{\nu}_{S^c}} \leq \lone{\widehat{\nu}_S}$.
\end{lemma}

\begin{proof}
Since $\widehat{\beta}$ is feasible for \eqref{eqn:cpcLasso}, by the triangle inequality,
\begin{equation*}
\lone{\beta_S^*} = R \geq \lone{\beta^* + \widehat{\nu}} = \lone{\beta_S^* + \widehat{\nu}_S} + \lone{\widehat{\nu}_{S^c}} \geq \lone{\beta_S^*} - \lone{\widehat{\nu}_S} + \lone{\widehat{\nu}_{S^c}}.
\end{equation*}
\end{proof}

\begin{lemma}\label{lem:lpcLasso_rec}
If we solve the Lagrangian form of pcLasso \eqref{eqn:lpcLasso} with $\lambda \geq \frac{2}{n} \linf{\bX^T \bw - n\theta \bA \beta^*}$, then $\lone{\widehat{\nu}_{S^c}} \leq 3\lone{\widehat{\nu}_S}$.
\end{lemma}

\begin{proof}
Define $G(\nu) = \dfrac{1}{2n}\ltwo{\widetilde{\by} - \widetilde{\bX}(\beta^* + \nu)}^2 + \lambda \lone{\beta^* + \nu}$. By definition, $G(\widehat{\nu}) \leq G(0)$, so
\begin{align}
\frac{1}{2n}\ltwo{\widetilde{\bw} - \widetilde{\bX}\widehat{\nu}}^2 + \lambda \lone{\beta^* + \widehat{\nu}} &\leq \frac{1}{2n}\ltwo{\widetilde{\bw}}^2 + \lambda \lone{\beta^*}, \nonumber \\
\frac{1}{2n}\widehat{\nu}^T (\bX^T \bX + n \theta \bA) \widehat{\nu} &\leq \frac{\widetilde{\bw}^T \widetilde{\bX} \widehat{\nu}}{n} + \lambda (\lone{\beta^*} - \lone{\beta^* + \widehat{\nu}}). \label{eqn:basic-pcLasso1}
\end{align}

Note that $\lone{\beta^* + \widehat{\nu}} = \lone{\beta_S^* + \widehat{\nu}_S} + \lone{\widehat{\nu}_{S^c}} \geq \lone{\beta_S^*} - \lone{\widehat{\nu}_S} + \lone{\widehat{\nu}_{S^c}}$; substituting this in the above gives
\begin{equation}\label{eqn:basic-pcLasso2}
\frac{1}{2n}\widehat{\nu}^T (\bX^T \bX + n\theta \bA) \widehat{\nu} \leq \frac{\widetilde{\bw}^T \widetilde{\bX} \widehat{\nu}}{n} + \lambda (\lone{\widehat{\nu}_S} - \lone{\widehat{\nu}_{S^c}}).
\end{equation}

By H\"{o}lder's inequality and our choice of $\lambda$, we have
\begin{equation}\label{eqn:basic-pcLasso3}
\frac{\widetilde{\bw}^T \widetilde{\bX} \widehat{\nu}}{n} \leq \frac{1}{n} \linf{\widetilde{\bX}^T \widetilde{\bw}} \lone{\widehat{\nu}} = \frac{1}{n} \linf{\bX^T \bw - n\theta \bA \beta^*} \lone{\widehat{\nu}} \leq \frac{1}{2}\lambda \lone{\widehat{\nu}}.
\end{equation}
Substituting this inequality into \eqref{eqn:basic-pcLasso2} and noting that the LHS of \eqref{eqn:basic-pcLasso2} is always non-negative, we obtain the desired conclusion.

\end{proof}

\subsection{Proof of Lemma \ref{lem:betterbound2}}\label{sec:proof:betterbound2}
Let $\nu \in \mathcal{C}$ be expressed as $\nu = \sum_{j=1}^p a_j v_j$, where the $v_j$'s are the columns of $\bV$. Direct computation gives
\begin{equation}\label{eqn:coordinates}
\ltwo{\nu}^2 = \sum_{j=1}^p a_j^2, \qquad \nu^T \bX^T \bX \nu = \sum_{j=1}^p d_j^2 a_j^2, \qquad \nu^T (n\theta \bA) \nu = \sum_{j=2}^p n\theta(d_1^2 - d_j^2)a_j^2.
\end{equation}

Thus, for any $\alpha \in [0, 1]$,
\begin{align}
\nu^T (\widetilde{\bX}^T \widetilde{\bX}) \nu &= \nu^T (\bX^T \bX + n\theta \bA) \nu \nonumber \\
&= \sum_{j=1}^p d_j^2 a_j^2 + \sum_{j=2}^p n\theta(d_1^2 - d_j^2)a_j^2 \nonumber \\
&= \alpha \sum_{j=1}^p d_j^2 a_j^2 + (1-\alpha)d_1^2 a_1^2 + \sum_{j=2}^p \left[ n\theta d_1^2 + (1 - \alpha - n\theta) d_j^2 \right] a_j^2 \nonumber \\
&\geq \alpha n \gamma \sum_{j=1}^p a_j^2 + (1-\alpha)d_1^2 a_1^2 + \min \left[ n\theta d_1^2, n\theta d_1^2 + (1-\alpha-n\theta)d_2^2 \right] \sum_{j=2}^p a_j^2, \nonumber \\
\nu^T (\widetilde{\bX}^T \widetilde{\bX}) \nu &\geq \left[ \alpha n\gamma + \min \left( (1-\alpha)d_1^2, n\theta d_1^2, n\theta d_1^2 + (1-\alpha- n\theta)d_2^2 \right) \right]\ltwo{\nu}^2. \label{eqn:betterbound}
\end{align}

If $n\theta \leq 1$, setting $\alpha = 1 - n\theta$ in \eqref{eqn:betterbound} gives
\begin{align*}
\nu^T (\widetilde{\bX}^T \widetilde{\bX}) \nu &\geq \left[ (1- n\theta) n\gamma + \min \left( n\theta d_1^2, n\theta d_1^2, n\theta d_1^2 \right) \right]\ltwo{\nu}^2 \\
&= [(1- n\theta)n\gamma + n\theta d_1^2] \ltwo{\nu}^2.
\end{align*}

If $n\theta \geq 1$, setting $\alpha = 0$ in \eqref{eqn:betterbound} gives
\begin{align*}
\nu^T (\widetilde{\bX}^T \widetilde{\bX}) \nu &\geq \min \left( d_1^2, n\theta (d_1^2 - d_2^2) + d_2^2 \right)\ltwo{\nu}^2 \\
&\geq \min \left( d_1^2, (d_1^2 - d_2^2) + d_2^2 \right)\ltwo{\nu}^2 \\
&= d_1^2 \ltwo{\nu}^2.
\end{align*}

\subsection{Proof of Lemma \ref{lem:uec}}\label{sec:proof:uec}
Let $\nu \in \mathbb{R}^p$ be expressed as $\nu = \sum_{j=1}^p a_j v_j$, where the $v_j$'s are the columns of $\bV$. Using the formulas in \eqref{eqn:coordinates},

\begin{align*}
\nu^T (\widetilde{\bX}^T \widetilde{\bX}) \nu &= \nu^T (\bX^T \bX + n\theta \bA) \nu \\
&= \sum_{j=1}^p d_j^2 a_j^2 + \sum_{j=2}^p n\theta(d_1^2 - d_j^2)a_j^2 \\
&= \sum_{j=1}^p \left[ n\theta d_1^2 + (1-n\theta) d_j^2 \right] a_j^2 \\
&\geq \min_j \left[ n\theta d_1^2 + (1-n\theta) d_j^2 \right] \sum_{j=1}^p a_j^2 \\
&\begin{cases} \geq n\theta d_1^2 \ltwo{\nu}^2 &\text{if } n\theta \leq 1, \\ \geq d_1^2\ltwo{\nu}^2 &\text{if } n\theta \geq 1  \end{cases} \\
&= \min (n\theta, 1) d_1^2 \ltwo{\nu}^2.
\end{align*}

\subsection{Proof of Theorem \ref{thm:l2bound-c}}\label{proof:thm:l2bound-c}
Let $\widehat{\nu} = \widehat{\beta} - \beta^*$. By definition of $\widehat{\beta}$ and using the relation $\widetilde{\by} = \widetilde{\bX}\beta^* + \widetilde{\bw}$,
\begin{align}
\ltwo{\widetilde{\by} - \widetilde{\bX} \widehat{\beta} }^2 &\leq \ltwo{\widetilde{y} - \widetilde{\bX} \beta^* }^2, \nonumber \\
\ltwo{\widetilde{\bw} - \widetilde{\bX} \widehat{\nu}}^2 &\leq \ltwo{\widetilde{\bw}}^2, \nonumber \\
\widehat{\nu}^T (\bX^T \bX + n\theta \bA) \widehat{\nu} &\leq 2(\widetilde{\bX}^T \widetilde{\bw})^T \widehat{\nu}. \label{eqn:l2bound-proof1}
\end{align}

By Lemma \ref{lem:betterbound2}, the LHS is $\geq n \eta \ltwo{\widehat{\nu}}^2$, where $n\eta = \min \left[ (1 - n\theta)n\gamma + n\theta d_1^2, d_1^2 \right]$. We can get an upper bound on the RHS:
\begin{align*}
(\widetilde{\bX}^T \widetilde{\bw})^T \widehat{\nu} &\leq \linf{\widetilde{\bX}^T \widetilde{\bw}} \lone{\widehat{\nu}} &\text{(H\"{o}lder's inequality)} \\
&= \linf{\bX^T \bw - n\theta \bA \beta^*} (\lone{\widehat{\nu}_{S}} + \lone{\widehat{\nu}_{S^c}}) \\
&\leq 2 \linf{\bX^T \bw - n\theta \bA \beta^*} \lone{\widehat{\nu}_{S}} &\text{(Lemma \ref{lem:cpcLasso_rec})} \\
&\leq 2 \linf{\bX^T \bw - n\theta \bA \beta^*} \sqrt{|S|} \ltwo{\widehat{\nu}}. &\text{(Cauchy-Schwarz)}
\end{align*}

Hence,
\begin{align*}
n\eta \ltwo{\widehat{\nu}}^2 &\leq 4 \linf{\bX^T \bw - n\theta A \beta^*} \sqrt{|S|} \ltwo{\widehat{\nu}}, \\
\ltwo{\widehat{\nu}} &\leq \frac{4 \sqrt{|S|} \linf{\bX^T \bw - n\theta \bA \beta^*} }{n\eta}.
\end{align*}

The second inequality in \eqref{eqn:l2bound} is the result of using the triangle inequality on the $\linf{\cdot}$ norm.

When $\beta^*$ is aligned with the first principal component of $\bX$, $\bA \beta^* = 0$. Hence, the first inequality in \eqref{eqn:l2bound} reduces to \eqref{eqn:l2bound-1pc}.

\subsection{Proof of Theorem \ref{thm:l2bound-c2}}\label{proof:thm:l2bound-c2}
The proof of Theorem \ref{thm:l2bound-c2} is exactly the same as that for Theorem \ref{thm:l2bound-c}, except that we bound the LHS of \eqref{eqn:l2bound-proof1} from below by $\min \left( n\theta, 1 \right) d_1^2 \ltwo{\widehat{\nu}}^2$. We can do this due to Lemma \ref{lem:uec}.

\subsection{Proof of Theorem \ref{thm:l2bound-l}}\label{proof:thm:l2bound-l}
As in the proof of Lemma \ref{lem:lpcLasso_rec}, we have \eqref{eqn:basic-pcLasso2}:
\begin{equation*}
\frac{1}{2n}\widehat{\nu}^T (\bX^T \bX + n\theta \bA) \widehat{\nu} \leq \frac{\widetilde{\bw}^T \widetilde{\bX} \widehat{\nu}}{n} + \lambda (\lone{\widehat{\nu}_S} - \lone{\widehat{\nu}_{S^c}}).
\end{equation*}

By Lemma \ref{lem:betterbound2}, the LHS is $\geq \frac{\eta}{2} \ltwo{\widehat{\nu}}^2$, where $n\eta = \min \left[ (1 - n\theta)n\gamma + n\theta d_1^2, d_1^2 \right]$. By our choice of $\lambda$, we may use \eqref{eqn:basic-pcLasso3} to obtain
\begin{align*}
\frac{\eta}{2} \ltwo{\widehat{\nu}}^2 &\leq \frac{1}{2}\lambda \lone{\widehat{\nu}} + \lambda (\lone{\widehat{\nu}_S} - \lone{\widehat{\nu}_{S^c}}) \\
&\leq \frac{3\lambda}{2}\lone{\widehat{\nu}_S} \leq \frac{3\lambda}{2} \sqrt{|S|} \ltwo{\widehat{\nu}}, \\
\ltwo{\widehat{\nu}} &\leq \frac{3\lambda \sqrt{|S|}}{\eta},
\end{align*}
as required.

\subsection{Proof of Theorem \ref{thm:predbound-lslow}}\label{proof:thm:predbound-lslow}
By \eqref{eqn:basic-pcLasso1},
\begin{align*}
0 \leq \frac{1}{2n}\widehat{\nu}^T (\bX^T \bX + n\theta \bA) \widehat{\nu} &\leq \frac{\widetilde{\bw}^T \widetilde{\bX} \widehat{\nu}}{n} + \lambda (\lone{\beta^*} - \lone{\beta^* + \widehat{\nu}}) \\
&\leq \frac{1}{n} \linf{\widetilde{\bX}^T \widetilde{\bw}} \lone{\widehat{\nu}} + \lambda (\lone{\beta^*} - \lone{\beta^* + \widehat{\nu}}) &\text{(H\"{o}lder's inequality)} \\
&\leq \left(\frac{1}{n} \linf{\bX^T \bw - n\theta \bA \beta^*} - \lambda \right)\lone{\widehat{\nu}} + 2 \lambda \lone{\beta^*} &\text{(triangle inequality)} \\
&\leq -\frac{\lambda}{2} \lone{\widehat{\nu}} + 2 \lambda \lone{\beta^*}, &\text{(by choice of $\lambda$)},
\end{align*}
so $\lone{\widehat{\nu}} \leq 4 \lone{\beta^*}$. Using \eqref{eqn:basic-pcLasso1} again,
\begin{align*}
\frac{1}{2n}\ltwo{\bX(\widehat{\beta} - \beta^*)}^2 = \frac{1}{2n}\widehat{\nu}^T \bX^T \bX \widehat{\nu} &\leq \frac{1}{2n} \widehat{\nu}^T (\bX^T \bX + n\theta \bA) \widehat{\nu} \\
&\leq \frac{1}{n}\linf{\widetilde{\bX}^T \widetilde{\bw}} \lone{\widehat{\nu}} + \lambda \lone{\widehat{\nu}} \\
&\leq \frac{3\lambda}{2}\lone{\widehat{\nu}} \\
&\leq 6 \lambda \lone{\beta^*}.
\end{align*}
Multiplying both sides by $2$ establishes the claim.

\subsection{Proof of Theorem \ref{thm:predbound-lslow2}}\label{proof:thm:predbound-lslow2}
Let $\kappa = \min (n\theta, 1) d_1^2$. Mimicking the first part of the proof of Theorem \ref{thm:predbound-lslow} and using Lemma \ref{lem:uec}, we have
\begin{align*}
-\frac{\lambda}{2} \lone{\widehat{\nu}} + 2 \lambda \lone{\beta^*} &\geq \frac{1}{2n}\widehat{\nu}^T (\widetilde{\bX}^T \widetilde{\bX}) \widehat{\nu} \geq \kappa \ltwo{\widehat\nu}^2 \\
&\geq \kappa \frac{\lone{\widehat{\nu}}^2}{p}, &(\text{Cauchy-Schwarz}), \\
\frac{\kappa}{p}\lone{\widehat{\nu}}^2 + \frac{\lambda}{2} \lone{\widehat{\nu}} - 2 \lambda \lone{\beta^*} &\leq 0, \\
\lone{\widehat{\nu}} &\leq \frac{-\lambda/2 + \sqrt{\lambda^2 / 4 + 8\lambda \lone{\beta^*} \kappa / p }}{2\kappa / p} \\
&= \frac{-\lambda p + \sqrt{\lambda^2 p^2 + 32\lambda \lone{\beta^*} \kappa p }}{4\kappa}.
\end{align*}

Thus, mimicking the second part of the proof of Theorem \ref{thm:predbound-lslow},
\begin{align*}
\frac{1}{2n}\ltwo{\bX(\widehat{\beta} - \beta^*)}^2 &\leq \frac{3\lambda}{2}\lone{\widehat{\nu}} \\
&\leq \frac{3\lambda}{2} \frac{-\lambda p + \sqrt{\lambda^2 p^2 + 32\lambda \lone{\beta^*} \kappa p }}{4\kappa}, \\
\frac{1}{n}\ltwo{\bX(\widehat{\beta} - \beta^*)}^2 &\leq \frac{ 3\lambda (-\lambda p + \sqrt{\lambda^2 p^2 + 32\lambda \lone{\beta^*} \kappa p} )}{4\kappa},
\end{align*}
as required. To show that the RHS above is indeed a better bound than that in Theorem \ref{thm:predbound-lslow}:
\begin{align*}
&&12 \lambda \lone{\beta^*} &\geq \frac{ 3\lambda (-\lambda p + \sqrt{\lambda^2 p^2 + 32\lambda \lone{\beta^*} \kappa p} )}{4\kappa} \\
\Leftrightarrow \qquad && 16 \kappa \lone{\beta^*} &\geq -\lambda p + \sqrt{\lambda^2 p^2 + 32\lambda \lone{\beta^*} \kappa p} \\
\Leftrightarrow \qquad && (\lambda p + 16 \kappa \lone{\beta^*})^2 &\geq \lambda^2 p^2 + 32\lambda \lone{\beta^*} \kappa p \\
\Leftrightarrow \qquad && 256 \kappa^2 \lone{\beta^*}^2 &\geq 0,
\end{align*}
which is obviously true.

\subsection{Proof of Theorem \ref{thm:predbound-lfast}}\label{proof:thm:predbound-lfast}
We have
\begin{align*}
\frac{1}{2n}\widehat{\nu}^T (\bX^T \bX + n\theta A) \widehat{\nu} &\leq \frac{\widetilde{\bw}^T \widetilde{\bX} \widehat{\nu}}{n} + \lambda (\lone{\widehat{\nu}_S} - \lone{\widehat{\nu}_{S^c}}) &\text{(\eqref{eqn:basic-pcLasso2} in Lemma \ref{lem:lpcLasso_rec})} \\
&\leq \frac{1}{n} \linf{\widetilde{\bX}^T \widetilde{\bw}} \lone{\widehat{\nu}} + \lambda \lone{\widehat{\nu}} &\text{(H\"{o}lder's inequality)} \\
&\leq \frac{3 \lambda}{2} \lone{\widehat{\nu}} &\text{(choice of $\lambda$)} \\
&\leq \frac{3 \lambda}{2} \sqrt{|S|} \ltwo{\widehat{\nu}}. &\text{(Cauchy-Schwarz inequality)}
\end{align*}

By Lemma \ref{lem:lpcLasso_rec}, $\widehat{\nu}$ lies in the set $\{ \nu \in \bbR^p: \lone{\nu_{S^c}} \leq 3\lone{\nu_{S}} \}$, and so the restricted eigenvalue condition and Lemma \ref{lem:betterbound2} apply, i.e. $\ltwo{\widehat{\nu}}^2 \leq \frac{1}{n\eta} \ltwo{\widetilde{\bX}\widehat{\nu}}^2$, where $n \eta = \min \left[ (1 - n\theta)n\gamma + n\theta d_1^2, d_1^2 \right]$. Hence,
\begin{align*}
\frac{1}{2n}\widehat{\nu}^T (\bX^T \bX + n\theta \bA)\widehat{\nu} &\leq \frac{3 \lambda}{2} \sqrt{|S|} \ltwo{\widehat{\nu}} \leq \frac{3 \lambda}{2} \sqrt{|S|} \sqrt{\frac{1}{n\eta}} \ltwo{\widetilde{\bX}\widehat{\nu}}, \\
\ltwo{\bX\widehat{\nu}} \leq \ltwo{\widetilde{\bX}\widehat{\nu}} &\leq 3n \lambda \sqrt{|S|} \sqrt{\frac{1}{n\eta}}, \\
\frac{\ltwo{\bX\widehat{\nu}}^2}{n} &\leq \frac{9|S| \lambda^2}{\eta},
\end{align*}
as required.


\section{Full details of simulation study in Section \ref{sec:sim}}\label{sec:simstudy-detail}

In this appendix, we present details on how data was generated for the simulation study, as well as results that we obtained for a variety of settings.

\subsection{Data generation details}

Let $n$ be the size of the training data, $p$ be the number of features for each observation, and $K$ be the number of groups of features. Let $size$ be a vector of length $K$ such that its $k$th element denotes the number of features in group $k$.

\subsubsection{Generating $\bX$}

Let $\bX_k$ denote the design matrix for group $k$, and let $\bX$ denote the full design matrix.

\begin{enumerate}
    \item Let $\rho$ denote the correlation between features in each group. For each group $k$, build the covariance matrix $\Sigma_k = \rho {\bf 1} + (1-\rho) \bI \in \mathbb{R}^{size_k \times size_k}$. 
    
    \item Generate the rows of $\bX_k$ independently from $\mathcal{N}(0, \Sigma_k)$.
    
\end{enumerate}

\subsubsection{Generating $\by$}
\begin{enumerate}
    \item Perform singular value decomposition (SVD) on each $\bX_k$ to get $\bU_k$, $\bD_k$ and $\bV_k$.

    \item Let $n_{ev}$ be the number of principal components (PCs) to take from each group, and let $SNR$ be the desired signal-to-noise ratio (SNR).
    
    \item We consider three settings for determining which PCs will form the signal:
    
    \begin{itemize}
        \item \textit{``Home court"}: Set $\bW_k$ to be the first $n_{ev}$ columns of $\bV_k$.
        
         \item \textit{``Neutral court"}: Set $\bW_k$ to be $n_{ev}$ random columns of $\bV_k$.
         
         \item \textit{``Hostile court"}: Set $\bW_k$ to be the bottom $n_{ev}$ columns of $\bV_k$.
    \end{itemize}
    
    \item For each $k$, specify coefficients ${\bf b}_k \in \mathbb{R}^{n_{ev}}$. For simplicity, we set ${\bf b}_k = {\bf 2}$ if the group is related to the response, ${\bf b}_k = {\bf 0}$ otherwise.
    
    \item Set $signal = \sum_k \bX_k \bW_k {\bf b}_k$.
    
    \item Set response $\by = signal + \epsilon$, where the entries of $\epsilon$ are iid $\mathcal{N}(0, V)$, where
    \[ V = \frac{\text{Var}(signal)}{SNR} = \frac{\sum_k {\bf b}_k^T \bW_k^T \Sigma_k \bW_k {\bf b}_k }{SNR}. \]
    
\end{enumerate}

\subsubsection{Generating test data}
\begin{enumerate}
    \item For each $k$, generate $\bX_{test, k}$ such that each row is i.i.d. $\mathcal{N}({\bf 0}, \Sigma_k)$. Generate the $\bX_{test, k}$ independently of each other.
    
    \item Compute $signal_{test} = \sum_k \bX_{test, k} \bW_k {\bf b}_k$.
    
\end{enumerate}

\subsection{Simulation details}

We ran simulations for the following settings:
\begin{itemize}
\item $n = 200$, $p = 50$, no groups, $n_{ev} = 5$: home court, neutral court and hostile court.

\item $n = 200$, $p = 50$, 5 groups of 10 predictors each, $n_{ev} = 1$, signal only depends on the first two groups: home court, neutral court and hostile court.

\item $n = 200$, $p = 200$, 10 groups of 20 predictors each, $n_{ev} = 2$, signal only depends on the first group: home court and neutral court.

\item $n = 200$, $p = 1,000$, 10 groups of 100 predictors each, $n_{ev} = 2$, signal only depends on the first group: home court and neutral court.
\end{itemize}

For each setting above, we ran 30 simulations for each combination of within-group correlation $\rho \in \{ 0, 0.3\}$ and SNR $SNR \in \{0.5, 1, 2\}$.

We compared the following models:
\begin{itemize}
	\item The null model, i.e. the mean of the responses in the training data.
	
	\item pcLasso with CV across $rat = 0.25, 0.5, 0.75, 0.9, 0.95, 1$, with $\lambda$ values selected both by \texttt{lambda.min} and \texttt{lambda.1se}.	
	
	\item Elastic net with CV across $\alpha = 0, 0.2, 0.4, 0.6, 0.8$ and $1$, with $\lambda$ values selected both by \texttt{lambda.min} and \texttt{lambda.1se}.

	\item Lasso with CV, with $\lambda$ values selected both by \texttt{lambda.min} and \texttt{lambda.1se}.
	
	\item Principal components (PC) regression with CV across ranks.
\end{itemize}

For each method, we present a boxplot of the test mean-squared error (MSE) on $5,000$ test points, i.e. $MSE = \mathbb{E}[(\widehat{y}_{test} - signal_{test})^2]$. We also present boxplots of the support size of the model which the method selects. Finally, we present a line histogram of the values of $rat$ which the two versions of pcLasso select during model fitting.

\newpage
\subsubsection{Small $p$ no groups, ``home court" for pcLasso}
$n = 200$, $p = 50$, 1 group of uncorrelated predictors, response related to top 5 eigenvectors.

\begin{figure}[ht]
\includegraphics[width=2.2in]{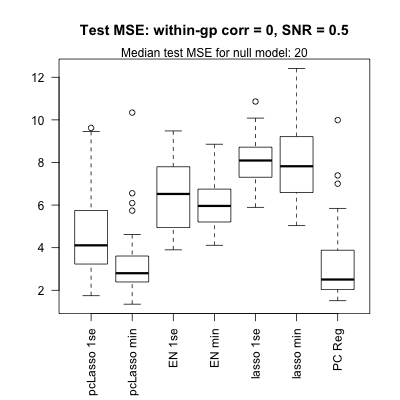}\includegraphics[width=2.2in]{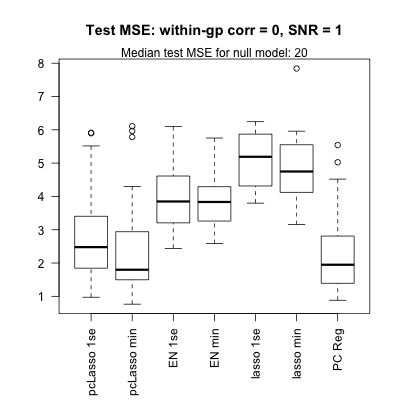}\includegraphics[width=2.2in]{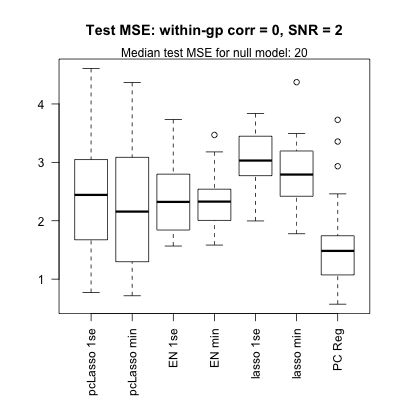}

\includegraphics[width=2.2in]{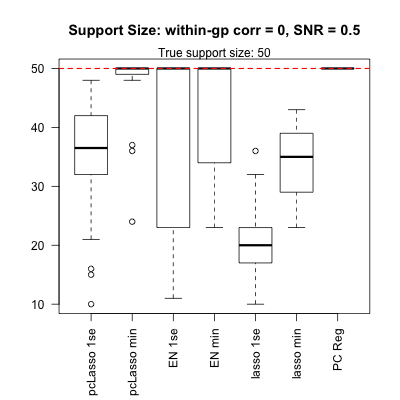}\includegraphics[width=2.2in]{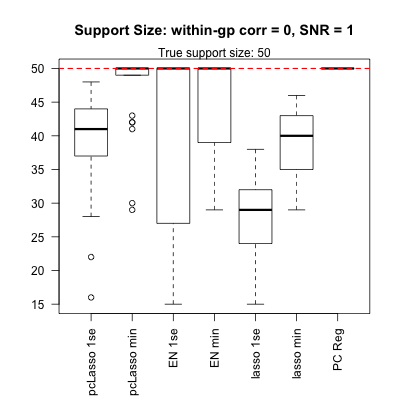}\includegraphics[width=2.2in]{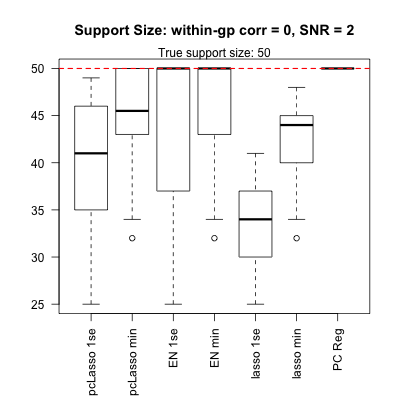}

\includegraphics[width=2.2in]{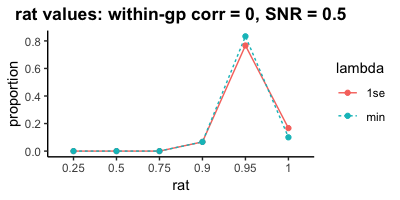}\includegraphics[width=2.2in]{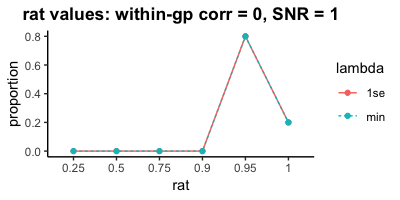}\includegraphics[width=2.2in]{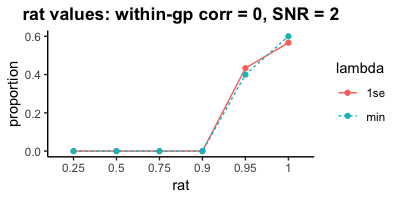}
\end{figure}

\newpage
As above, but with predictors having pairwise correlation of 0.3.

\begin{figure}[ht]
\includegraphics[width=2.2in]{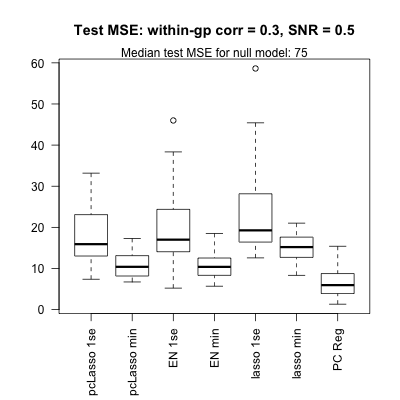}\includegraphics[width=2.2in]{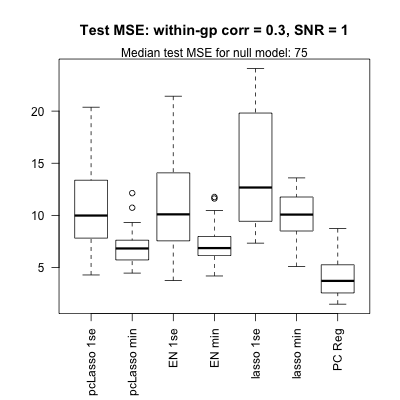}\includegraphics[width=2.2in]{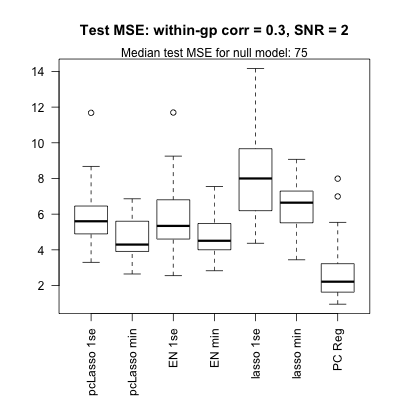}

\includegraphics[width=2.2in]{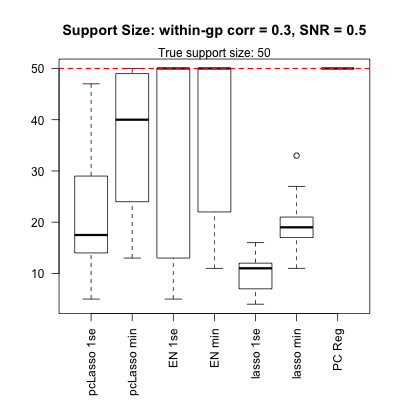}\includegraphics[width=2.2in]{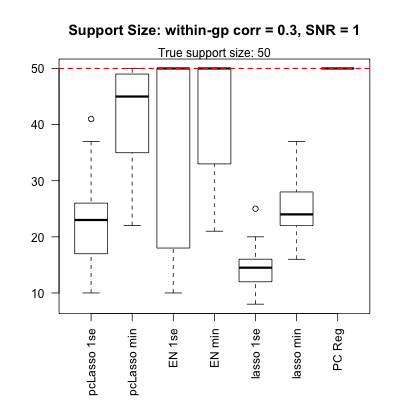}\includegraphics[width=2.2in]{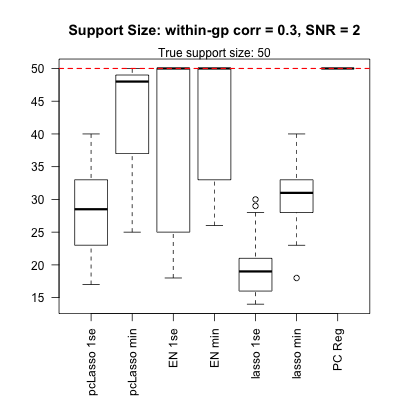}

\includegraphics[width=2.2in]{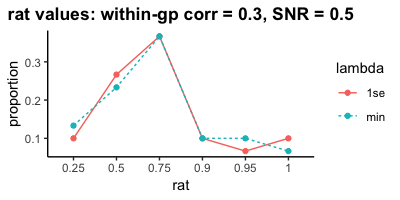}\includegraphics[width=2.2in]{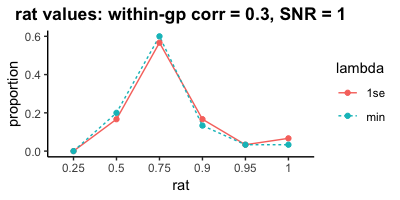}\includegraphics[width=2.2in]{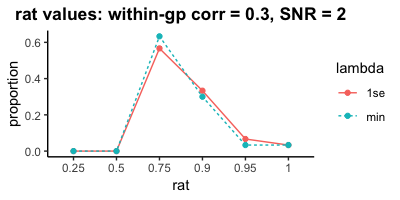}
\end{figure}

\newpage
\subsubsection{Small $p$ no groups, ``neutral court" for pcLasso}
$n = 200$, $p = 50$, 1 group of uncorrelated predictors, response related to 5 random eigenvectors.

\begin{figure}[ht]
\includegraphics[width=2.2in]{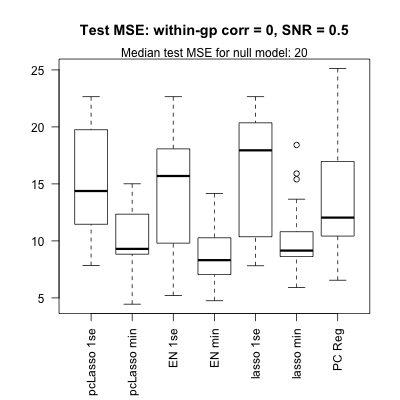}\includegraphics[width=2.2in]{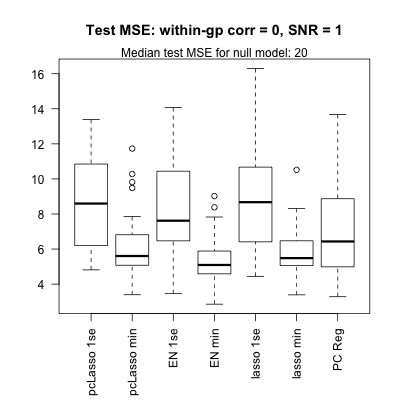}\includegraphics[width=2.2in]{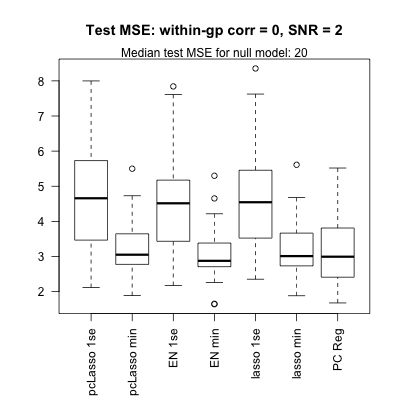}

\includegraphics[width=2.2in]{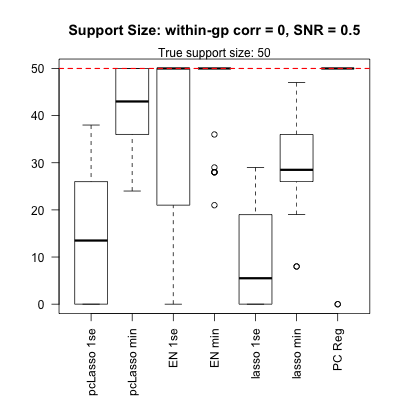}\includegraphics[width=2.2in]{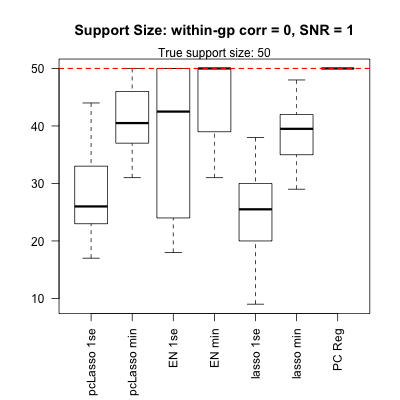}\includegraphics[width=2.2in]{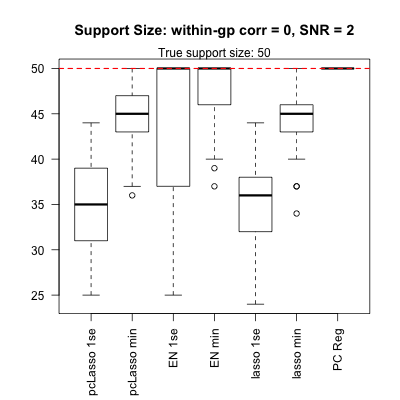}

\includegraphics[width=2.2in]{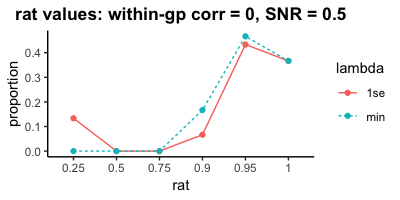}\includegraphics[width=2.2in]{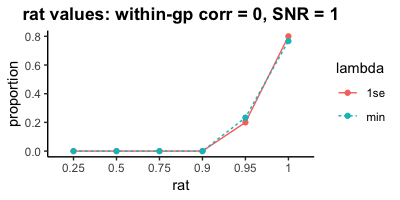}\includegraphics[width=2.2in]{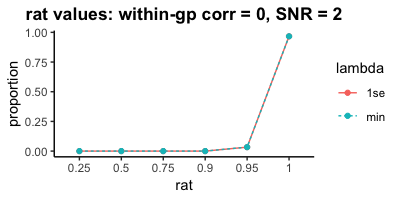}
\end{figure}

\newpage
As above, but with predictors having pairwise correlation of 0.3.

\begin{figure}[ht]
\includegraphics[width=2.2in]{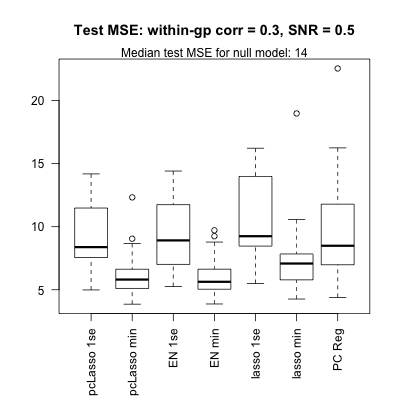}\includegraphics[width=2.2in]{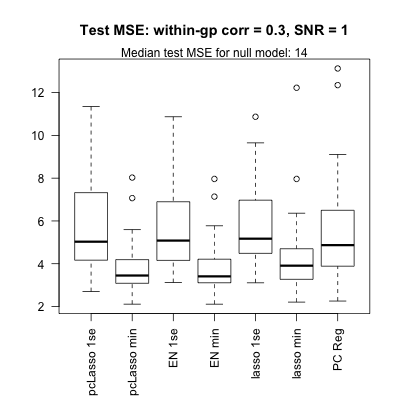}\includegraphics[width=2.2in]{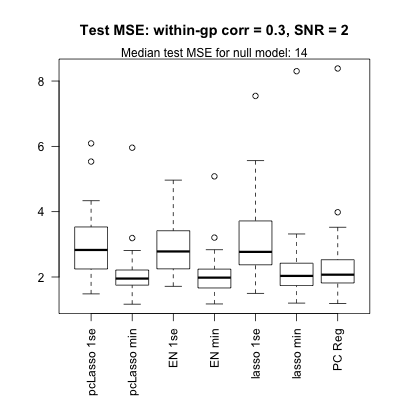}

\includegraphics[width=2.2in]{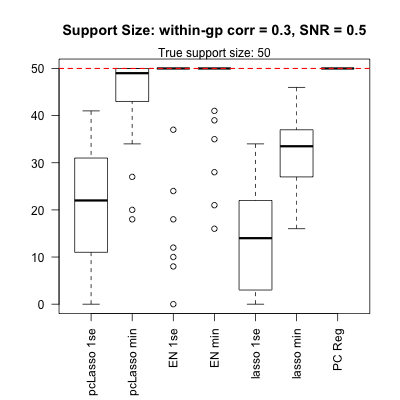}\includegraphics[width=2.2in]{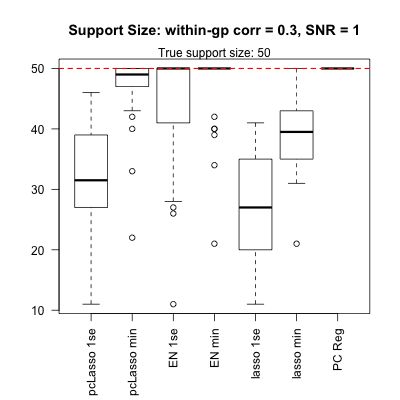}\includegraphics[width=2.2in]{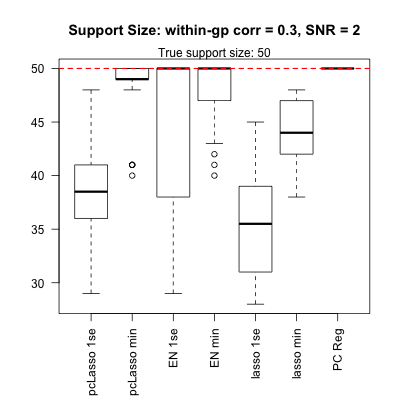}

\includegraphics[width=2.2in]{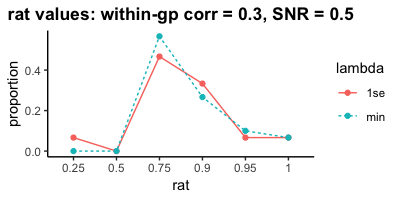}\includegraphics[width=2.2in]{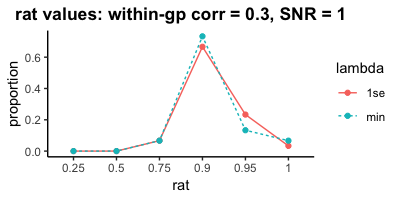}\includegraphics[width=2.2in]{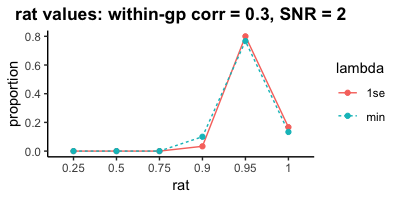}
\end{figure}

\newpage
\subsubsection{Small $p$ no groups, ``hostile court" for pcLasso}
$n = 200$, $p = 50$, 1 group of uncorrelated predictors, response related to the bottom 5 eigenvectors.

\begin{figure}[ht]
\includegraphics[width=2.2in]{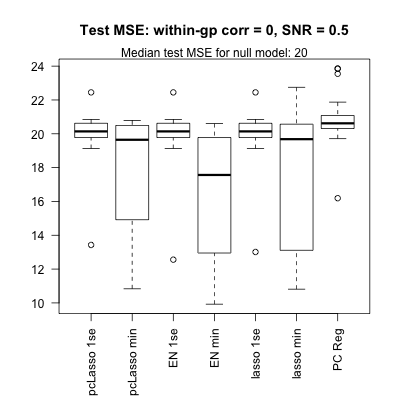}\includegraphics[width=2.2in]{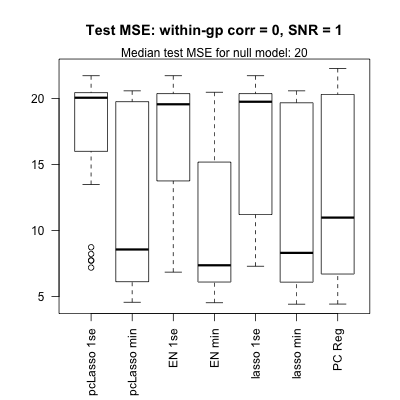}\includegraphics[width=2.2in]{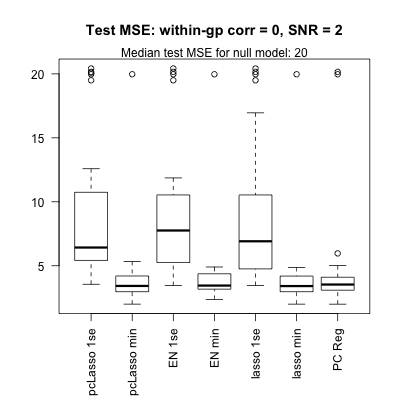}

\includegraphics[width=2.2in]{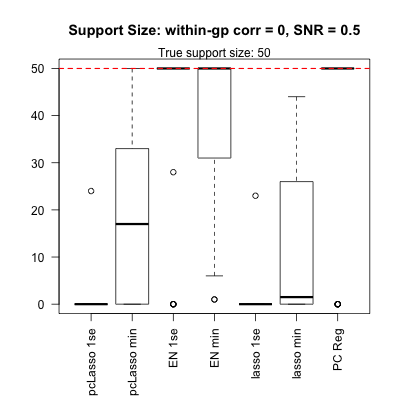}\includegraphics[width=2.2in]{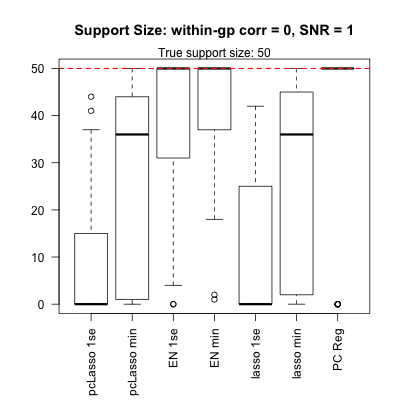}\includegraphics[width=2.2in]{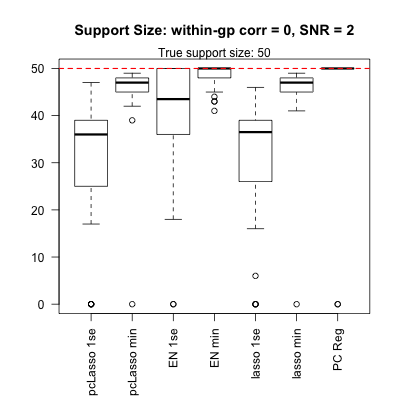}

\includegraphics[width=2.2in]{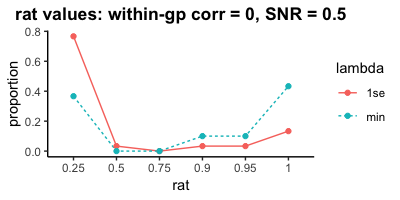}\includegraphics[width=2.2in]{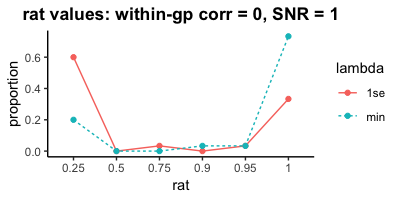}\includegraphics[width=2.2in]{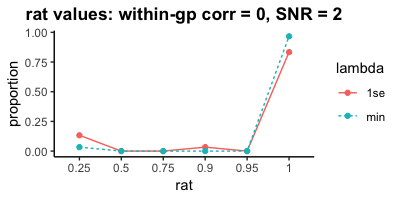}
\end{figure}

\newpage
As above, but with predictors having pairwise correlation of 0.3.

\begin{figure}[ht]
\includegraphics[width=2.2in]{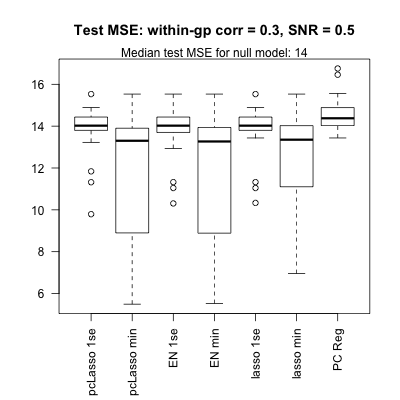}\includegraphics[width=2.2in]{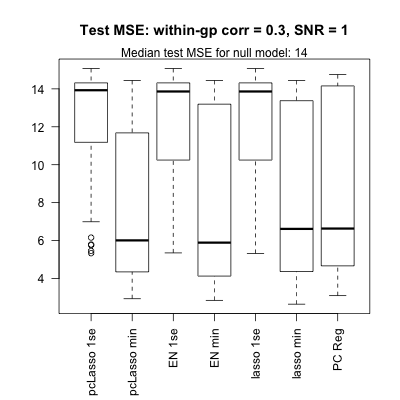}\includegraphics[width=2.2in]{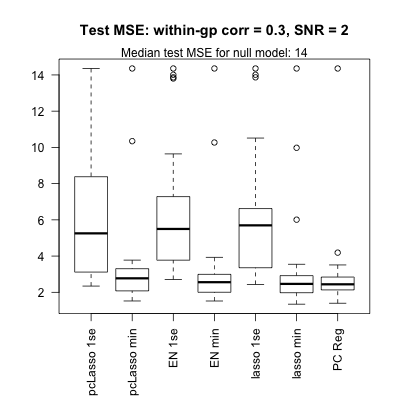}

\includegraphics[width=2.2in]{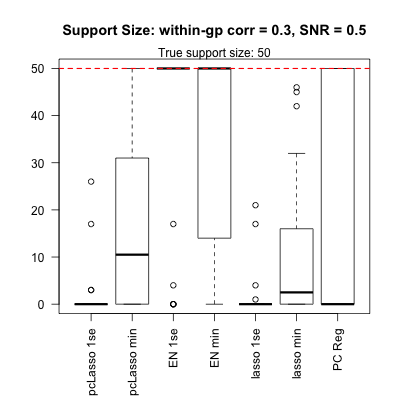}\includegraphics[width=2.2in]{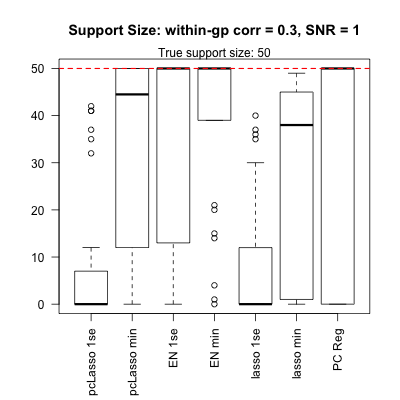}\includegraphics[width=2.2in]{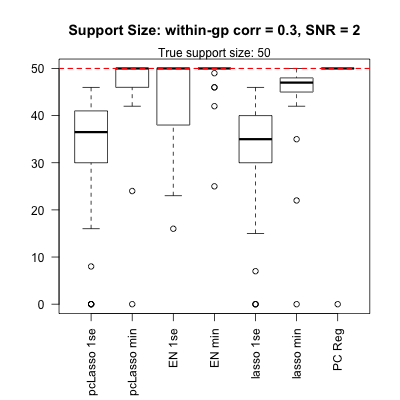}

\includegraphics[width=2.2in]{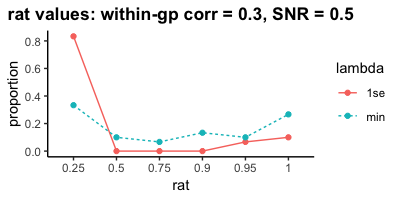}\includegraphics[width=2.2in]{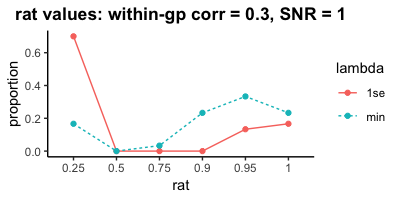}\includegraphics[width=2.2in]{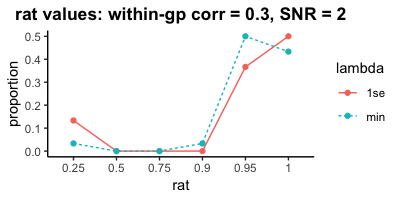}
\end{figure}

\newpage
\subsubsection{Small $p$ with groups, ``home court" for pcLasso}
$n = 200$, $p = 50$, 5 groups of 10 uncorrelated predictors, response related to the top eigenvector in the first 2 groups.

\begin{figure}[ht]
\includegraphics[width=2.2in]{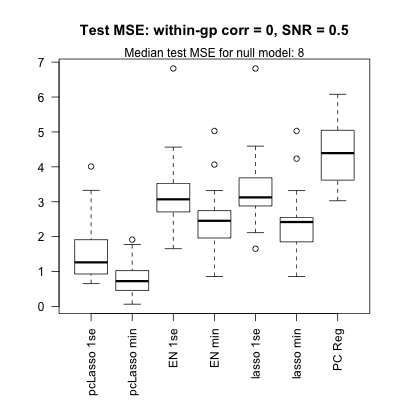}\includegraphics[width=2.2in]{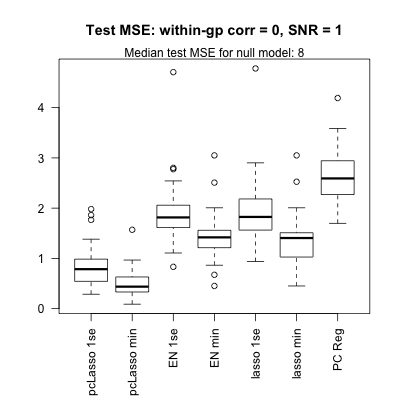}\includegraphics[width=2.2in]{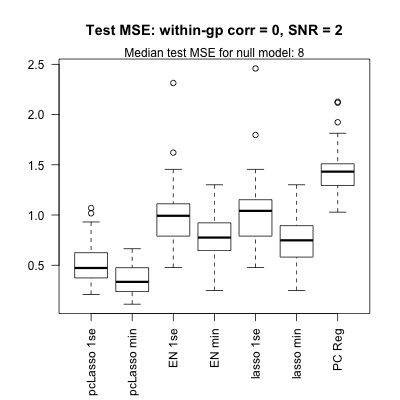}

\includegraphics[width=2.2in]{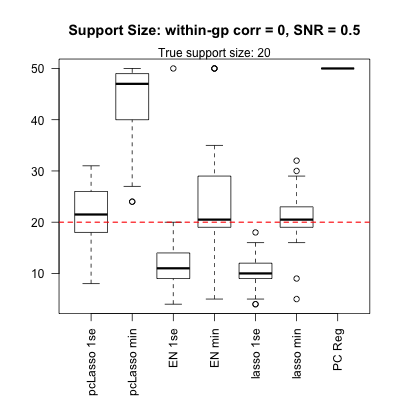}\includegraphics[width=2.2in]{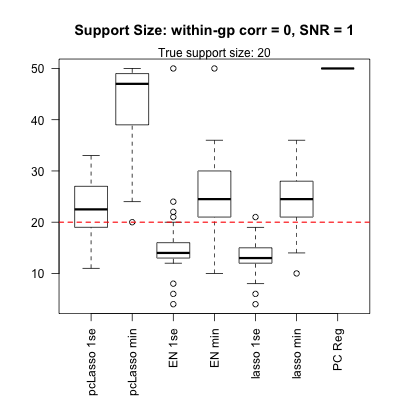}\includegraphics[width=2.2in]{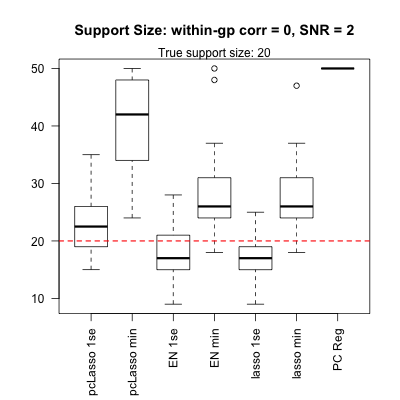}

\includegraphics[width=2.2in]{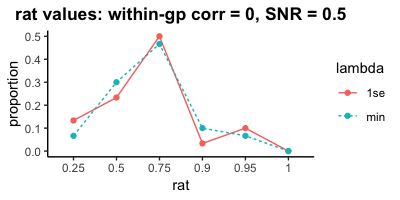}\includegraphics[width=2.2in]{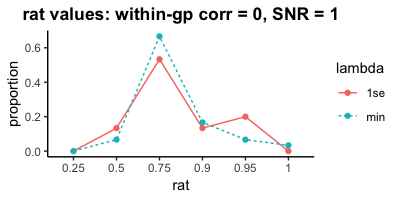}\includegraphics[width=2.2in]{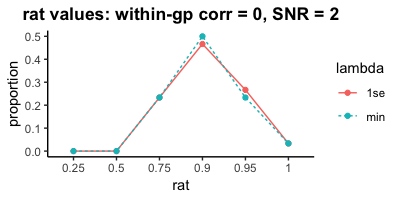}
\end{figure}

\newpage
As above, but with predictors in the same group having pairwise correlation of 0.3.

\begin{figure}[ht]
\includegraphics[width=2.2in]{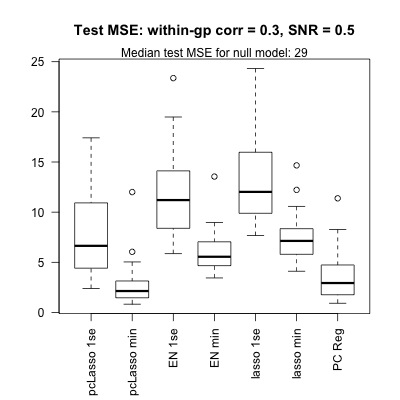}\includegraphics[width=2.2in]{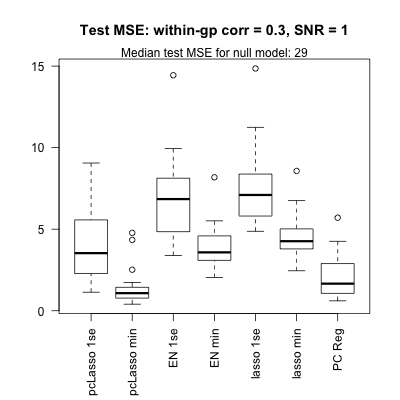}\includegraphics[width=2.2in]{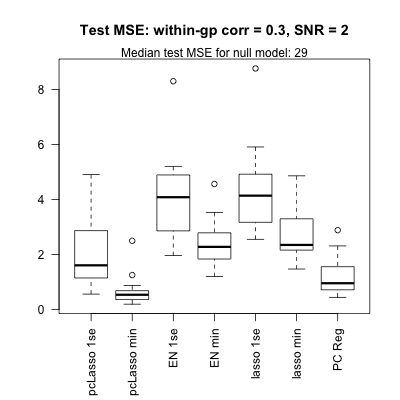}

\includegraphics[width=2.2in]{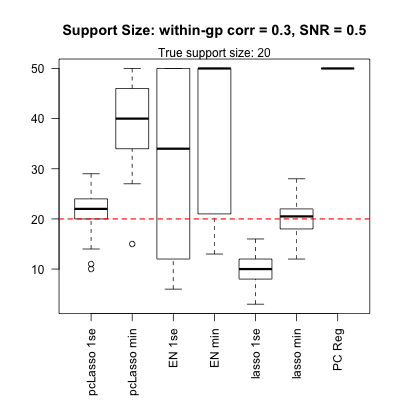}\includegraphics[width=2.2in]{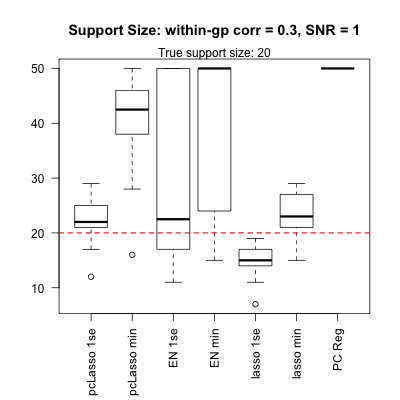}\includegraphics[width=2.2in]{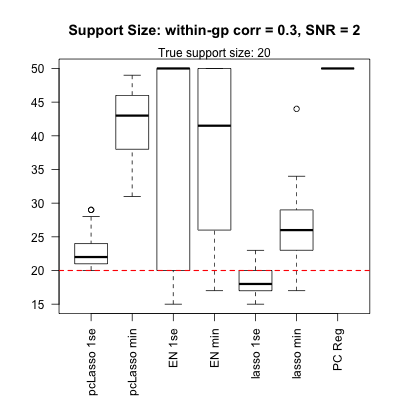}

\includegraphics[width=2.2in]{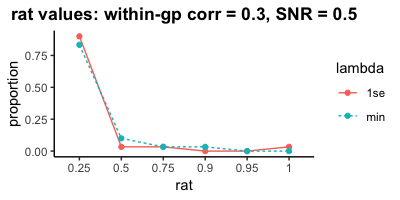}\includegraphics[width=2.2in]{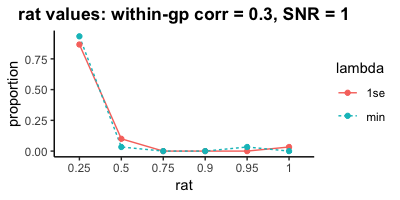}\includegraphics[width=2.2in]{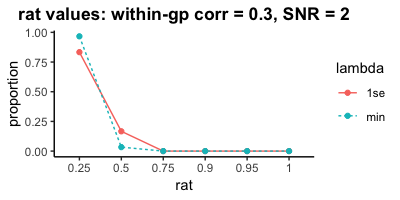}
\end{figure}

\newpage
\subsubsection{Small $p$ with groups, ``neutral court" for pcLasso}
$n = 200$, $p = 50$, 5 groups of 10 uncorrelated predictors, response related to a random eigenvector in the first 2 groups.

\begin{figure}[ht]
\includegraphics[width=2.2in]{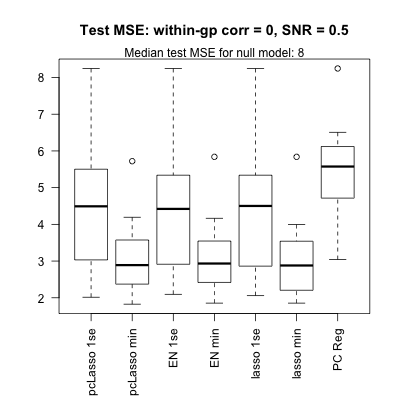}\includegraphics[width=2.2in]{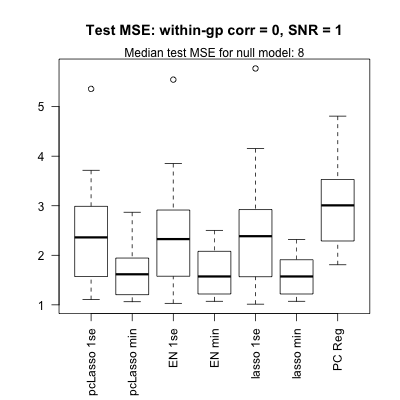}\includegraphics[width=2.2in]{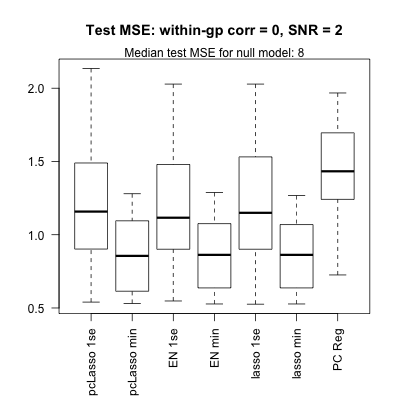}

\includegraphics[width=2.2in]{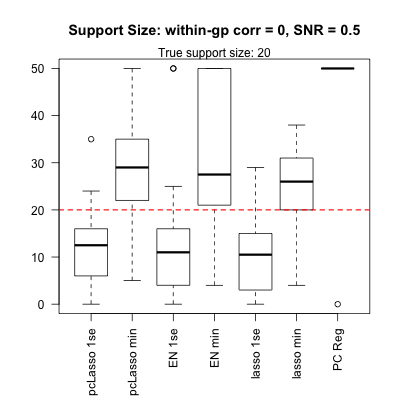}\includegraphics[width=2.2in]{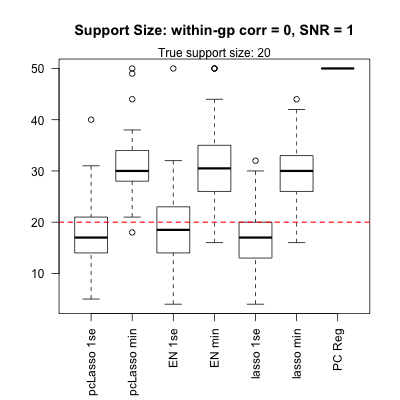}\includegraphics[width=2.2in]{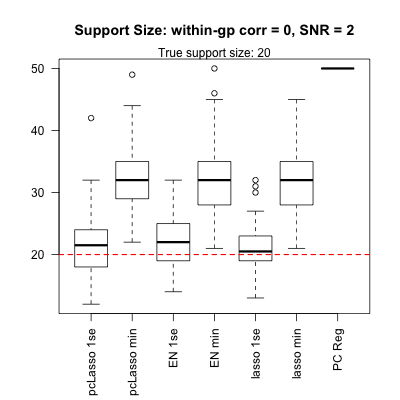}

\includegraphics[width=2.2in]{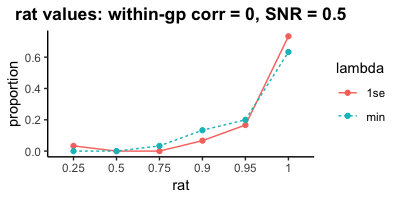}\includegraphics[width=2.2in]{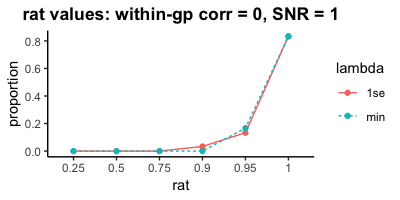}\includegraphics[width=2.2in]{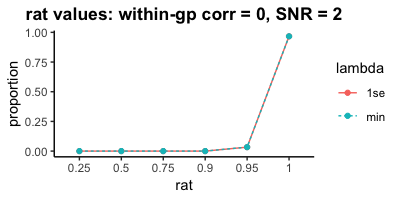}
\end{figure}

\newpage
As above, but with predictors in the same group having pairwise correlation of 0.3.

\begin{figure}[ht]
\includegraphics[width=2.2in]{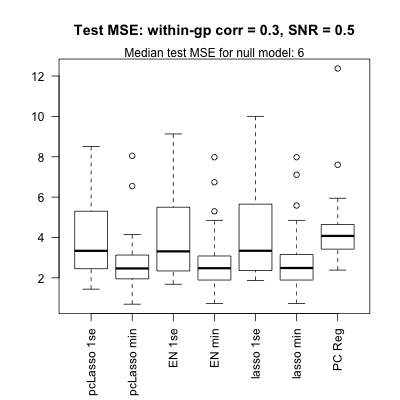}\includegraphics[width=2.2in]{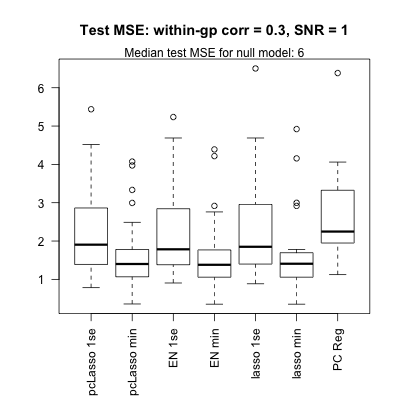}\includegraphics[width=2.2in]{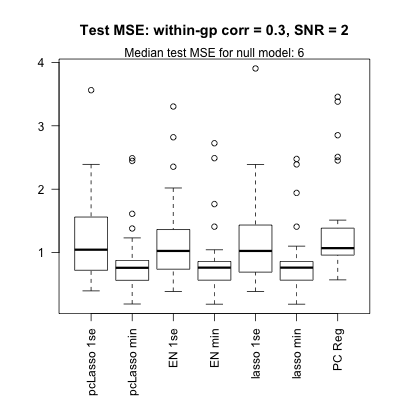}

\includegraphics[width=2.2in]{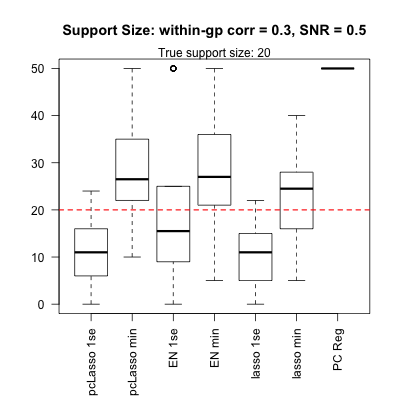}\includegraphics[width=2.2in]{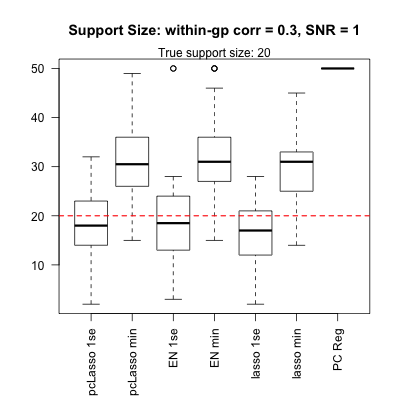}\includegraphics[width=2.2in]{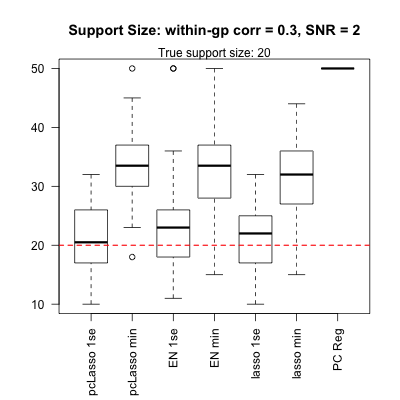}

\includegraphics[width=2.2in]{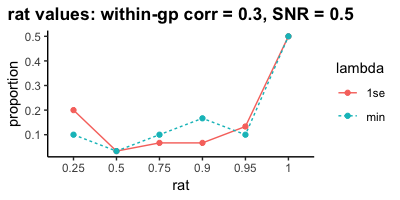}\includegraphics[width=2.2in]{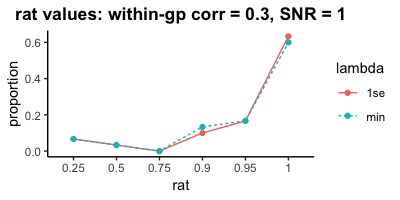}\includegraphics[width=2.2in]{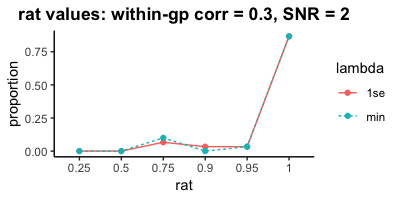}
\end{figure}

\newpage
\subsubsection{Small $p$ with groups, ``hostile court" for pcLasso}
$n = 200$, $p = 50$, 5 groups of 10 uncorrelated predictors, response related to the bottom eigenvector in the first 2 groups.

\begin{figure}[ht]
\includegraphics[width=2.2in]{6_MSE_corr0_SNR0_5.png}\includegraphics[width=2.2in]{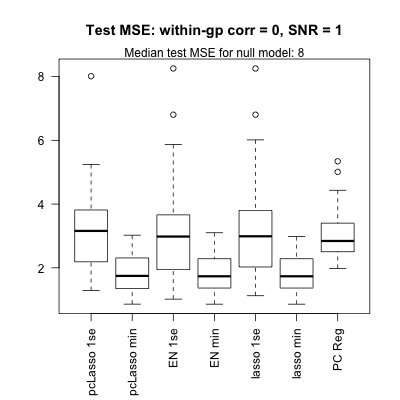}\includegraphics[width=2.2in]{6_MSE_corr0_SNR2.png}

\includegraphics[width=2.2in]{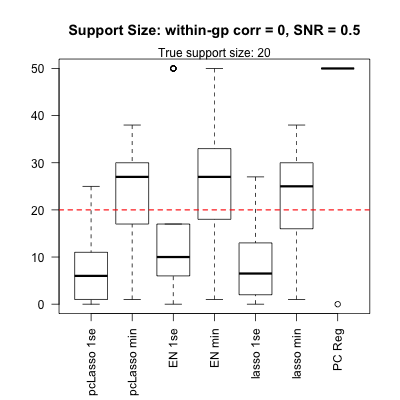}\includegraphics[width=2.2in]{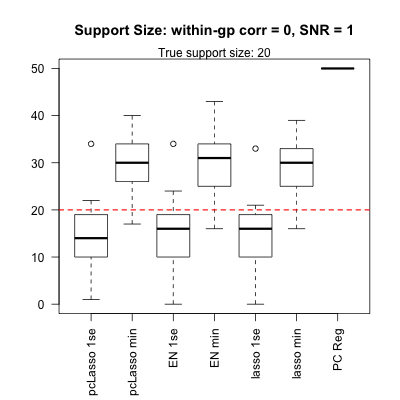}\includegraphics[width=2.2in]{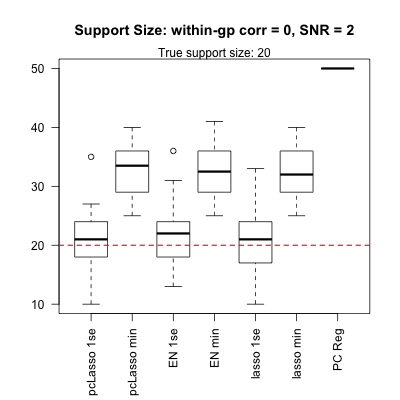}

\includegraphics[width=2.2in]{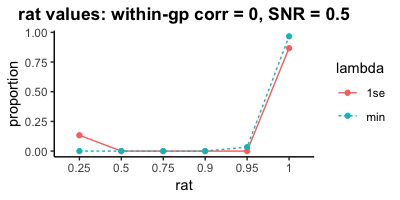}\includegraphics[width=2.2in]{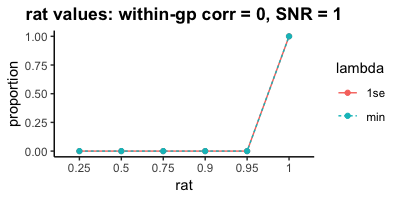}\includegraphics[width=2.2in]{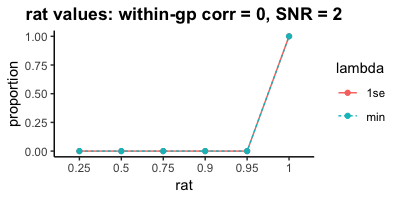}
\end{figure}

\newpage
As above, but with predictors in the same group having pairwise correlation of 0.3.

\begin{figure}[ht]
\includegraphics[width=2.2in]{6_MSE_corr0_3_SNR0_5.png}\includegraphics[width=2.2in]{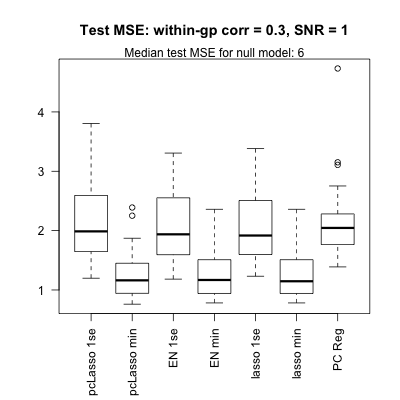}\includegraphics[width=2.2in]{6_MSE_corr0_3_SNR2.png}

\includegraphics[width=2.2in]{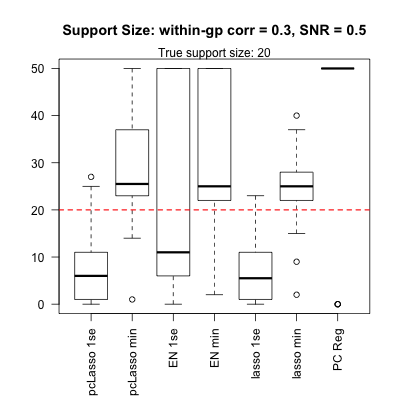}\includegraphics[width=2.2in]{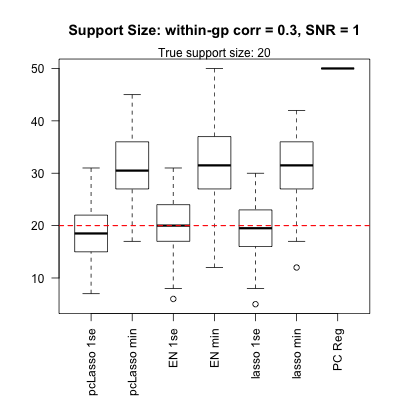}\includegraphics[width=2.2in]{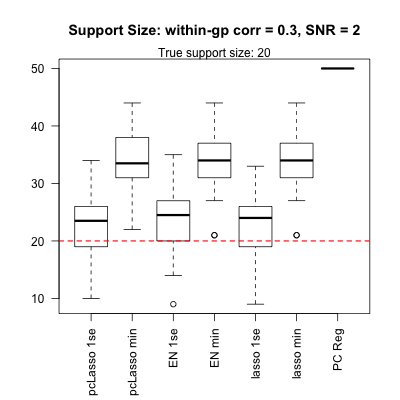}

\includegraphics[width=2.2in]{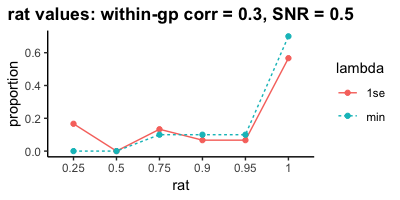}\includegraphics[width=2.2in]{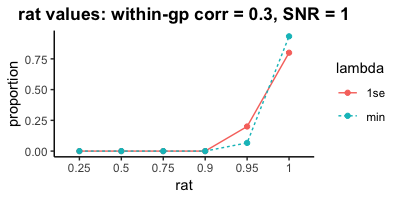}\includegraphics[width=2.2in]{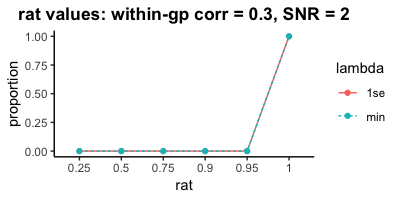}
\end{figure}

\newpage
\subsubsection{Medium $p$ with groups, ``home court" for pcLasso}
$n = 200$, $p = 200$, 10 groups of 20 uncorrelated predictors, response related to the top 2 eigenvectors in the first group.

\begin{figure}[ht]
\includegraphics[width=2.2in]{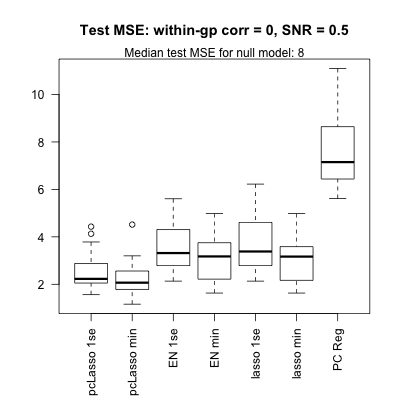}\includegraphics[width=2.2in]{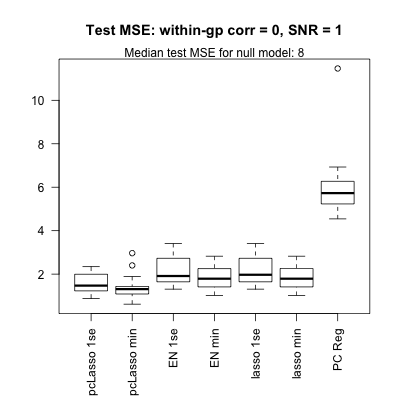}\includegraphics[width=2.2in]{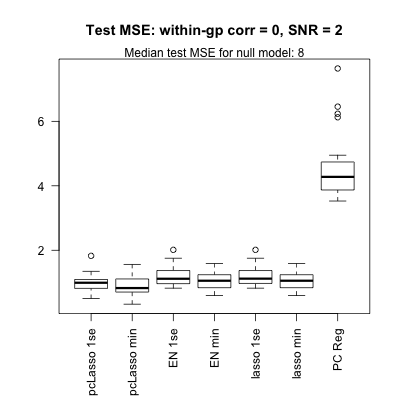}

\includegraphics[width=2.2in]{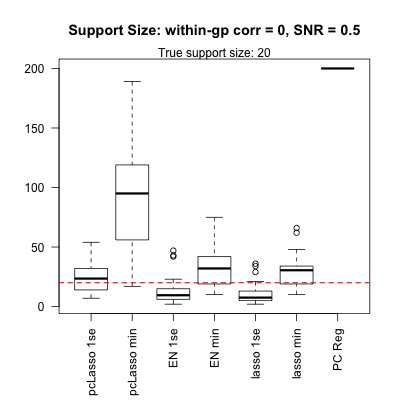}\includegraphics[width=2.2in]{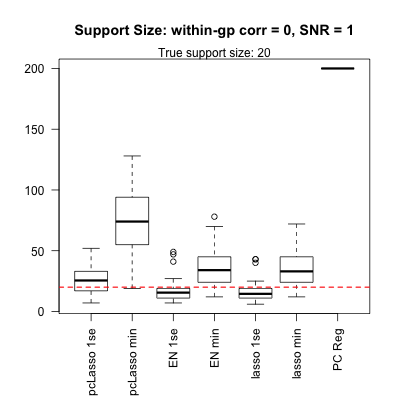}\includegraphics[width=2.2in]{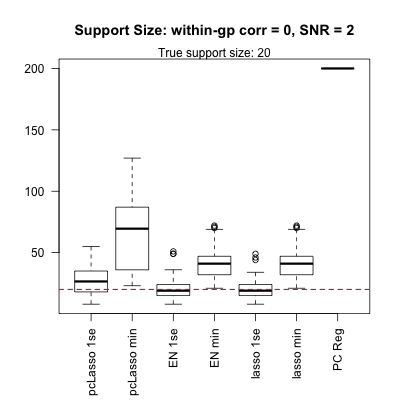}

\includegraphics[width=2.2in]{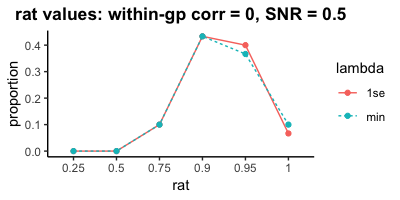}\includegraphics[width=2.2in]{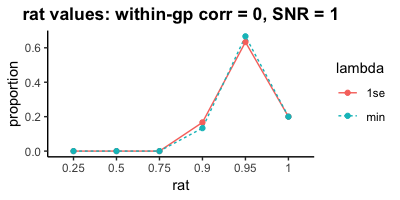}\includegraphics[width=2.2in]{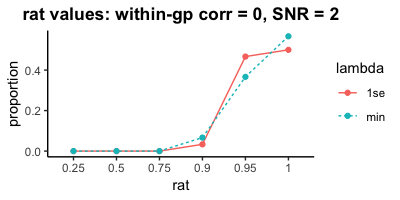}
\end{figure}

\newpage
As above, but with predictors in the same group having pairwise correlation of 0.3.

\begin{figure}[ht]
\includegraphics[width=2.2in]{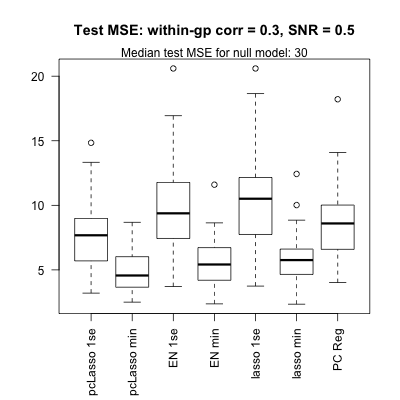}\includegraphics[width=2.2in]{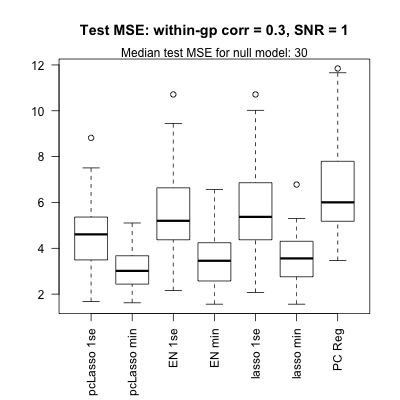}\includegraphics[width=2.2in]{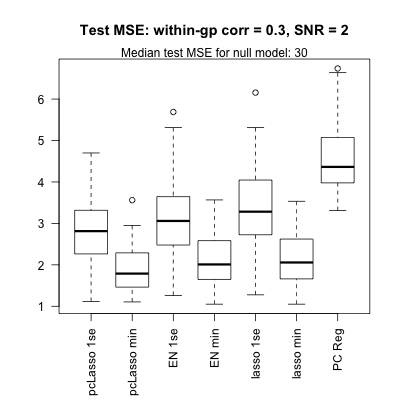}

\includegraphics[width=2.2in]{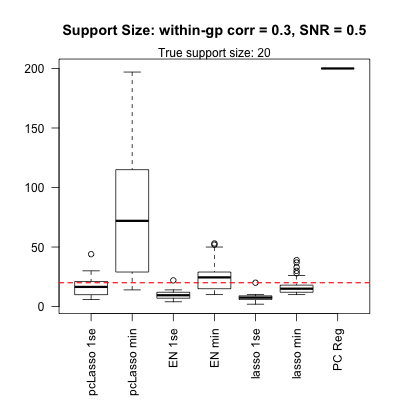}\includegraphics[width=2.2in]{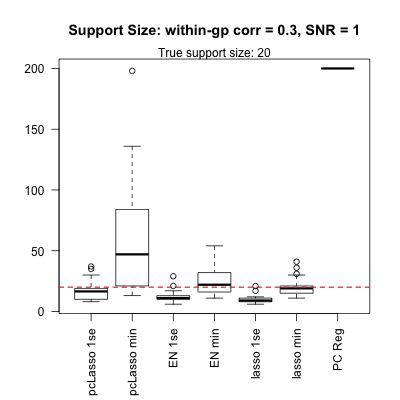}\includegraphics[width=2.2in]{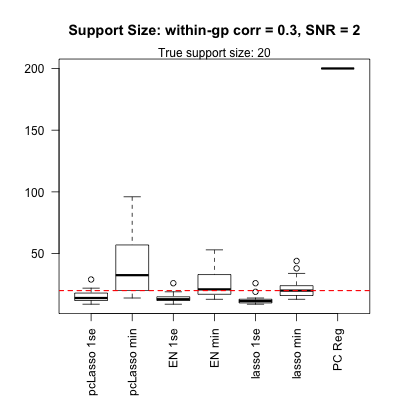}

\includegraphics[width=2.2in]{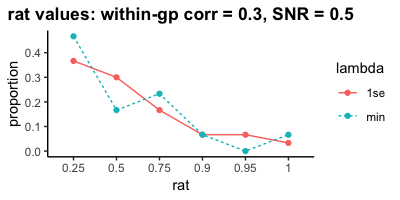}\includegraphics[width=2.2in]{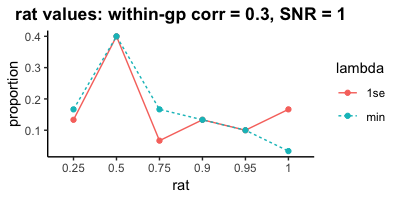}\includegraphics[width=2.2in]{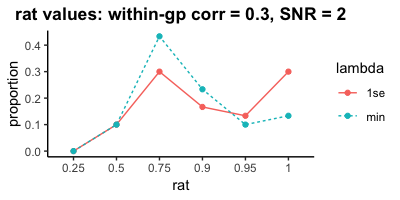}
\end{figure}

\newpage
\subsubsection{Medium $p$ with groups, ``netural court" for pcLasso}
$n = 200$, $p = 200$, 10 groups of 20 uncorrelated predictors, response related to 2 random eigenvectors in the first group.

\begin{figure}[ht]
\includegraphics[width=2.2in]{8_MSE_corr0_SNR0_5.png}\includegraphics[width=2.2in]{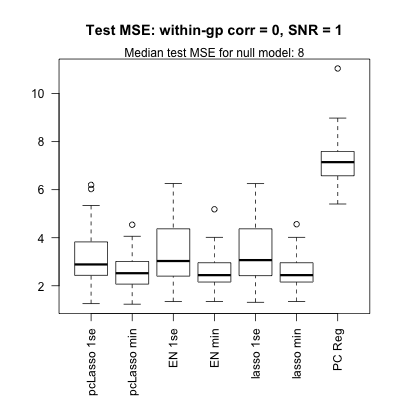}\includegraphics[width=2.2in]{8_MSE_corr0_SNR2.png}

\includegraphics[width=2.2in]{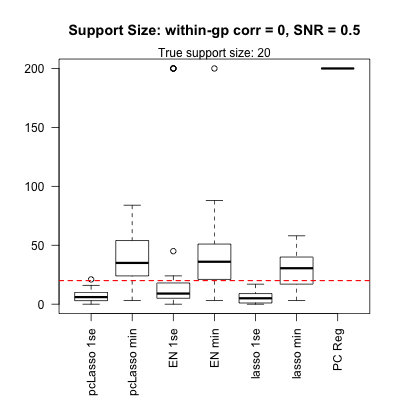}\includegraphics[width=2.2in]{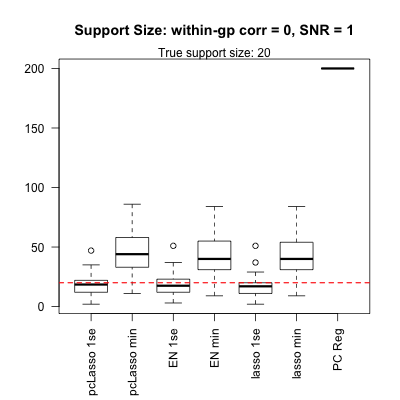}\includegraphics[width=2.2in]{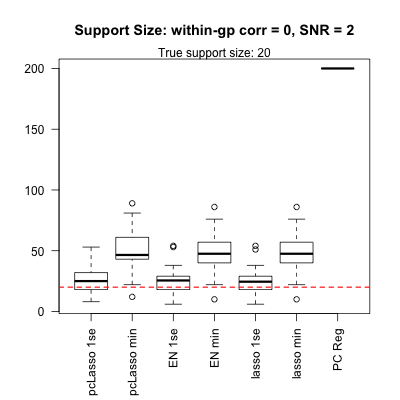}

\includegraphics[width=2.2in]{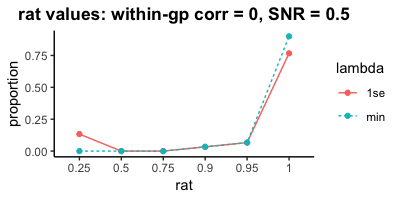}\includegraphics[width=2.2in]{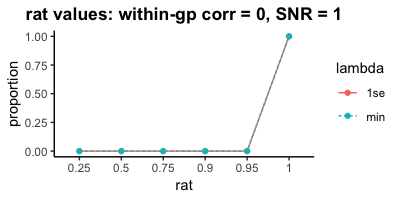}\includegraphics[width=2.2in]{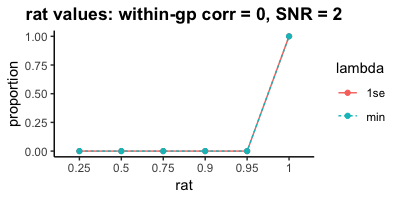}
\end{figure}

\newpage
As above, but with predictors in the same group having pairwise correlation of 0.3.

\begin{figure}[ht]
\includegraphics[width=2.2in]{8_MSE_corr0_3_SNR0_5.png}\includegraphics[width=2.2in]{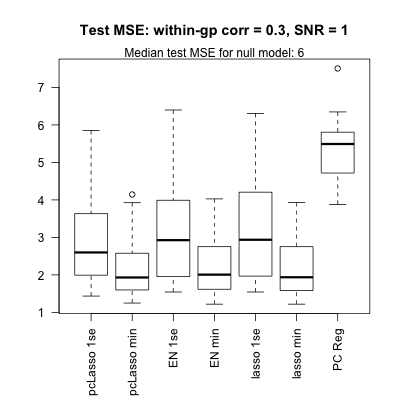}\includegraphics[width=2.2in]{8_MSE_corr0_3_SNR2.png}

\includegraphics[width=2.2in]{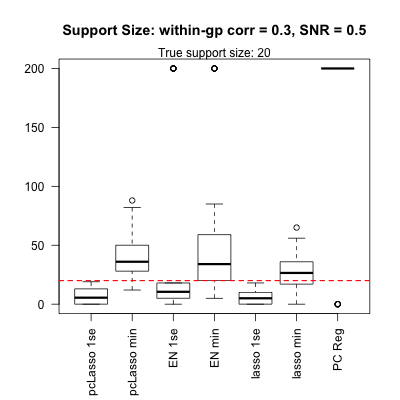}\includegraphics[width=2.2in]{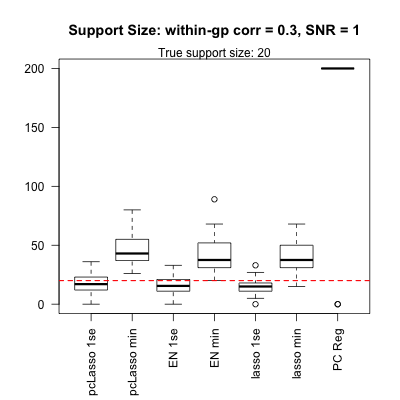}\includegraphics[width=2.2in]{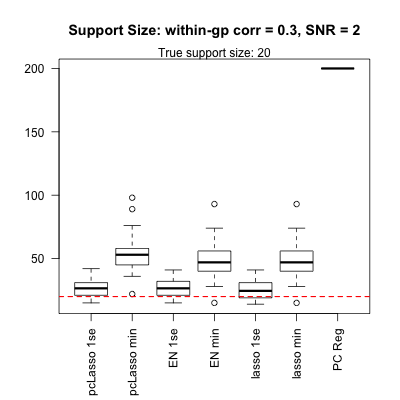}

\includegraphics[width=2.2in]{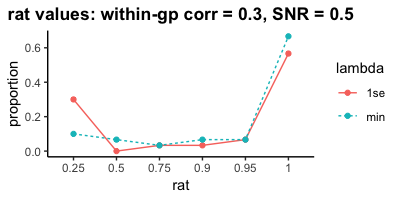}\includegraphics[width=2.2in]{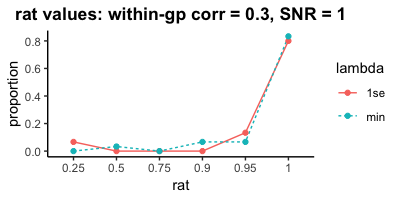}\includegraphics[width=2.2in]{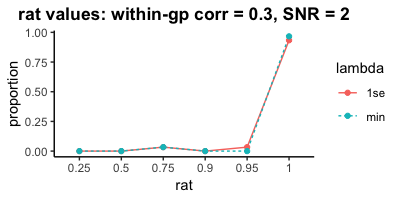}
\end{figure}

\newpage
\subsubsection{Large $p$ with groups, ``home court" for pcLasso}
$n = 200$, $p = 1000$, 10 groups of 100 uncorrelated predictors, response related to the first 2 eigenvectors in the first group.

\begin{figure}[ht]
\includegraphics[width=2.2in]{9_MSE_corr0_SNR0_5.png}\includegraphics[width=2.2in]{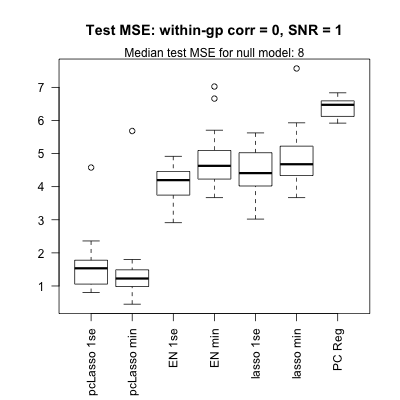}\includegraphics[width=2.2in]{9_MSE_corr0_SNR2.png}

\includegraphics[width=2.2in]{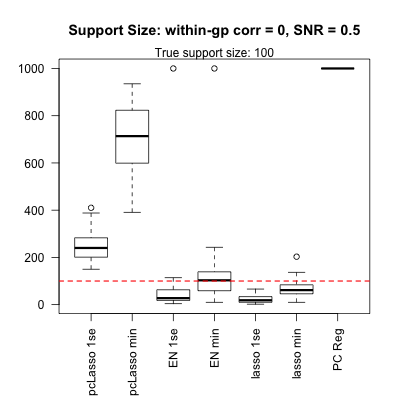}\includegraphics[width=2.2in]{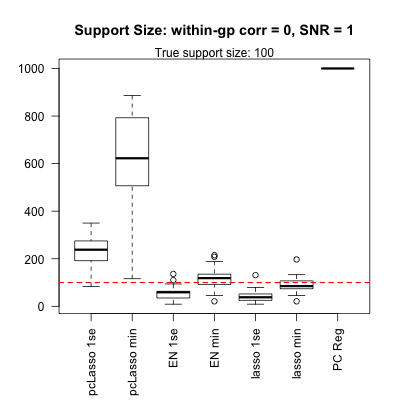}\includegraphics[width=2.2in]{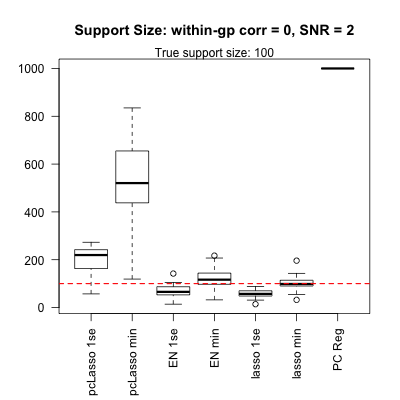}

\includegraphics[width=2.2in]{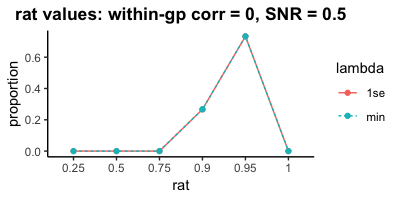}\includegraphics[width=2.2in]{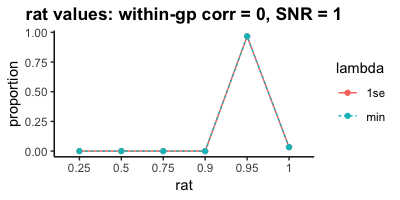}\includegraphics[width=2.2in]{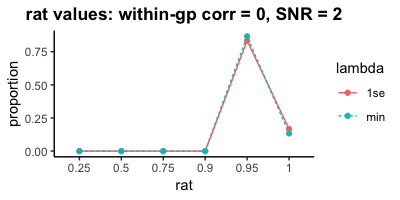}
\end{figure}

\newpage
As above, but with predictors in the same group having pairwise correlation of 0.3.

\begin{figure}[ht]
\includegraphics[width=2.2in]{9_MSE_corr0_3_SNR0_5.png}\includegraphics[width=2.2in]{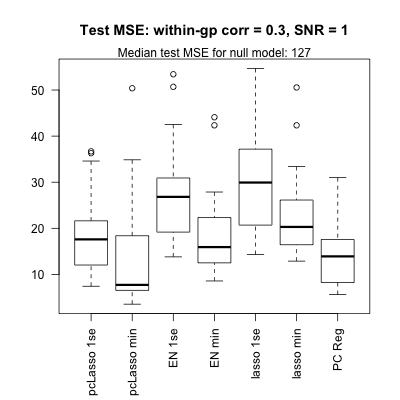}\includegraphics[width=2.2in]{9_MSE_corr0_3_SNR2.png}

\includegraphics[width=2.2in]{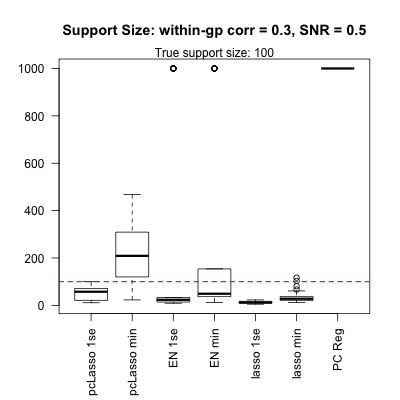}\includegraphics[width=2.2in]{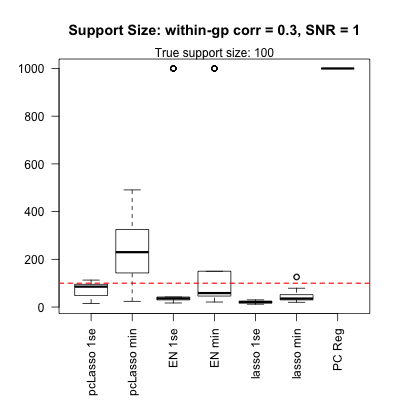}\includegraphics[width=2.2in]{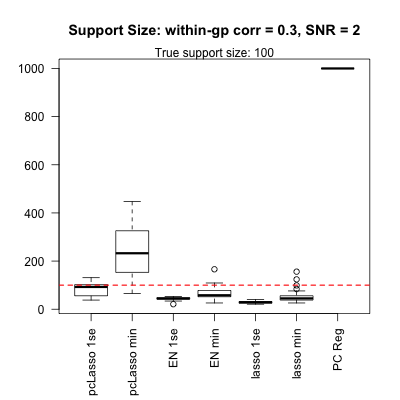}

\includegraphics[width=2.2in]{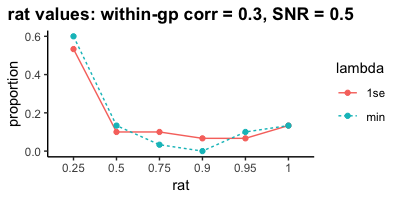}\includegraphics[width=2.2in]{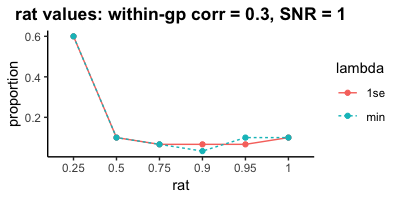}\includegraphics[width=2.2in]{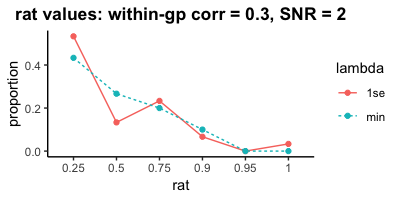}
\end{figure}

\newpage
\subsubsection{Large $p$ with groups, ``neutral court" for pcLasso}
$n = 200$, $p = 1000$, 10 groups of 100 uncorrelated predictors, response related to 2 random eigenvectors in the first group.

\begin{figure}[ht]
\includegraphics[width=2.2in]{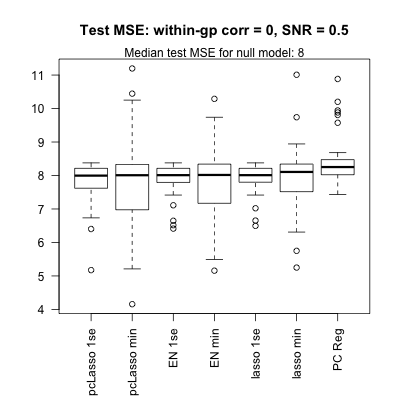}\includegraphics[width=2.2in]{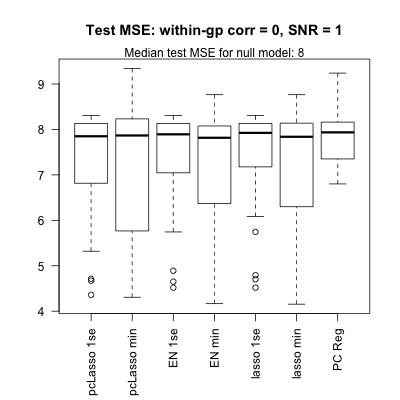}\includegraphics[width=2.2in]{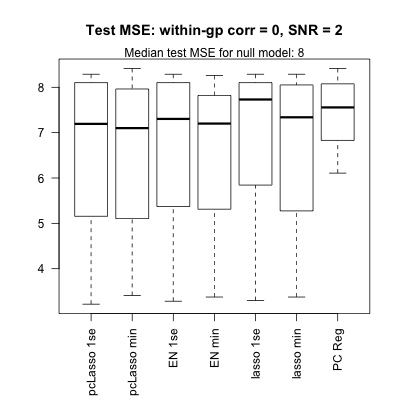}

\includegraphics[width=2.2in]{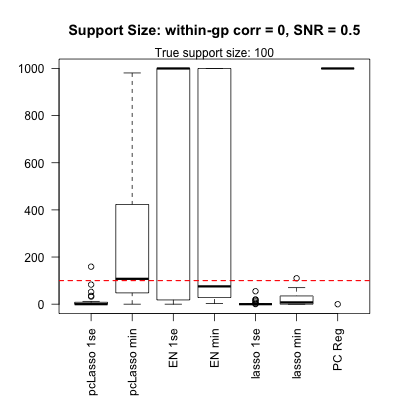}\includegraphics[width=2.2in]{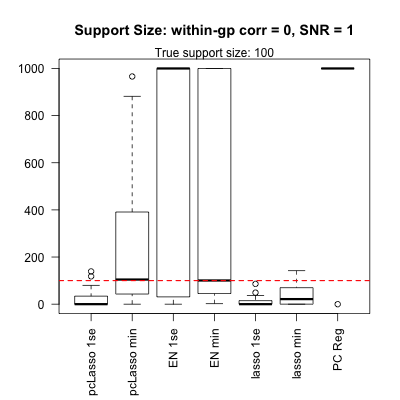}\includegraphics[width=2.2in]{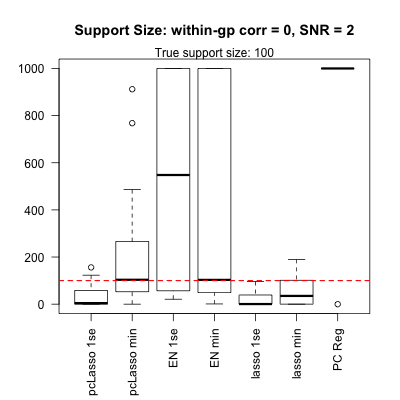}

\includegraphics[width=2.2in]{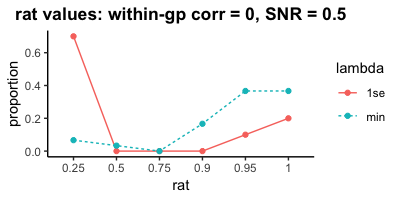}\includegraphics[width=2.2in]{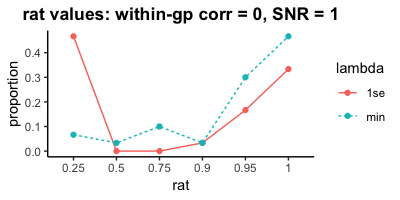}\includegraphics[width=2.2in]{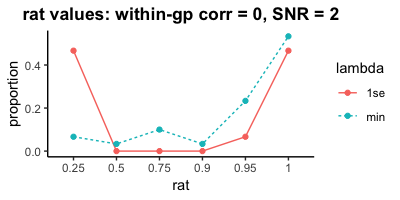}
\end{figure}

\newpage
As above, but with predictors in the same group having pairwise correlation of 0.3.

\begin{figure}[ht]
\includegraphics[width=2.2in]{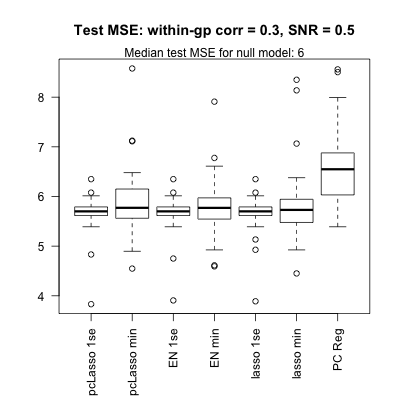}\includegraphics[width=2.2in]{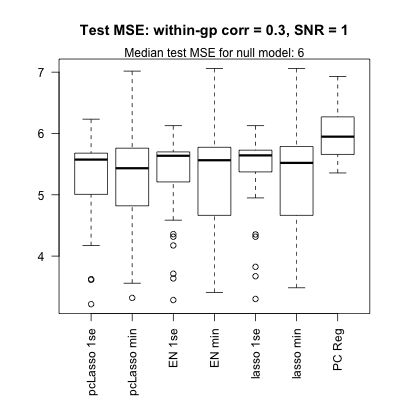}\includegraphics[width=2.2in]{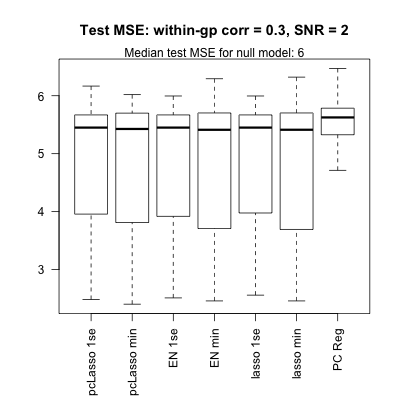}

\includegraphics[width=2.2in]{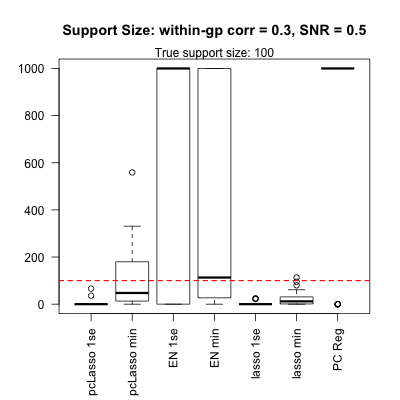}\includegraphics[width=2.2in]{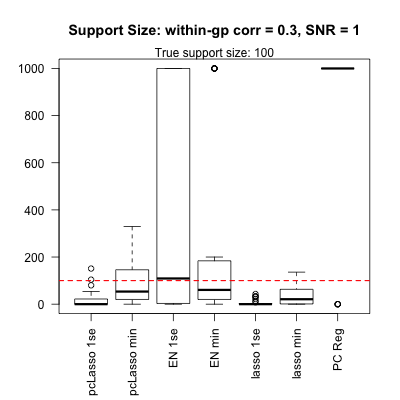}\includegraphics[width=2.2in]{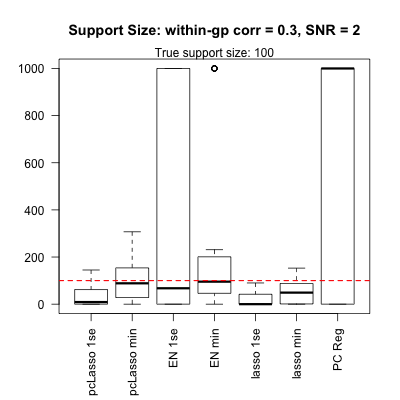}

\includegraphics[width=2.2in]{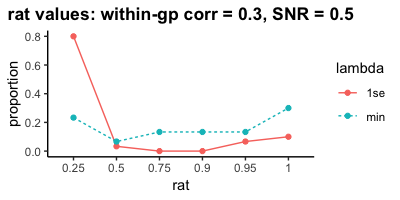}\includegraphics[width=2.2in]{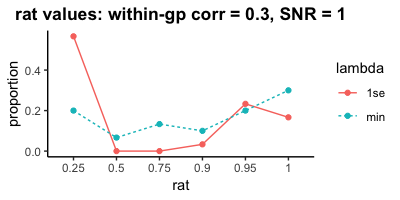}\includegraphics[width=2.2in]{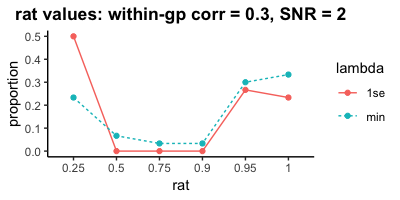}
\end{figure}

\bibliographystyle{agsm}
\bibliography{tibs}

\end{document}